\documentclass{amsart}
\usepackage[margin=1.3in]{geometry}
\usepackage{setspace}
\usepackage[T1]{fontenc}
\usepackage[utf8]{inputenc}
\usepackage{lmodern}

\usepackage[usenames, dvipsnames]{xcolor}

\usepackage{microtype}

\usepackage{amssymb}
\usepackage{amsthm}

\usepackage{mathtools}
\numberwithin{equation}{section}

\usepackage{tabularx,graphicx}
\usepackage{epstopdf}
\usepackage{makecell}
\usepackage{latexsym}
\usepackage{colortbl}
\usepackage{psfrag}
\usepackage{slashed}
\usepackage{multirow}
\usepackage{url}
\usepackage{adamsmacros}
\usepackage{spectralsequences}

\usepackage{quiver}
\usepackage{tikz-cd}
\usepackage{booktabs}

\usepackage[normalem]{ulem}

\usepackage{subcaption}
\usepackage{comment}
\usepackage{fix-cm,anyfontsize}
\usepackage{float}

\definecolor{darkblue}{rgb}{0.,0.,0.4}
\definecolor{darkred}{rgb}{0.5,0.,0.}
\definecolor{BlueViolet}{RGB}{138,43,226}
\definecolor{SkyBlue}{RGB}{30,144,255}
\definecolor{DarkGreen}{RGB}{0,100,0}

\usepackage[backref=page]{hyperref}
\hypersetup{
    colorlinks=true,
    linkcolor=darkblue,
    citecolor=blue,
    urlcolor=darkred
}

\usepackage[capitalize, noabbrev]{cleveref} 

\newcounter{mainthm}

\newtheorem{maintheorem}[mainthm]{Theorem}

\newtheorem{thm}[equation]{Theorem}
\newtheorem{conjecture}[equation]{Conjecture}
\newtheorem{lem}[equation]{Lemma}

\newtheorem{prop}[equation]{Proposition}
\newtheorem{cor}[equation]{Corollary}
\newtheorem{ques}[equation]{Question}
\newtheorem*{thm*}{Theorem}

\theoremstyle{definition}
\newtheorem{defn}[equation]{Definition}
\newtheorem{fakedefn}[equation]{``Definition''}
\newtheorem{example}[equation]{Example}
\theoremstyle{remark}
\newtheorem{rem}[equation]{Remark}
\newtheorem{Disclaimer}[equation]{Disclaimer}
\newtheorem{ansatz}[equation]{Ansatz}

\crefname{thm}{Theorem}{Theorems}
\crefname{maintheorem}{Theorem}{Theorems}
\crefname{prop}{Proposition}{Propositions}
\crefname{lem}{Lemma}{Lemmas}
\crefname{defn}{Definition}{Definitions}
\crefname{cor}{Corollary}{Corollaries}

\newcommand{\beq}{\begin{eqnarray}}
	\newcommand{\eeq}{\end{eqnarray}}

\newcommand{\bsp}{\begin{aligned}}
	\newcommand{\esp}{\end{aligned}}

\DeclarePairedDelimiter{\set}{\{}{\}}

\newcommand{\Z}{\mathbb{Z}}

\usetikzlibrary{calc}

\newcommand{\inj}{\hookrightarrow}

\usepackage[mathscr]{eucal}
\makeatletter
\def\instring#1#2{TT\fi\begingroup
  \edef\x{\endgroup\noexpand\in@{#1}{#2}}\x\ifin@}
\def\isuppercase#1{%
  \instring{#1}{ABCDEFGHIJKLMNOPQRSTUVWXYZ}%
}%
\newcommand{\C@lIfUpper}[1]{
 \if\isuppercase{#1}\mathscr{#1}%
 \else #1%
 \fi
}
\newcommand{\cat}[1]{\mathit{\@tfor\next:=#1\do{\C@lIfUpper{\next}}}}
\makeatother

\usepackage{xparse}
\newcommand{\bZ}{\mathbb Z}
\newcommand{\cZ}{\mathcal Z}

\newcommand{\C}{\mathbb C}
\newcommand{\U}{\mathrm U}
\newcommand{\Spin}{\mathrm{Spin}}
\newcommand{\MTSpin}{\mathit{MTSpin}}

\newcommand{\ko}{\mathit{ko}}

\newcommand{\Sq}{\mathrm{Sq}}
\newcommand{\abs}[1]{\lvert #1 \rvert}
\newcommand{\pt}{\mathrm{pt}}

\newcommand{\SO}{\mathrm{SO}}

\newcommand{\GL}{\mathrm{GL}}

\renewcommand{\O}{\mathrm O}
\newcommand{\id}{\mathrm{id}}

\newcommand{\CP}{\mathbb{CP}}
\newcommand{\bl}{\text{--}}
\newcommand{\R}{\mathbb R}

\newcommand{\Pin}{\mathrm{Pin}}

\newcommand{\SH}{\mathit{SH}}
\newcommand{\Vect}{\mathbf{Vect}}
\newcommand{\tVect}{\mathbf{2Vect}}
\newcommand{\tsVect}{\mathbf{2sVect}}

\newcommand{\Mod}{\mathbf{Mod}}
\newcommand{\Bimod}{\mathbf{Bimod}}

\DeclareRobustCommand*{\RaiseBoxByDepth}{%
    \raisebox{\depth}%
}
\newcommand{\uQ}{\RaiseBoxByDepth{\protect\rotatebox{180}{$Q$}}}

\newcommand{\KO}{\mathit{KO}}

\newcommand{\Sph}{\mathbb S}

\newcommand{\term}{\emph}
\newcommand{\cA}{\mathcal A}
\newcommand{\cB}{\mathcal B}

\newcommand{\fB}{\mathfrak B}
\newcommand{\fC}{\mathfrak C}

\newcommand{\SW}{\mathscr{S}\mathscr{W}}

\newcommand{\Ext}{\mathrm{Ext}}
\newcommand{\Hom}{\mathrm{Hom}}
\newcommand{\vp}{\varphi}

\newcommand{\MTSO}{\mathit{MSO}}
\newcommand{\sWitt}{s\mathscr{W}itt}

\newcommand{\sW}{s\mathcal{W}}

\newcommand{\MSpin}{\MTSpin}

\newcommand{\ku}{\mathit{ku}}

\newcommand{\sAut}{\mathscr{A}ut}
\newcommand{\sPic}{\mathscr{P}ic}
\newcommand{\sVect}{\mathbf{sVect}}

\newcommand{\alignedintertext}[1]{%
  \noalign{%
    \vskip\belowdisplayshortskip
    \vtop{\hsize=0.75\linewidth#1\par
    \expandafter}%
    \expandafter\prevdepth\the\prevdepth
  }%
}

\DeclareDocumentCommand{\shortexact}{s O{} O{} mmmm}{
\IfBooleanTF{#1}{ 
\begin{tikzcd}[ampersand replacement=\&]
	{1} \& {#4} \& {#5} \& {#6} \& {1#7}
	\arrow[from=1-1, to=1-2]
	\arrow["#2", from=1-2, to=1-3]
	\arrow["#3", from=1-3, to=1-4]
	\arrow[from=1-4, to=1-5]
\end{tikzcd}
}{ 
\begin{tikzcd}[ampersand replacement=\&]
	{0} \& {#4} \& {#5} \& {#6} \& {0#7}
	\arrow[from=1-1, to=1-2]
	\arrow["#2", from=1-2, to=1-3]
	\arrow["#3", from=1-3, to=1-4]
	\arrow[from=1-4, to=1-5]
\end{tikzcd}
}}
\DeclareDocumentCommand{\rightexact}{s O{} O{} mmmm}{
\IfBooleanTF{#1}{ 
\begin{tikzcd}[ampersand replacement=\&]
	{#4} \& {#5} \& {#6} \& {1#7}
	\arrow["#2", from=1-1, to=1-2]
	\arrow["#3", from=1-2, to=1-3]
	\arrow[from=1-3, to=1-4]
\end{tikzcd}
}{ 
\begin{tikzcd}[ampersand replacement=\&]
	{#4} \& {#5} \& {#6} \& {0#7}
	\arrow["#2", from=1-1, to=1-2]
	\arrow["#3", from=1-2, to=1-3]
	\arrow[from=1-3, to=1-4]
\end{tikzcd}
}}

\newcommand{\AdamsTower}[1]{\DoUntilOutOfBounds{
	\class[#1](\lastx, \lasty+1)
	\structline[#1]
}}

\usepackage{xspace}
\newcommand{\pinp}{pin\textsuperscript{$+$}\xspace}
\newcommand{\pinm}{pin\textsuperscript{$-$}\xspace}

\newcommand{\Aut}{\mathrm{Aut}}

\newcommand\MAILTO[1]{\href{mailto:#1}{\nolinkurl{#1}}}

\begin{document}

\title{How to Build Anomalous (3+1)-dimensional Topological Quantum Field Theories}

\author{Arun Debray}
\address{Department of Mathematics, University of Kentucky, 719 Patterson Office Tower, Lexington, KY 40506-0027}
\email{a.debray@uky.edu}

\author{Weicheng Ye}
\address{Department of Physics and Astronomy, and Stewart Blusson Quantum Matter Institute, University of British Columbia, Vancouver, BC, Canada V6T 1Z1}
\email{victoryeofphysics@gmail.com}

\author{Matthew Yu}
\address{Mathematical Institute, University of Oxford,  Woodstock Road, Oxford, UK}
\email{yumatthew70@gmail.com}

\begin{abstract}
We develop a systematic framework for constructing (3+1)-dimensional topological orders or topological quantum field theories (TQFTs)
that realize specified anomalies of finite symmetries, as encountered in gauge theories with fermions or in fermionic
lattice systems. Our approach generalizes the symmetry-extension construction to the fermionic
setting, and is grounded in recent advances in the categorical classification of anomalous TQFTs in (3+1)d. In this framework, symmetry-extension data of a \textit{supercohomology} theory are translated into a fusion 2-category, on which the anomalous TQFT is built. Building on this machinery, we demonstrate explicit calculations for various symmetry groups and their associated anomalies, with the help of a \textit{hastened Adams spectral sequence} for computing supercohomology groups which we will detail in a planned sequel. Finally, we prove that all supercohomology anomalies can be realized by fermionic topological orders, whereas beyond-supercohomology anomalies cannot, resolving a question of Córdova--Ohmori for fermionic (3+1)d systems with finite symmetries.
\end{abstract}

\maketitle

\tableofcontents

\section{Introduction}
A central question in high-energy and condensed matter physics is whether a given phase can be realized as the infrared (IR) description of some ultraviolet (UV) theory, where the UV theory may be either a weakly coupled gauge theory or a lattice model. In particular, when a UV theory carries a specified anomaly, the corresponding IR dynamics are forbidden to be \textit{trivial} and difficult to analyze directly. Anomaly matching has proven to be a powerful tool in addressing this question. In (2+1)d, anomaly matching for ordinary symmetries and categorical symmetries, has played a central role in mapping out phase diagrams \cite{Gomis:2017ixy,Cordova:2017vab,Cordova:2017kue,Choi:2018tuh,2021PhRvX..11c1043Z,Antinucci:2025fjp}, conjecturing dualities among distinct UV gauge theories with matter \cite{Cordova:2018qvg,2016AnPhy.374..395S,2016JHEP...09..095H,Turner:2019wnh}, and providing a complete classification of the possible lattice realizations of a given UV theory \cite{2021PhRvX..11c1043Z,2022ScPP...13...66Y,Kobayashi:2024dqj,2024PhRvX..14b1053Y,2025ScPP...18..161L,Feng:2025qgg}. This success naturally motivates extending the perspective to higher-dimensional theories. 

In this paper we focus on symmetries of theories in (3+1)d. As a concrete example of a theory that exhibits some of the symmetries we consider, consider a gauge theory with $N_f$ left-handed chiral fermions transforming in some representation $R$ of the gauge group. Classically, there is a chiral $\U(1)$ symmetry rotating the fermions, but the ABJ anomaly reduces this symmetry to
\begin{equation}
    \U(1) \rightarrow \Z/({2 N_f \cdot T(R)})
\end{equation}
where $T(R)$ is the Dynkin index of the representation. Given the strongly-coupled nature of such theories, it is natural to ask whether they flow to a nontrivial TQFT in the IR. Questions of this kind frequently arise in the study of the dynamics of $(3+1)$d gauge theories \cite{Hsi18,CO2} and, notably, in understanding aspects of the Standard Model \cite{Wang:2020mra,2020arXiv200616996W,Wang:2025oow,Cheng:2024awi}.

Assuming that the symmetry remains unbroken in the IR, which can be justified either numerically or through theoretical results such as the Vafa--Witten theorem \cite{PhysRevLett.53.535,Vafa:1983tf}, there is a purely mathematical way to investigate the question in the previous paragraph: letting $\alpha$ be the anomaly of the gauge theory, which we model as a reflection-positive invertible field theory (IFT). If the gauge theory flows to a TQFT with unbroken symmetry in the IR, the anomaly of that TQFT must be deformation-equivalent to $\alpha$. So a good first question is, \emph{is} there a TQFT with anomaly deformation-equivalent to $\alpha$?

The \term{symmetry extension procedure} is an approach to constructing candidate IR symmetry-enriched topological orders (SETs) or topological quantum field theories (TQFTs)\footnote{Categorically, fermionic TQFTs in (3+1)d are classified by certain fusion 2-categories, whereas a fermionic topological order is classified by fusion 2-categories with extra structure including unitarity. Since a mathematical understanding of these extra structures is still underway (see \cite{ferrer2024dagger,MS23,SS24}), we will mostly use the terminology TQFT. Nevertheless, the finite gauge theories we construct using the method of symmetry extension are all fermionic topological orders. Moreover, we need the extra assumption in \Cref{mainthmB} that the TQFT we analyzed assigns a 1-dimensional Hilbert space on $S^3$, which is supposed to be satisfied by topological orders.} realizing a given anomaly $\alpha$ through explicit constructions realized by gauge theories. This procedure was first studied in special cases by Kapustin--Thorngren~\cite{2014arXiv1404.3230K,2014PhRvL.112w1602K} and Thorngren--von Keyserlingk~\cite{Thorngren:2015gtw}, then in generality by Wang--Wen--Witten~\cite{Wang:2017loc} and Tachikawa~\cite{Tachikawa:2017gyf} (see also \cite{Witten:2016cio}) for reflection-positive IFTs whose partition functions are integrals of group cohomology classes. For applications to fermionic theories,\footnote{In this paper, a fermionic theory refers to a field theory -- topological or non-topological -- whose definition requires a \textit{twisted} spin structure. This includes the trivial twist, for which one gets a spin structure in the usual sense.} one would like to generalize the symmetry extension procedure to a broader class of anomalies, and there are some recent works aimed at constructing fermionic topological orders with a given anomaly in $(2+1)$d, including \cite{KOT20,Thorngren:2020aph,2021PhRvR...3a3056C,Loo:2025}, but a higher-dimensional version of the symmetry extension procedure for fermionic theories is completely open.
In this paper, we focus on (3+1)d fermionic theories and ask: 
\begin{ques}\label{question:main}
    How can one construct a (3+1)d fermionic topological order or TQFT that saturates a given anomaly associated with some finite symmetry?
\end{ques}

In this work, we incorporate the categorical formulation of (3+1)d TQFTs into the symmetry extension procedure to answer Question \ref{question:main} in full generality. The outline of the symmetry extension procedure is the same as in the prior work mentioned above, except replacing ordinary cohomology with a generalized cohomology theory called \term{(extended) supercohomology}~\cite{Wang:2017moj, KT17, PhysRevX.10.031055}, though the details involved in making this work are nontrivial. This data is then fed into the machinery of~\cite{DY2025}, and we are thus able to construct a candidate (3+1)d fermionic TQFT with the given anomaly.

Not every (4+1)d reflection-positive IFT of twisted spin manifolds can be written using supercohomology, and our symmetry extension procedure does not apply to these ``beyond-supercohomology'' IFTs. It is thus natural to wonder whether they could be incorporated into a more general symmetry extension procedure. However, we prove that this is impossible: a (4+1)d twisted spin reflection-positive IFT $\alpha$, where the twist is associated to a finite group, admits a topological order as a boundary theory if and only if the partition function of $\alpha$ equals the integral of a supercohomology class.  This nonexistence result causes \textit{symmetry-enforced gaplessness}. This generalizes work of Córdova--Ohmori~\cite{CO2}, who restricted to even-order cyclic groups, and answers a question they raise in \textit{loc.\ cit.} (see also Brennan~\cite{Bre23}).

\subsection{Fermionic symmetries and supercohomology anomalies}\label{subsection:fermionicsupercoh}
Notably, the classification of (3+1)d TQFTs uses the data of \textit{supercohomology},\footnote{Confusingly, there are two closely related generalized cohomology theories called ``supercohomology:'' the \term{restricted supercohomology} of~\cite{Fre08, PhysRevB.90.115141}, and the \term{extended supercohomology} of~\cite{Wang:2017moj, KT17, PhysRevX.10.031055}. In this paper we will exclusively use $\SH$ to denote the latter, and use $\mathit{rSH}$ for the former.}, which we denote as $\SH^5(BG,s,\omega)$, whose precise meaning will be clear later. Before we state the main results of our paper, we first review mathematical formulations of fermionic symmetries and supercohomology, and quickly comment on the reason for the natural appearance of supercohomology in our context.

A fermionic symmetry~\cite[\S 7]{Ben88} is given by a symmetry group $G_f$,\footnote{Unless stated otherwise, all symmetry groups in this paper are finite groups.} and two additional pieces of data: (1) a map $\rho\colon G_f\rightarrow \Z/2$ such that the symmetry element is antiunitary or unitary if the image under $\rho$ is $1$ or $0$, respectively, and (2) a central $\Z/2$ subgroup $\langle (-1)^F\rangle \subset G_f$ in the kernel of $\rho$ generated by fermion parity. This motivates describing the fermionic symmetry using the following three pieces of data: 
\begin{itemize}
\item a (bosonic) symmetry group $G\coloneqq G_f/\langle(-1)^F\rangle$;
\item a class $s \in H^1(BG;\Z/2)$, corresponding to $\rho$;
\item a class $\omega \in H^2(BG; \Z/2)$, classifying the extension of $G$ by the $\Z/2$ subgroup $\langle(-1)^F\rangle$ to get $G_f$.
\end{itemize}

The corresponding supercohomology $\SH^n(BG,s,\omega)$ is a generalized cohomology theory first proposed in \cite{Wang:2017moj, KT17} for classifying reflection-positive fermionic invertible field theories (IFTs) or fermionic symmetry-protected topological orders (SPTs)\footnote{We will use (reflection-positive) fermionic IFTs and fermionic SPTs interchangeably in this paper.} in a restricted setting. In~\cite{Wang:2017moj}, supercohomology is defined as the cohomology of an explicit cochain complex: the $n$-cochains are triples $(a, b, c)$ as follows:
  \begin{itemize}
      \item a cochain $a  \in C^{n-2}(BG;\Z/2)$, called the Majorana layer.
       \item a cochain $b \in C^{n-1}(BG;\Z/2)$, called the Gu--Wen layer.
       \item a cochain $c \in C^n(BG;\mathbb{C}^\times)$, called the Dijkgraaf--Witten layer.
  \end{itemize}
The differential mixes together information from different layers, so that cocycles satisfy certain equations relating $a$, $b$, and $c$. These equations were derived in \cite{Wang:2017moj, KT17}, and we review them in Appendix \ref{subsection:twistedSH}. There we also discuss twisted supercohomology, introduced by~\cite{PhysRevX.10.031055}, which incorporates the data of $(s, \omega)$ into those equations. 

For any fermionic symmetry $(G, s, \omega)$, it is possible to choose a set of generators of $\SH^n(BG,s,\omega)$, such that each generator has a cocycle representative with exactly one of $a$, $b$, or $c$ nonzero. Accordingly, we will say that the generator is in the Majorana ($a \ne 0$), Gu--Wen ($b\ne 0$), or Dijkgraaf--Witten ($c\ne 0$) layer as part of our descriptions of supercohomology groups in the main results section.

We can compare supercohomology with the full classification of \textit{'t Hooft anomalies} of continuum field theories. The 't Hooft anomaly is usually defined as the ``obstructions to gauging'' for $G$-symmetry and classified by the cobordism group $I_{\Z}\MSpin^{n+1}(BG,s,\omega)$ \cite{FH}, the Anderson dual of (twisted) spin bordism\footnote{Because we are interested in RG invariants and anomaly matching, we study anomalies as \emph{deformation} classes of reflection-positive invertible field theories (IFTs) with (twisted) spin structure. Freed--Hopkins~\cite[\S 5.4]{FH} and Grady~\cite{Gra23} show that deformation classes of these IFTs are classified by how that deformation classes are described by the Anderson dual $\Sigma I_\Z(\text{--})$ of spin bordism. However, classifications of fusion $2$-categories most naturally use the Pontryagin dual $I_{\C^\times}(\text{--})$, e.g.\ in~\cite{JFY2,JFR,D11,DY23a,Decoppet:2024htz, Teixeira:2025qsg}, and so supercohomology is built using $I_{\C^\times}$ as well. Though the distinction between $I_{\C^\times}$ and $\Sigma I_\Z$ is conceptually important in general, it does not come into play in this paper: for anomalies of finite-group symmetries of $4$-dimensional theories, the Pontryagin-to-Anderson map is an isomorphism. We will therefore not dwell on this difference.}
, which we also denote as $\mho_\Spin^n(BG,s,\omega)$. 
The missing piece is the so-called $p+ip$ layer represented by a cochain in $C^{n-3}(BG;\Z)$. Moreover, there is a natural map from supercohomology to the classification of 't Hooft anomalies,
\begin{equation}\label{eq:supercohomologycomparison}
   SH^n(BG, s, \omega) \rightarrow \mho^{n+1}(BG, s, \omega).
\end{equation} 
Thus, we can think of supercohomology as an ``approximation'' of the classification of fermionic SPTs without taking into account the effect of the $p+ip$ layer. Even more interestingly, a particular anomaly with a nontrivial element in this layer is identified as an obstruction to the construction of fermionic topological orders saturating the given anomaly in recent works \cite{CO2}. This piece also does not appear explicitly in the higher-categorical framework of (3+1)d TQFTs \cite{Decoppet:2024htz,DY2025}. We will formulate these observations into a theorem in the next section, showing that if an SPT valued in $\mho_\Spin^5(BG,s,\omega)$ has a contribution from the $p+ip$ layer, then there is no (3+1)d fermionic topological order which can live at its boundary without breaking the symmetry. In effect, the boundary for this SPT is gapless.

Motivated by this result, we decompose elements in the classification of 't Hooft anomalies into two classes: those that lie in the image of \cref{eq:supercohomologycomparison},
which we refer to as \emph{supercohomology anomalies} and label them just as elements in $SH^n(BG, s, \omega)$, and those that lie outside the image, which we call \emph{beyond-supercohomology anomalies}. These supercohomology anomalies valued in $SH^n(BG, s, \omega)$ serve as the starting point of our construction.\footnote{The map \cref{eq:supercohomologycomparison} may not be injective, meaning that certain elements in supercohomology may map to trivial element in the classification of 't Hooft anomalies. We will only consider the elements in supercohomology that remain nontrivial under \cref{eq:supercohomologycomparison}.}

\subsection{Main Results}\label{subsection:results}
In this subsection, we summarize our construction of anomalous (3+1)d fermionic TQFTs and examine several symmetry groups central to physical applications. A companion paper \cite{DYY3} provides a concise summary of our primary results tailored for a physics audience, while further exploring their broader physical implications.

As discussed in \cref{subsection:fermionicsupercoh}, we start with a fermionic symmetry, written as $(BG, s, \omega)$, and construct a (3+1)d topological order that saturates a particular anomaly valued in supercohomology $\SH^5(BG, s, \omega)$. Our construction involves the following steps. We find a group $H$, with $p\colon H \twoheadrightarrow G$, such that the generator for the group $\SH^5(BG,s,\omega)$ trivializes when pulled back to $\SH^5(BH,s',\omega')$, where $s' = p^*(s)$ and $\omega' = p^*(\omega)$. The trivialization gives a torsor over $\SH^4(BH,s',\omega')$, which, as we will review in \cref{subsection:TO}, gives rise to a (3+1)d  fermionic $G$-SET. Let $K \hookrightarrow H$ be the kernel of $p$:
\begin{equation}\label{eq:extension}
    \begin{tikzcd}
        1 \arrow[r]& K \arrow[r]& H \arrow[r] & G \arrow[r]& 1\,;
    \end{tikzcd}
\end{equation}
then by gauging $K$ we obtain a $K$-gauge theory, with the Lagrangian description given by a class in $\SH^4(BK)$, equipped with a $G$-symmetry and anomaly $\SH^5(BG,s,\omega)$. 

We note that this makes the construction parallel to the construction in the bosonic settings, with regular cohomology replaced by supercohomology, and many theorems/results in the bosonic setting will reappear in the fermionic setting. In particular, for bosonic anomalies, it was shown in \cite{Tachikawa:2017gyf} that it is always possible to trivialize an anomaly of a finite group $G$ valued in group cohomology via a finite number of group extensions. We show that an analogous result is also true for fermionic anomalies. For any supercohomology reflection-positive IFT, we prove the following theorem :
\begin{maintheorem}[\Cref{thm:always_triv}]\label{mainthmA}
For any fermionic symmetry labeled by $(BG, s, \omega)$ with finite $G$ and any supercohomology class $\alpha\in SH^n(BG,s,\omega)$ with $n\geq 5$, there is an algorithmic construction of a finite group $\widetilde G$ and a map $\rho\colon \widetilde G\to G$ such that \begin{equation}
    \rho^*(\alpha) = 0\in \SH^n(B\widetilde G, \rho^*(s), \rho^*(\omega)).
\end{equation}
\end{maintheorem}
We prove this theorem in \S\ref{section:always_works} by giving an algorithmic way of trivializing $\alpha$. Physically this has the following implication for constructing anomalous TQFTs.

\begin{cor}[\Cref{existence_and_nonexistence} part \eqref{SH:yes}]\label{cor:maincor}
    Let $\alpha\in  \mho_\Spin^5(BG, 0, \omega)$ be a reflection-positive invertible TFT. If $\alpha$ is in the image of the map $\SH^5(BG, 0, \omega)\to\mho_\Spin^5(BG, 0, \omega)$, then there is an algorithmic construction of a 4d TQFT $Z$ of $(BG, 0, \omega)$-twisted spin manifolds such that the deformation class of the anomaly of $Z$ equals $\alpha$.
\end{cor}

Moreover, for any beyond-supercohomology anomaly, meaning there is a contribution from the $p+ip$ layer, we have the following result.
\begin{maintheorem}[\Cref{existence_and_nonexistence} part \eqref{SH:no}]\label{mainthmB}
Let $\alpha\in  \mho_\Spin^5(BG, s, \omega)$ be an invertible TFT, if $\alpha$ is \emph{not} in the image of the map $\SH^5(BG, s, \omega)\to\mho_\Spin^5(BG, s, \omega)$, then there is no (3+1)d fermionic topological order whose anomaly is equal to $\alpha$.
\end{maintheorem}

Now we turn to concrete examples. For each example, we classify all possible anomalies valued in supercohomology and then provide the data of an extension $H$ required to construct the corresponding $(3+1)$d topological orders that saturate these anomalies. To compute the relevant supercohomology groups, we develop in \cite{DYY2} a \emph{hastened Adams spectral sequence} (HASS) for supercohomology and certain related generalized cohomology theories, which plays the role of the Adams spectral sequence in the computation of spin cobordism groups. This new method, combined with the more elementary \emph{Atiyah–Hirzebruch spectral sequence} (AHSS), renders our calculations tractable. We also make heavy use of the \textit{Smith long exact sequence} developed in \cref{Smith_LES,LES_is_Smith} to identify the image under the pullback. The results are summarized in \Cref{tab:results}. 

\begin{Disclaimer}
    Our objective is to construct topological quantum field theories that saturate an anomaly that a given UV theory may possess. We do not, however, assert that the theories obtained in this way necessarily arise as the IR limit of the UV theory in question. In particular, there may exist physical mechanisms -- beyond the scope of our present analysis -- that drive the IR dynamics to a gapless phase. Moreover, alternative choices in the construction we present could lead to distinct yet equally reasonable TQFTs saturating the same anomaly. In particular, we do not explicitly verify that our construction is minimal in the sense that  $|H|/|G|$ is smallest possible, although we have made effort to present $H$ in its simplest form.
\end{Disclaimer}

\begin{example}\label{ex:cyclic}
    We consider fermionic theories with a $G=\Z/n$ symmetry with no twist, i.e., both $s$ and $\omega$ are trivial, with the corresponding supercohomology group that classifies the anomalies given as follows.
    \begin{itemize}
        \item If $n$ is odd, the map from $\C^\times$-cohomology to supercohomology is an isomorphism, so $\SH^5(BG) \cong \Z/n$.
        \item If $n = 2$, by \Cref{prop:Z2notwist}, $\SH^5(B\Z/2) = 0$.
        \item If $n=2^k$ and $k\geq 2$, by \Cref{prop:2knotwist}, $\SH^5(BG) \cong \Z/2^{k-1}$. 
    \end{itemize}
    
For $n$ odd and $n=2^k,k\geq 2$, where the corresponding supercohomology group is nontrivial, we write down $H$, such that the generator for each supercohomology group is trivialized when pulled back to $\SH^5(BH)$. 
    
    \begin{thm}\label{thm:Znsymmetry}
    \hfill 
   \begin{enumerate}
       \item When $n=p^k$ with $p$ an odd prime, the short exact sequence
       \begin{equation}\label{eq:subZ2}
         1 \longrightarrow \Z/p\longrightarrow \Z/p^{k+1} \xrightarrow{p} \Z/p^k\longrightarrow 1,
    \end{equation}
    induces a trivial pullback $p^*\colon\SH^5(B\Z/p^k)\to\SH^5(B\Z/p^{k+1})$.
    \item\label{2k_untwist} For $n=2^k$ and $m =\lceil \frac{k-1}{2} \rceil$, the short exact sequence
    \begin{equation}\label{eq:subZm}
         1 \longrightarrow \Z/2^m\longrightarrow \Z/2^{k+m} \xrightarrow{p} \Z/2^k \longrightarrow 1,
    \end{equation}
    induces a trivial pullback $p^*\colon\SH^5(B\Z/2^k)\to \SH^5(B\Z/2^{k+m})$.
   \end{enumerate}
    \end{thm}
    According to the general procedure, by gauging the $\Z/n$-subgroup of $\Z/n^2$ for $n$ odd, or the $\Z/2^m$-subgroup of $\Z/2^{k+m}$, we see that: 
\begin{cor}\label{cor:Z2SET}
    For cyclic groups $G=\Z/n$ with $n$ odd or $n=2^k$, any reflection-positive IFT given by a class in $\SH^5(BG)$ can be realized as the anomaly of a (3+1)d gauge theory by gauging the $\Z/n$, resp.\ $\Z/2^m$ subgroups of a $\Z/n^2$, resp.\ $\Z/2^{k+m}$ symmetric state as in~\cref{eq:subZ2,eq:subZm}, where $m=\lceil \frac{k-1}{2}\rceil$.
\end{cor}

Other cyclic groups can be understood by localizing to all the prime factors, as discussed in \S\ref{subsub:cyclic}.

\end{example}

\begin{example}\label{ex:cyclicfermion}
     We now consider $G=\Z/n$ unitary symmetry with $n$ even, $s=0$ but $\omega$ nontrivial. Namely, the corresponding fermionic symmetry algebra is given by:
     \begin{equation}
         g^n = (-1)^F\,,
     \end{equation}
     where $g$ is the generator for $G$, and $(-1)^F$ is fermion parity.
     \begin{itemize}
         \item If $n=2$, by \Cref{prop:Z2x2twist}, $\SH^5(B\Z/2,0,x^2) \cong \Z/8$, where $x\in H^1(B\Z/2;\Z/2)$, and the generator of this $\Z/8$ resides in the Majorana layer.  
         \item If $n=2^k$ and $k\geq 2$, by \cref{prop:2kytwist}, $\SH^5(B\Z/2^k,0,y) \cong \Z/2 \oplus \Z/2^{k+1}$, where $y\in H^2(B\Z/2^k;\Z/2)$. This isomorphism may be chosen so that a generator of $\Z/2^{k+1}$ is in the Gu--Wen layer, and the generator for $\Z/2$ is in the Majorana layer.
     \end{itemize}
\end{example}

\begin{thm}\label{thm:trivializedGen}
\hfill
\begin{enumerate}
    \item \label{exp2_2}
For $n=2$, the short exact sequence
\begin{equation}\label{eq:seq3}
            1 \longrightarrow \Z/4\longrightarrow \Z/8 \xrightarrow{p} \Z/2\longrightarrow 1,
        \end{equation}
        induces a trivial pullback $p^*\colon\SH^5(B\Z/2, 0, x^2)\to\SH^5(B\Z/8, 0, 0)$.
       \item\label{its_m} For $n=2^k,k\geq 2$ and $m = \lceil \frac{k}{2} \rceil$, the short exact sequence
       \begin{equation}\label{eq:seq3firstcase}
         1 \longrightarrow \Z/2^{m}\longrightarrow \Z/2^{k+m} \xrightarrow{p} \Z/2^k \longrightarrow 1,
    \end{equation}
        induces the pullback $p^*\colon\SH^5(B\Z/2^k,0,y)\to\SH^5(B\Z/2^{k+m}, 0, 0)$ such that its action on the generator of $\Z/2^{k+1}$ in the Gu-Wen layer is zero. The short exact sequence
       \begin{equation}\label{eq:seq3secondcase}
         1 \longrightarrow \Z/4 \longrightarrow \Z/2^{k+2} \xrightarrow{p} \Z/2^k \longrightarrow 1,
    \end{equation}
        induces the pullback $p^*\colon\SH^5(B\Z/2^k,0,y)\to\SH^5(B\Z/2^{k+2}, 0, 0)$ such that its action on the generator of $\Z/2$ in the Majorana layer is zero. 
       \end{enumerate}
\end{thm}

\begin{cor}\label{x2_TFTs}
    For cyclic groups $G=\Z/2$,  any reflection-positive IFT given by a class in $\SH^5(B\Z/2, 0, x^2)$ can be realized as the anomaly of a (3+1)d gauge theory (up to deformation equivalence) by gauging the $\Z/4$ subgroup of a $\Z/8$ symmetric state as in~\cref{eq:seq3}.
\end{cor}

\begin{cor}\label{y_TFTs}
    For cyclic groups $G=\Z/2^k,k\geq 2$, the deformation class of reflection-positive IFTs corresponding to any of the classes in $\SH^5(B\Z/2^k, 0, y)$ generating the $\Z/2^{k+1}$ factor can be realized as the anomaly of a (3+1)d gauge theory by gauging the $\Z/2^{m}$ subgroup of a $\Z/2^{k+m}$ symmetric state as in~\cref{eq:seq3firstcase}, where $m=\lceil \frac{k}{2} \rceil$. The IFT corresponding to the class in $\SH^5(B\Z/2^k, 0, y)$ generating the $\Z/2$ factor can be realized as the anomaly of a (3+1)d gauge theory by gauging the $\Z/4$ subgroup of a $\Z/2^{k+2}$ symmetric state as in~\cref{eq:seq3secondcase}.
\end{cor}

Finally, we also would like to extend our constructions to fermionic symmetries involving time-reversal symmetries, or nontrivial $s$-twist. Unfortunately, the necessary fusion 2-categorical framework for such an extension is not yet fully developed. Nevertheless, it is straightforward to generalize the construction at a formal level: one can still write down an extension of the form \cref{eq:extension} that trivializes the given anomaly upon pullback, and we conjecture that the resulting $K$-gauge theory is the desired TQFT that realizes the specified anomaly.

 \begin{example}\label{ex:timerev}
     We give one example in \Cref{appendix:timerev} that involves time-reversal, fermion parity, and chiral symmetry. Let $T$ be the generator of a time-reversal symmetry, and $g$ be the generator of a unitary symmetry. We consider the following symmetry algebra, where $k\ge 2$:
\begin{equation}\label{eq:GandT}
    g^{2^k}=T^2 = (-1)^F.
\end{equation}
This corresponds to the twist $s = x_1$, $\omega = y$ for the group $\Z/2\times\Z/2^k$, where $x_1$ generates $H^1(B\Z/2;\Z/2)$ and $y$ generates $H^2(B\Z/2^k;\Z/2)$. We compute the corresponding twisted supercohomology group in \cref{prop:22kxy}:
\begin{equation}
\label{intro224}
    \SH^5(B\Z/2\times B\Z/2^k,x_1,y) \cong \Z/2 \oplus \Z/2 \oplus \Z/4\,,
\end{equation}
We also describe how to choose this isomorphism such that the classes $\alpha_{\mathrm{Maj}}$, $\alpha_{\mathrm{DW}}$, and $\alpha_{\mathrm{GW}}$ mapping to $(1, 0, 0)$, $(0, 1, 0)$, and $(0, 0, 1)$ under~\eqref{intro224}, respectively, are in the Majorana, Dijkgraaf--Witten, and Gu--Wen layers respectively, and we show that the kernel of the map to spin bordism is the subgroup generated by $\alpha_{\mathrm{Maj}}$.
 \end{example}

\begin{rem}
If $x$ denotes the generator of $H^1(B\Z/2^k;\Z/2)$, then the twists $(x_1, y)$ and $(x_1, x^2+y)$ over $B\Z/2\times B\Z/2^k$ are equivalent: as we describe in \cref{appendix:timerev}, they are exchanged by an automorphism of $B\Z/2\times B\Z/2^k$. This corresponds to redefining the generator $T$ by replacing it with $T\cdot g^{2^{k-1}}$. We will work with $(x_1,y)$ in this paper.
\end{rem}

\begin{thm}\label{thm:timerev}
If $p$ denotes the map in the short exact sequence
\begin{equation}
   \begin{tikzcd}
     1 \to \Z/4\times\Z/2\to \Z/8 \times \Z/2^{k+1} \xrightarrow{p} \Z/2 \times \Z/2^k \to 1,
    \end{tikzcd}
    \end{equation}
then $\alpha_{\mathrm{GW}}$ and $\alpha_{\mathrm{DW}}$ are in the kernel of
\begin{equation}
\begin{aligned}
    p^*\colon\SH^5(B\Z/2\times &B\Z/2^k, x_1, y) \longrightarrow\SH^5(B\Z/8\times B\Z/2^{k+1}, x_1, 0).
\end{aligned}
\end{equation}
\end{thm}

\begin{conjecture}
If $\alpha\in\SH^5(B\Z/2\times B\Z/2^k, x_1, y)$ is in the subgroup spanned by $\alpha_{\mathrm{DW}}$ and $\alpha_{\mathrm{GW}}$, there is a (3+1)d $\Z/4\times\Z/2$ gauge theory with anomaly $\alpha$ obtained by generalizing the Wang--Wen--Witten symmetry extension construction to anti-unitary symmetries.
\end{conjecture}

\begin{table*}[ht]
\centering
\renewcommand{\arraystretch}{1.2}
\resizebox{\textwidth}{!}{
\begin{tabular}
{c | c c c | c|| c c c | c|| c }
\hline \hline 
$G_f$ & $G$ & $s$ & $\omega$ & $\SH^5$ & $H$ & $s$ & $\omega$ & $\SH^5$ & $K$ \\ \hline\hline 
$\Z/p^k\times \Z/2^F$ & $\Z/p^k$ & 0 &0 & $(\Z/p^k,\text{DW})$ &$\Z/p^{k+1}$ & 0 & 0& $(\Z/p^{k+1},\text{DW})$ & $\Z/p$\\ \hline 
{$\Z/2^k\times \Z/2^F,k\geq 2$} & {$\Z/2^k$} & 0 & 0 &$(\Z/2^{k-1},\text{DW})$ & {$\Z/2^{k+m},m=\lceil \frac{k-1}{2}\rceil$} & $0$ &0 & $(\Z/2^{k+m-1},\text{DW})$ & $\Z/2^m$ \\ \hline 
$\Z/4^F$ & $\Z/2$ & $0$ & $x^2$ & $(\Z/8,\text{Maj})$& $\Z/8$ &0 &0 & $
(\Z/2$, DW)  & $\Z/4$\\ \hline 
\multirow{2}{*}{$\Z/(2^{k+1})^F,k\geq 2$} & \multirow{2}{*}{$\Z/2^k$} & \multirow{2}{*}{0}& \multirow{2}{*}{$y$} &$(\Z/2^{k+1},\text{GW})$ &  {$\Z/2^{k+m},m=\lceil \frac{k}{2}\rceil$} & $0$ & $0$ & {$(\Z/2^{k+m-1},\text{DW})$} & $\Z/2^m$\\ 
& & & & $\oplus \,({\Z/2},\text{Maj})$& {$\Z/2^{k+2}$}&  0 & 0 & {$(\Z/2^{k+1},\text{DW})$}  & $\Z/4$ \\
\hline 
\multirow{2}{*}{$\Z/(2^{k+1})^F\times \Z/2^T,k\geq 2$} & \multirow{2}{*}{$\Z/2\times \Z/2^k$} & \multirow{2}{*}{$x_1$} & \multirow{2}{*}{$y/y + x_1^2$} & $\textcolor{green}{\Z/2} \oplus (\Z/2,\text{DW})$ &$\Z/2\times \Z/2^{k+1}$ &$x_1$ &0 & $\Z/2 \oplus \Z/2 \oplus \Z/4$  & $\Z/2$ \\ 
& &  &   & $(\Z/4,\text{GW})$  &  $\Z/8\times  \Z/2^{k+1}$ & $x_1$ & $y$ & \text{Order at Most} 32 &$\Z/4\times \Z/2$\\
\hline  
\hline
\end{tabular}}
\caption{On the left, we present the fermionic groups $G_f$ we consider and their associated bosonic groups $G$ with twists $s\in H^1(BG;\Z/2)$ and $\omega\in H^2(BG;\mathbb{Z}/2)$. 
In the column titled ``$\SH^5$'', each item within a box gives the direct summands of the full group $\SH^5(X,s,\omega)$. Alongside each summand we provide the layer in which the generator for that group resides: either the Majorana (Maj), Gu--Wen (GW), or Dijkgraaf--Witten (DW) layer. The generator in green indicates that its image in the twisted spin bordism is zero.
On the right, we present the extended group $H$ with associated twists, such that the corresponding generators of $\SH^5$ on the left is trivialized, as well as the group $K$ indicating the TQFT we construct using the symmetry extension procedure. In the first line, $p$ is an odd prime number.}
\label{tab:results}
\end{table*}

\subsection{Outline}

The structure of this paper is as follows. In \S\ref{subsection:symmetryextension} we review the symmetry extension construction in the bosonic setting, and explain in \S\ref{subsection:TO} how the data of the symmetry extension construction, with regular cohomology replaced by supercohomology, naturally fits into the data of the classification of fermionic anomalous SETs using the language of fusion 2-categories. We prove the main theorems \cref{mainthmA,mainthmB} in \S\ref{section:always_works}. In 
\S\ref{section:construction}, we go to specific symmetries and anomalies listed in \S\ref{subsection:results}, show how to trivialize the generators for the anomalies and identify the anomalous TQFTs that saturate these anomalies. We conclude in \S\ref{section:conclusion}.

Various appendices provide additional background information and technical supercohomology computations that complement the main text. In Appendix~\ref{subsec:generalized_coho}, we review the basics of generalized cohomology theories, which will be used in the next two appendices. In Appendix \ref{subsection:twistedSH}, we explain two equivalent definitions of supercohomology and how to realize the twists of supercohomology corresponding to $s$ and $\omega$. In Appendix~\ref{section:anomaliesObstructions}, we explain how to obtain the anomalies of topological orders from obstruction theory, and its relationship to both supercohomology and the regular 't Hooft anomaly.  
In Appendix~\ref{appendix:B}, we fill in the technical details of supercohomology computations needed for \S\ref{subsection:results}. In Appendix~\ref{appendix:timerev}, we provide one extra example involving time-reversal symmetry, following the analogous computations for unitary symmetries.

\section{Symmetry extension construction of fermionic anomalous topological orders}\label{section:preliminaries}

In this section, we detail our construction of fermionic anomalous topological orders, which combines the symmetry-extension procedure with the structure of fusion 2-categories.

\subsection{Symmetry extension construction}\label{subsection:symmetryextension}
We now review the symmetry extension procedure of~\cite{2014arXiv1404.3230K,2014PhRvL.112w1602K,Thorngren:2015gtw,Wang:2017loc,Tachikawa:2017gyf}, who construct a bosonic $n$-dimensional TQFT with $G$-symmetry that realizes the anomaly labeled by a class $\alpha \in H^{n+1}(BG;\C^\times)$.\footnote{While the classification of bosonic invertible field theories goes beyond simply group cohomology, we will only consider those that are classified by cohomology in this review.} In particular, Wang–Wen–Witten describe how to construct a $K$-gauge theory equipped with an anomalous $G$-symmetry, provided that we have
a short exact sequence 
\begin{equation}\label{eq:SES}
    1 \longrightarrow K \longrightarrow H \xrightarrow{p} G \longrightarrow 1
\end{equation}
such that the pullback $p^*\alpha$ is trivial in $H^{n+1}(BH;\C^\times)$.

Specifically, the construction goes as follows. We denote by $\widetilde\alpha \in Z^{n+1}(BG;\mathbb{C}^\times)$ a cocycle lift of $\alpha$. The first step is to choose $H$ such that $p^*\widetilde \alpha = d \widetilde \varphi$, with $\widetilde \varphi \in C^{n}(BH;\mathbb{C}^\times)$, i.e.\ $\alpha$ is trivializable upon pulling back. Let $N$ be an $(n+1)$-dimensional manifold with $\partial N = M$, and $P\to N$ be a principal $G$-bundle with the classifying map $g\colon N\rightarrow BG$ whose restriction to $M$ lifts to a principal $H$-bundle $Q\to M$ with the classifying map $h\colon M\rightarrow BH$. The action for the invertible TQFT on $N$ corresponding to $\alpha$ can be written as 
\begin{equation}
\cZ(N) = \exp\left(2\pi i \int_N g^*\widetilde{\alpha}\right).
\end{equation}
Since $g^* \widetilde\alpha  = d (h^* \widetilde \varphi)$
is satisfied on the boundary $M$, we can construct a boundary theory with the partition function 
\begin{equation}\label{eq:boundaryaction}
    \cZ(M)= \frac{1}{\mathrm{Aut}(P)}.\sum_{P\in\pi_0\cat{Bun}_H(M)} \exp\left(-2\pi i \int_M h^*\widetilde \varphi \right).
\end{equation}
Here $\widetilde \varphi(P)$ denotes the pullback of $\widetilde \varphi$ by the classifying map of $P$.

The boundary theory couples to the bulk theory, i.e.\ the boundary theory has $G$-anomaly labelled by $\alpha$, because
\begin{equation}\label{eq:bosonictrivialize}
    \int_N g^* \widetilde \alpha - \int_M h^* \widetilde \varphi
\end{equation}
is invariant upon the $G$-gauge transformations. 
By taking the restriction in $K$, we find that $d\widetilde \varphi|_K = p^*\widetilde \alpha|_K = 0$. Hence, $\widetilde \varphi |_K$ represents an element in $H^n(BK;\C^\times)$ and the boundary theory is a $K$-gague theory with the Dijkgraaf-Witten twist given by $[\widetilde \varphi |_K]$. The class $[\widetilde \varphi |_K]$ thus obtianed is called the \textit{transgression} of $\alpha$ under the short exact sequence \cref{eq:SES} \cite{2014arXiv1404.3230K}. We see that this construction gives a TQFT realizing the given anomaly.

For bosonic anomalies $\alpha$ valued in regular cohomology and $n\geq 3$, this construction is always possible.
Let $\widehat{K} = \mathrm{Hom}(K, \mathrm U(1))$; according to \cite[\S 2.7]{Tachikawa:2017gyf}, it is always possible to choose an abelian group $K$ such that $\widetilde \alpha = \widetilde e\cup \widetilde z$, where $\widetilde e\in Z^2(BG;K)$ and $\widetilde z\in Z^{n-1}(BG;\widehat{K})$.\footnote{Given a $K$-valued $m$-cochain $\omega$ and a $\widehat K$-valued $n$-cochain $\theta$, we will let $\omega\cup\theta$ denote the \emph{logarithm} of the Pontryagin pairing of $\omega$ and $\theta$, so that it is an element of $C^{m+n}(\text{--};\R/\Z)$.} The cocycle $\widetilde e$ gives rise to an extension
\begin{equation}
    \shortexact*[][p]{K}{H}{G},
\end{equation}
and vacuously $p^*(\widetilde e)$ is a coboundary. Thus $p^*(\widetilde\alpha)$ is also a coboundary, so $p^*(\alpha) = 0$.
Therefore, we obtain the desired construction.  

A natural next step is to generalize the symmetry extension construction to the fermionic case, to produce fermionic TQFTs that saturate a given anomaly.  
In \cite{Kobayashi:2019lep}, Kobayashi--Ohmori--Tachikawa make progress in this direction by writing down a path integral for a boundary fermionic TQFT, for which the bulk SPT is classified by a cocycle pair  $(b,c)$ with $c \in C^{d+1}(BG;\mathbb{C}^\times)$ in the Dijkgraaf-Witten layer and  $b \in C^{d}(BG;\Z/2)$ in the Gu-Wen layer.  As we mentioned in the introduction, the full obstruction should have contributions in supercohomology which notably includes a third layer, the \term{Majorana layer}. However, a path integral description of the $K$-gauge theory that generalizes \cite{Kobayashi:2019lep} to include the third layer remains elusive.

Therefore it remains unclear what the symmetry extension construction using supercohomology, i.e.\ pulling back a supercohomology class in such a way that it trivializes on a larger group, actually yields in the fermionic setting. A priori it is only a formal manipulation. To ameliorate this situation, we will explain in \S\ref{subsection:TO} how a fermionic symmetry extension construction naturally arises when axiomatizing anomalous (3+1)d fermionic theories with fusion 2-categories. Specifically, we review the classification (3+1)d symmetry-enriched topological orders (SETs), which  uses fusion 2-categories and twisted supercohomology. The data required to implement the symmetry extension construction is precisely what the classification of (3+1)d $G$-SETs provides. This offers a conceptual foundation for why even without a path integral presentation in terms of cocycles, it is possible to construct a well-defined (3+1)d fermionic topological quantum field theory that is the boundary for an invertible fermionic topological field theory.

\subsection{Anomalous topological orders from fusion 2-categories}\label{subsection:TO}
We now summarize how the classification of (3+1)d fermionic $G$-SETs, as well as their anomalies, is formulated in terms of fusion 2-categories and twisted supercohomology. We then explain how this classification naturally integrates into the framework of the fermionic symmetry extension construction.

Let $G$ be a finite group. The classification of (3+1)d fermionic $G$-SETs is conducted by first starting with the classification of (3+1)d fermionic topological order without any symmetry, and then enriching it with a $G$-symmetry. This enrichment is made precise categorically in \cite{DY2025}, and physically it amounts to adding in $G$-symmetry defects into the theory.
Let $\cZ(\fC)$ denote the Drinfeld center of a fusion 2-category $\fC$~\cite{Kong:2019brm}, which is a nondegenerate braided fusion 2-category.
Topological orders in (3+1)d were classified in \cite{Lan_2018,Lan_2019,JF}, and separated into three cases: 
\begin{enumerate}
    \item When all the excitations are bosonic, the category that describes the topological order takes the form $\cZ(\tVect^\pi_K)$, where $\tVect^\pi_K$ denotes the fusion $2$-category of $K$-graded 2-vector spaces with pentagonator twisted by a class $\pi \in H^4(BK;\mathbb{C}^\times)$~\cite[Construction 2.1.16]{douglas2018fusion}.\footnote{To define an actual fusion $2$-category, one must choose a cocycle representative of $\pi$, but the Morita class, and therefore the topological order, does not depend on this choice.}
    \item When the spectrum contains an emergent fermion, the category that describes the topological order takes the form $\cZ(\tsVect^\varphi_K)$, where $\tsVect^\varphi_K$ is the fusion $2$-category of $K$-graded $2$-super vector spaces with pentagonator twisted by (a cocycle representative of) $\varphi \in \SH^{4}(BK,\omega)$ \cite{Lan_2019}.
    \item When the spectrum contains a local fermion, so that the theory couples to spin structure, the topological orders are classified by a gauge group $K$ and a class in $\SH^4(BK)$ \cite[Corollary V.5]{JF}.
\end{enumerate}

As explained in Appendix~\ref{subsection:twistedSH}, the appearance of supercohomology is essential, as it precisely corresponds to the Picard 2-groupoid $\tsVect$. This correspondence underlies the use of supercohomology in the classification of both $(3+1)$d topological orders and $(3+1)$d (anomalous) SETs.

Analogously to bosonic $G$-SETs in (2+1)d,  bosonic (3+1)d  $G$-SETs are, categorically, nondegenerate  faithfully graded $G$-crossed braided fusion $2$-categories. The work of \cite{DY2025} shows that in the case when the SET has a local fermion they are described and parametrized as follows.
\begin{thm}[{\cite[Proposition 4.4]{DY2025}}]\label{thm:localfermionicSET}
    (3+1)d $G$-SETs with local fermions are equivalent to nondegenerate  $\tsVect$-enriched braided fusion 2-categories with a fully faithful braided 2-functor from $\mathbf{2Rep}(G)$.
\end{thm}

The data of an enrichment of a braided fusion 2-category $\fB$ over $\tsVect$ is a sylleptic functor $\tsVect \to \cZ_2(\fB)$. The notation $\cZ_2(\fB)$ refers to the sylleptic center of the braided fusion 2-category $\fB$, as defined in \cite{Cr}. The objects in the sylleptic center are those braid trivially with all other objects in $\fB$, and thus the functor picks out the local fermion. 
By unpacking this theorem, $\tsVect$-enriched nondegenerate fermionic braided fusion 2-categories
with a fully faithful braided 2-functor from $\mathbf{2Rep}(G)$ are classified by the following data:
\begin{itemize}
    \item A group $H$ fitting into a short exact sequence,
    \begin{equation}\label{eq:symextnF}
       \shortexact*{K}{H}{G}.
    \end{equation}
    \item A class $\varphi'\in\SH^{4}(BH)$.
\end{itemize}
The classification data in 
\cref{thm:localfermionicSET} gives rise to a fermionic symmetry extension construction; the extension $H\to G$ in question is exactly the one appearing in~\eqref{eq:symextnF}.
This gives a  categorically precise explanation of why the following ansatz can be used to construct an anomalous (3+1)d topological order.
\begin{ansatz}\label{ansatz:fernionicTQFT}
 A (3+1)d fermionic topological order with $G$-symmetry and $\alpha\in \SH^{5}(BG,0, \omega)$ obstruction is realized as a $K$-gauge theory from the following data:
\begin{itemize}
    \item A group $H$ such that
    \begin{equation}\label{one_SES}
        1 \longrightarrow K \longrightarrow H \overset p\longrightarrow G \longrightarrow 1
    \end{equation}
    is a short exact sequence.
     \item A 4-cochain $\varphi$ for the group $H$ which trivializes the pullback of $\alpha$ to $H$. 
    \item A class in $\SH^{4}(BK)$ describing the $K$-gauge theory\footnote{If the extension~\eqref{one_SES} splits, then a cocycle representing this supercohomology class is the Lagrangian in a supercohomology version of Dijkgraaf--Witten theory. See Kim~\cite{Kim22} for more on Dijkgraaf--Witten theory in this generality and Freed--Neitzke~\cite{FN23, FN24} for a related construction of a classical Chern--Simons theory. In general, when~\eqref{one_SES} does not split, the Lagrangian description of the $K$-gauge theory is messier, and we found it easier to work with fusion $2$-categories.} obtained from gauging the non-anomalous subgroup $K$  of $H$.
\end{itemize}
\end{ansatz}
\noindent We now elaborate on how the second piece of data involving $\varphi$ is relevant for the construction viewed in the framework of fusion 2-categories. Using the symmetry extension sequence, we consider a braided fusion 2-category $\fC$ constructed from equivariantizing a $H$-action on a grouplike fusion 2-category $\fB$,\footnote{The objects of such a fusion 2-category have grouplike fusion rules for the objects and also go by the name ``strongly fusion 2-categories'', which is a term coined in \cite{Johnson-Freyd:2020ivj}.} given by  a map $BH\rightarrow B \Aut^{br}(\fB)$. See \cite[Section 3.4]{DY2025} for an explanation of how braided fusion 2-categories arise from equivariantization.  Suppose the anomaly $\alpha \in SH^5(BG,0,\omega)$ can be trivialized on $H$ so that $p^* \alpha = d \varphi$. If there is a short exact sequence with $K \hookrightarrow H$  then one can consider the restriction $\varphi |_K$. In particular one can construct the category of bimodules of $\tVect^{\varphi |_K}_K$ over $\fC$.  The resulting category of bimodules therefore yields in fact a non-degenerate braided fusion 2-category, and hence a topological order in (3+1)d. This reflects the fact that $\varphi$ trivializes on $K$, and hence the anomaly is trivializable when restricted from the general $H$-symmetry to a $K$-symmetry. In particular, we can gauge the $K$-symmetry. The categorical construction involving trivializing $\varphi$ when restricted to $K$ parallels the construction around \Cref{eq:bosonictrivialize} used for constructing bosonic theories.  

In summary, the rigorous definition of fermionic $G$-SETs ensures that the corresponding topological order can be obtained by gauging the full $H$-symmetry described in the second item above. Consequently, gauging any subgroup $K\subset H$ yields a well-defined $K$-gauge theory, even in the absence of an explicit path-integral construction. In the local fermion case, the data is essentially the same, except we only consider supercohomology and not twisted supercohomology whenever it appears.

We note that since the obstruction $\varphi$ is valued in twisted supercohomology with $\omega \in H^2(BG;\Z/2)$, if one wants to trivialize the pullback $p^*\varphi$ then it is necessary to pick $\omega' \in H^2(BH;\Z/2)$ such that it is equivalent to the pullback of  $\omega$ to $H$. The equation that needs to be solved to trivialize $p^*\varphi$ is given by
\begin{equation}\label{eq:matching}
    p^* \widetilde{\alpha} = d  \widetilde{\varphi}\,.
\end{equation}
Here $\widetilde{\alpha}$ is a cocycle representative of $\alpha$. Let $ \widetilde\alpha= (a,b,c)$ and $\widetilde\omega'$ be a cocycle representative of $\omega'$, from the explicit cochain formula in Appendix~\ref{subsection:twistedSH}, the equation we need to solve is 
\begin{equation}
    d \varphi = \left(da ,\,  db +(\Sq^2+\widetilde{\omega}')a,\, dc + (-1)^{(\Sq^2+\widetilde{\omega}')b}\cdot  f_{\widetilde{\omega}'}(a)\right).
\end{equation}

Furthermore, taking the cocycle $\widetilde\varphi = (\alpha,\beta,\gamma)$ we find that \Cref{eq:matching} becomes the following system of equations:
\begin{subequations}\label{cocyc_equations}
\begin{align}
    p^* \alpha & = da \\
    p^* \beta  &= db +(\Sq^2+\widetilde \omega')a \\ 
    p^* \gamma &= dc + (-1)^{(\Sq^2+\widetilde \omega')b}\cdot f_{\widetilde{\omega}'}(a)\,,
 \end{align}
 \end{subequations}
in which solving for $(a,b,c)$ would allow us to construct the theory with anomaly cocycle $\widetilde\varphi$.

In cases where the symmetry group involves time-reversal that mixes nontrivially with a finite unitary symmetry and fermion parity, one would expect to construct a boundary TQFT from a class in $\SH^5(X,s,\omega)$ with $s\ne 0$. However, the theory of fusion 2-categories so far does not accommodate twists of supercohomology that arise when $G$ has antiunitary generators, and thus the categorical description for $G$-SETs involving time reversal is not fully fledged. While a complete formulation of the corresponding category with the appropriate unitarity structures has yet to be established, we do not anticipate any fundamental obstructions to its construction.
We thus conjecture, by means of a physically reasonable extrapolation to the unitary setting, that there is an extension of the theory of fusion $2$-categories to not-necessarily-unitary twists, which agrees with the symmetry extension construction applied to $\SH^5(X,s,\omega)$.

\section{Supercohomology anomalies and beyond-supercohomology anomalies}\label{section:always_works}
In this section, we give a complete characterization of which 5d $(BG, s, \omega)$-twisted spin IFTs can be realized as anomalies of four-dimensional TFTs (assuming these TFTs have one-dimensional state spaces on $S^3$, which is the case relevant for topological order). Specifically, we prove the following theorem, which appears below 
\Cref{mainthmA} in the introduction. 
\begin{thm}
\label{existence_and_nonexistence}
Let $G$ be a finite group, $s\in H^1(BG;\Z/2)$, $\omega\in H^2(BG;\Z/2)$, and $\alpha\in\mho_\Spin^5(BG, s, \omega)$.
\begin{enumerate}
    \item\label{SH:yes} Suppose $s = 0$ and $\alpha$ is in the image of the map $\SH^5(BG, 0, \omega)\to \mho_\Spin^5(BG, 0, \omega)$. Then there is an algorithmic construction of a 4d TFT $Z$ of $(BG, 0, \omega)$-twisted spin manifolds such that the deformation class of the anomaly of $Z$ equals $\alpha$ and $Z(S^3)\cong\C$, where $S^3$ carries the $(BG,s,\omega)$-twisted spin structure induced from its unique spin structure..
    \item\label{SH:no} Suppose $\alpha$ is not in the image of $\SH^5(BG, s, \omega)\to \mho_\Spin^5(BG, s, \omega)$. Then there is no 4d TFT $Z$ of $(BG, s, \omega)$-twisted spin manifolds, which assigns a 1-dimensional Hilbert space to $S^3$, whose anomaly is deformation-equivalent to $\alpha$.
\end{enumerate}
\end{thm}
When we say ``an algorithmic construction of a 4d TFT,'' we represent this TFT as a nondegenerate braided fusion $2$-category. We assume, motivated by physics evidence, that the extra assumption, i.e., a 1-dimensional Hilbert space is assigned to $S^3$, is satisfied by fermionic topological orders.

We prove part~\eqref{SH:yes} of \cref{existence_and_nonexistence} in \S\ref{ss:yes}, and prove part~\eqref{SH:no} in \S\ref{ss:no}. We believe that, after developing the theory of unitary fusion $2$-categories, the $s = 0$ assumption in part~\eqref{SH:yes} will be able to be dropped; see \cref{antiunitary_speculation}.

\subsection{Algorithmically constructing TFTs with supercohomology anomalies}
\label{ss:yes}
The essential ingredient in our algorithmic construction is the following theorem. 
\begin{thm}
\label{thm:always_triv}
Let $G$ be a finite group, $s\in H^1(BG; \Z/2)$, $\omega\in H^2(BG; \Z/2)$, $n\ge 5$, and $\alpha\in\SH^n(BG, s, \omega)$. Then there is an algorithmic construction of a finite group $\widetilde G$ and a map $\rho\colon \widetilde G\to G$ such that \begin{equation}
    \rho^*(\alpha) = 0\in \SH^n(B\widetilde G, \rho^*(s), \rho^*(\omega)).
\end{equation}
\end{thm}
\begin{rem}
\label{WW_qn}
Wan--Wang~\cite[\S II]{WW25} asked, given $G$, $s$, $\omega$ and an $\alpha\in\mho_\Spin^5(BG, s, \omega)$, how can one find $(\widetilde G, \rho)$ trivializing $\alpha$ as in \cref{thm:always_triv} when they exist? (And, implicitly, how can one tell when no such $(\widetilde G, \rho)$ exists?) \Cref{existence_and_nonexistence,thm:always_triv} provide a complete answer to this question: when $\alpha$ is the image of a supercohomology class, one can construct $(\widetilde G, \rho)$ by the algorithm in \cref{thm:always_triv}. Conversely, if $\alpha$ is not the image of a supercohomology class, no such $(\widetilde G, \rho)$ could trivialize $\alpha$, as that would give rise to a TFT $Z$ with anomaly $\alpha$ and $\dim(Z(S^3)) = 1$, in violation of \cref{existence_and_nonexistence}, part~\eqref{SH:no}.

This answer does not close the book on Wan--Wang's question, though: the algorithm that will appear in the proof of \cref{thm:always_triv} is far from optimal, in that $\widetilde G$ is generally much larger than necessary, making it unwieldy in practice for actually constructing TFTs. This can be seen when $G\cong\Z/2^n$ and $s = 0$, where $\widetilde G\to G$ may be taken to be $\Z/2^{n+2}\to\Z/2^n$ by \cref{thm:Znsymmetry,thm:trivializedGen} and Wan--Wang~\cite[\S\S IV.B, V.A, V.C]{WW25}. It would be interesting to improve the algorithm in \cref{thm:always_triv} to produce smaller $(\widetilde G, \rho)$.
\end{rem}

The key ingredient in our proof of \cref{thm:always_triv} is \cref{trivcoh}, an analogue for ordinary cohomology. According to~\cite[\S 1]{DLM24}, \cref{trivcoh} is a folklore theorem, well-known to experts. Versions or special cases of it appear in work of Wang--Wen--Witten~\cite[\S V.A]{Wang:2017loc}, Tachikawa~\cite[\S 2.7]{Tachikawa:2017gyf},  Décoppet~\cite[Proposition 4.2.3]{D9}, and and DeLazzer Meunier~\cite[Theorem 3.4]{DLM24}.
\begin{thm}
\label{trivcoh}
Let $G$ be a finite group, $M$ be a $\Z[G]$-module, $n\ge 3$, and $\alpha\in H^n(BG; M)$. Then there is an algorithmic construction of a finite group $\widetilde G$ and a homomorphism $\rho\colon \widetilde G\to G$ such that $\rho^*(\alpha) = 0$.
\end{thm}
The references above do not state that this construction may be done algorithmically in the input data $(G, M, n, \alpha)$, but this is clear from DeLazzer Meunier's proof (\textit{ibid.}).
\begin{proof}[Proof of \cref{thm:always_triv}]
The proof proceeds by trivializing one layer at a time. Namely, first we find a group $G_1\rightarrow G$ such that the Majorana layer is trivialized, then we find a group $G_2\rightarrow G_1$ such that the Gu--Wen layer is trivialized, and finally we find a group $G_3\rightarrow G_2$ such that the Dijkgraaf--Witten layer is also trivialized. 

Since the Postnikov truncation of $\SH$ to degrees $-2$ and below has exactly one nonzero homotopy group, $\pi_{-2}(\tau_{\le -2}\SH)\cong\Z/2$, there is a canonical homotopy equivalence $\tau_{\le -2}\SH\overset\simeq\to \Sigma^{-2}H\Z/2$. The map induced by $\tau_{\le -2}$,
\begin{equation}
    (\bl)_{\mathrm{Maj}}\colon \SH^n(BG, s, \omega) \longrightarrow H^{n-2}(BG; \Z/2),
\end{equation}
sends a class $\alpha$ to its \term{Majorana layer} $\alpha_{\mathrm{Maj}}$. The fiber of $\tau_{\le -2}\colon\SH\to\Sigma^{-2} H\Z/2$ is the usual map $j\colon \mathit{rSH}\to\SH$ from restricted supercohomology to supercohomology.

Given $G$, $n$, and $\alpha$ as in the theorem statement, let $\rho_1\colon G_1\to G$ be the finite cover constructed algorithmically in \cref{trivcoh} such that $\rho_1^*(\alpha_{\mathrm{Maj}}) = 0 \in H^{n-2}(BG_1; \Z/2)$. Then we have the following commutative diagram, whose top row is exact:
\begin{equation}\begin{tikzcd}
	{\mathit{rSH}^n(BG_1, \rho_1^*(s), \rho_1^*(\omega))} & {\SH^n(BG_1, \rho_1^*(s), \rho_1^*(\omega))} & {H^{n-2}(BG_1;\Z/2)}. \\
	& {\SH^n(BG, s, \omega)} & {H^{n-2}(BG;\Z/2)}
	\arrow["j", from=1-1, to=1-2]
	\arrow["{(\bl)_{\mathrm{Maj}}}", from=1-2, to=1-3]
	\arrow["{\rho_1^*}", from=2-2, to=1-2]
	\arrow["{(\bl)_{\mathrm{Maj}}}", from=2-2, to=2-3]
	\arrow["{\rho_1^*}", from=2-3, to=1-3]
\end{tikzcd}\end{equation}
Thus $(\rho_1^*\alpha)_{\mathrm{Maj}} = \rho_1^*(\alpha_{\mathrm{Maj}}) = 0$, so by exactness there is a class $\beta\in\mathit{rSH}^n(BG_1, \rho_1^*(s), \rho_1^*(\omega))$ with $j(\beta) = \rho_1^*(\alpha)$. It therefore suffices to prove the theorem with $\mathit{rSH}$ in place of $\SH$: if we can trivialize $\beta$ after pulling back to a finite cover, the same is true for $j(\beta) = \rho_1^*(\alpha)$, and the composition of finite covers is finite.

Now do exactly the same thing with the \term{Gu--Wen layer} $\beta_{\mathrm{GW}}\in H^{n-1}(BG_1;\Z/2)$, constructed from the Postnikov truncation of $\mathit{rSH}$ to degrees $-1$ and below. We obtain a fiber sequence
\begin{equation}
    H\C^\times \overset{j'}{\longrightarrow} \mathit{rSH} \overset{(\bl)_{\mathrm{GW}}} \longrightarrow \Sigma^{-1}H\Z/2,
\end{equation}
so if $\rho_2\colon G_2\to G_1$ is the finite cover constructed algorithmically in \cref{trivcoh} such that $\rho_2^*(\beta_{\mathrm{GW}}) = 0 \in H^{n-1}(BG_2; \Z/2)$, then we have the following commutative diagram, whose top row is exact:
\begin{equation}\begin{tikzcd}
	{H^n(BG_2, \C^\times_{\rho_2^*(\rho_1^*(s))})} & {\mathit{rSH}^n(BG_2, \rho_2^*(\rho_1^*(s)), \rho_2^*(\rho_1^*(\omega)))} & {H^{n-1}(BG_2;\Z/2)}. \\
	& {\mathit{rSH}^n(BG_1, \rho_1^*(s), \rho_1^*(\omega))} & {H^{n-1}(BG_1;\Z/2).}
	\arrow["j'", from=1-1, to=1-2]
	\arrow["{(\bl)_{\mathrm{GW}}}", from=1-2, to=1-3]
	\arrow["{\rho_2^*}", from=2-2, to=1-2]
	\arrow["{(\bl)_{\mathrm{GW}}}", from=2-2, to=2-3]
	\arrow["{\rho_2^*}", from=2-3, to=1-3]
\end{tikzcd}\end{equation}
Just as before, we see $\rho_2^*(\rho_1^*(\alpha))$ is the image of a class $\gamma\in H^n(BG_2; \C^\times_{\rho_2^*(\rho_1^*(s))})$ under $j'$, and therefore it suffices to trivialize $\gamma$ by pulling back to a finite cover. For this, apply \cref{trivcoh} again.
\end{proof}
This almost implies part~\eqref{SH:yes} of \cref{existence_and_nonexistence} -- the only promise yet to fulfill is that we can construct $Z$ algorithmically. For this it suffices to produce algorithms doing the following two things for a finite group $G$ and twisting data $(s,\omega)$.
\begin{enumerate}
    \item Given $\alpha\in\SH^n(BG, s, \omega)$, find a cocycle representative $\widetilde\alpha$ for $\alpha$, using the cocycle description of supercohomology we gave in Appendix \ref{subsection:twistedSH}.
    \item Given a cocycle $\widetilde\gamma$ for supercohomology known to be exact, find a cochain $x$ solving the equations~\eqref{cocyc_equations}.
\end{enumerate}
We must take a slight detour before solving these problems.
\begin{defn}
Recall that $\tsVect^\times$ denotes the Picard $2$-groupoid of the Morita $2$-category of complex superalgebras. For any even natural number $N$, let ${}_N(\tsVect)^\times$ denote the sub-$2$-category of $\tsVect^\times$ consisting of all objects, all $1$-morphisms, and only the $2$-morphisms corresponding to $N^{\mathrm{th}}$ roots of unity $\mu_N$ inside $\C^\times$. Then, let $\SH_N$ denote $\Sigma^{-2}$ of the classifying spectrum for ${}_N(\tsVect)^\times$.
\end{defn}
$\Sigma^{-2}$ appears because $\SH$ is $\Sigma^{-2}$ of the classifying spectrum for $\tsVect^\times$. 
\begin{lem}\hfill
\label{SHN_lem}
\begin{enumerate}
    \item\label{SHN_cocyc} $\SH_N$ has a cocycle theory identical to the one for $\SH$ given in Appendix \ref{subsection:twistedSH}, except with $c\in C^n(BG; \mu_N)$ instead of $C^n(BG; \C^\times)$.
    \item\label{who_is_cof} The inclusion ${}_N(\tsVect)^\times\hookrightarrow \tsVect^\times$ induces a map of spectra $\iota_N\colon\SH_N\to\SH$ whose cofiber is $H\C^\times$.
    \item\label{mostly_isom} On homotopy groups, $\iota_N$ is an isomorphism in all degrees except $0$, where it is canonically identified with the inclusion $\mu_N\inj \C^\times$.
\end{enumerate}
\end{lem}
\begin{proof}
Part~\eqref{SHN_cocyc} is identical to the proof of the cocycle description for twisted supercohomology.
Both~\eqref{who_is_cof} and~\eqref{mostly_isom} follow from the standard description of the homotopy groups of the spectrum classifying a Picard $2$-groupoid~\cite[Lemma 3.2]{GJO17}. This description is natural in the Picard $2$-groupoid, so the fact that the objects and $1$-morphisms have not changed implies that $\iota_N$ is an isomorphism in degrees $-1$ and below; coconnectivity implies it is an isomorphism in degrees $1$ and above. In degree $0$, the map is the inclusion of $\Aut_{\Aut(1)}(1)$ from ${}_N(\tsVect)^\times$ into $\tsVect^\times$, which by construction is $\mu_N\subset\C^\times$. Thus using the long exact sequence on homotopy groups, the cofiber has a single homotopy group isomorphic to $\C^\times/\mu_N\cong\C^\times$.
\end{proof}
\begin{cor}
\label{SHN_isom}
Let $G$ be a finite group, $s\in H^1(BG;\Z/2)$, and $\omega\in H^2(BG;\Z/2)$, If $\# G\mid N$, then the map $\SH_N^*(BG,s,\omega)\to\SH^*(BG,s,\omega)$ is an isomorphism in degrees $1$ and above.
\end{cor}
\begin{proof}
By \cref{SHN_lem}, part~\eqref{who_is_cof}, the Postnikov truncation $\tau_{\le -1}$ applied to $\iota_N$ has cofiber $\tau_{\le -1}(H\C^\times) = 0$, so $\tau_{\le -1}(\iota_N)\colon\tau_{\le -1}\SH_N\to\tau_{\le -1}\SH$ is an equivalence of spectra.

Apply the fiber sequence $\tau_{\ge 0}\to \id\to\tau_{\le -1}$ to the map $\iota_N$ to obtain a commutative diagram, whose rows are fiber sequences:
\begin{equation}\begin{tikzcd}
	{H\mu_N} & {\SH_N} & {\tau_{\le -1}\SH_N} \\
	{H\C^\times} & \SH & {\tau_{\le -1}\SH}.
	\arrow["{\tau_{\ge 0}}", from=1-1, to=1-2]
	\arrow[from=1-1, to=2-1]
	\arrow["{\tau_{\le -1}}", from=1-2, to=1-3]
	\arrow["{\iota_N}", from=1-2, to=2-2]
	\arrow["\simeq", from=1-3, to=2-3]
	\arrow["{\tau_{\ge 0}}", from=2-1, to=2-2]
	\arrow["{\tau_{\le -1}}", from=2-2, to=2-3]
\end{tikzcd}\end{equation}
Map into this diagram with $BG$ to obtain a commutative diagram of long exact sequences. Since $N$ divides the order of $G$, the map $H^*(BG;\mu_N)\to H^*(BG;\C^\times)$ is an isomorphism in positive degrees, Thus two of the three maps in the diagram of long exact sequences are isomorphisms, so the third is as well by the five lemma.
\end{proof}
Now we can prove the ``existence'' part of \cref{existence_and_nonexistence}, that if $\alpha$ is a class in $\mho_\Spin^5$ realized as the image of a supercohomology class, then one can algorithmically construct a 4d TFT $Z$ whose anomaly is deformation-equivalent to $\alpha$. Here we assume the degree-$1$ part of the twist vanishes (though see \cref{antiunitary_speculation}).
\begin{proof}[Proof of \cref{existence_and_nonexistence}, part~\eqref{SH:yes}]
The symmetry extension procedure in this paper guarantees the existence of a TFT $Z$ whose anomaly's deformation class equals $\alpha$. We need to find a cocycle representative for $\alpha$ and, once we have trivialized $\alpha$ by pulling back, find a cochain solving the equations~\eqref{cocyc_equations}. The key takeaway from \cref{SHN_isom} is that we may equivalently work with $\SH_N$, which has a compatible cocycle description by \cref{SHN_lem}, part~\eqref{SHN_cocyc}. But since $G$ and $N$ are finite, this is a \emph{finite} cocycle description: in each degree $n$, there are finitely many cochains. Thus one algorithmic solution to the two questions we need to solve is to simply try each one of the cochains for $\SH_N$ in the correct degree, as \cref{thm:always_triv} guarantees that at least one will work.
\end{proof}

\begin{rem}
\label{antiunitary_speculation}
In the existence part of \cref{existence_and_nonexistence}, we assume $s = 0$. As we discussed in \cref{s_nonzero}, this is because the fusion $2$-category methods that we use in this paper have not yet been extended to antiunitary symmetries. We conjecture that, once this technology is in place, \cref{SH:yes} will generalize to arbitrary $(s,\omega)$.
\end{rem}

\subsection{The \texorpdfstring{$p+ip$}{p+ip} layer enforces gaplessness}
\label{ss:no}
In this subsection, we prove part~\eqref{SH:no} of \cref{existence_and_nonexistence}: the ``nonexistence'' part. We need to show that, if $\alpha\in\mho_\Spin^5(BG, s, \omega)$ is \emph{not} the image of a twisted supercohomology class, then no 4d TFT $Z$ with $\dim(Z(S^3)) = 1$ has anomaly deformation-equivalent to $\alpha$. This result is discussed from the symmetry extension point of view in \cite{Hsi18,Cheng:2024awi,WW25} and work in progress of Wan--Wang--Yau~\cite{ZJSNew}. They show that for beyond supercohomology anomalies, the corresponding elements cannot be trivialized by symmetry extension. This, together with the analysis in the context of fusion 2-categories, strongly suggests that beyond supercohomology anoamlies cannot be realized by four-dimensional TFTs. Our proof is inspired by an argument of Córdova--Ohmori~\cite{CO2}, and refines it: Córdova--Ohmori prove the theorem for $(G, s, \omega) = (\Z/2n, 0, y)$ (where $y$ is the nonzero element of $H^2(B\Z/2n;\Z/2)\cong\Z/2)$ using a combination of topological and index-theoretic methods; our proof replaces the index theory with an algebraic-topological argument and thus generalizes more easily to all finite groups.

First, we identify the complete obstruction to a class $\alpha\in\mho_\Spin^5(BG, s, \omega)$ being in the image of the map from $\SH^5(BG, s, \omega)$. We will call this obstruction the \term{$p+ip$ layer} $\alpha_{p+ip}\in H^1(BG; \C^\times_s)$.

The precise definition and characterization of the $p+ip$ layer (\cref{pip_defn}) is a little technical, so we give it in an appendix. As a good approximation, the reader can think of it as follows.
\begin{fakedefn}
\label{notpip}
In the Atiyah--Hirzebruch spectral sequence for $\mho_\Spin^*(BG, s, \omega)$, the submodule of the $E_\infty$-page in total degree $5$ is naturally a direct sum of four abelian groups: $E_\infty^{5,0}$, $E_\infty^{4,1}$, $E_\infty^{3,2}$, and $E_\infty^{1,4}$. For degree reasons, there are no nonzero differentials into $E_r^{1,4}$, so there is an inclusion $i\colon E_\infty^{1,4}\hookrightarrow E_2^{1,4}$.

The $p+ip$ layer is the homomorphism $\mho_\Spin^5(BG, s, \omega)\to H^1(BG; \C^\times_s)$ defined as the following composition:
\begin{equation}
    (\text{--})_{p+ip}\colon
        \mho_\Spin^5(BG, s, \omega) \to E_\infty^{\bullet, 5-\bullet} \twoheadrightarrow E_\infty^{1,4} \overset{i}{\hookrightarrow} E_2^{1,4}\cong H^1(BG; \C^\times_s).
\end{equation}
The first two arrows are the projection onto the associated graded, resp.\ onto a direct summand.
\end{fakedefn}
In \cref{pip_defns_agree}, we prove that if $(s, \omega)$ are $(w_1, w_2)$ of a vector bundle $V\to BG$, this provides a correct definition of the $p+ip$ layer, in agreement with the one in \cref{pip_defn}; however, such a bundle $V$ does not exist in general, which causes technical problems with the construction of the Atiyah--Hirzebruch spectral sequence used in \cref{notpip}.

The essential fact we use about the $p+ip$ layer is that it is the complete obstruction to realizing a 5d reflection-posiitve IFT as the image of a supercohomolgy class:
\begin{lem}
\label{cofib_pip}
There is a long exact sequence
\begin{equation}\label{pipexact}
    \dotsb \to\SH^5(BG, s, \omega) \longrightarrow
    \mho_\Spin^5(BG, s, \omega) \overset\varphi\longrightarrow
    H^2(BG;\Z_s)\to\dotsb
\end{equation}
Thus, given $\alpha\in\mho_\Spin^5(BG, s, \omega)$, $\alpha_{p+ip} = 0$ if and only if $\alpha$ is in the image of $\SH^5(BG, s, \omega) \to\mho_\Spin^5(BG, s, \omega)$.
\end{lem}
We will prove this in Appendix~\ref{s:p+ip}.

Therefore, what we will actually prove is that if $\alpha_{p+ip}\ne 0$, then there cannot be a 4d TFT $Z$ with anomaly $\alpha$. When $s = 0$, $\alpha_{p+ip}$ can be identified with a homomorphism $G\to\C^\times$, i.e.\ a group character.

The proof starts with the following observation by Córdova--Ohmori \cite{CO2}.
\begin{thm}[Córdova--Ohmori \cite{CO2}]
\label{thm:K3}
Let $Z$ be a four-dimensional unitary spin TFT such that $Z(S^3)$ is one-dimensional. Then the partition function $Z(\mathrm{K3}) = 0$.
\end{thm}
We sketch the proof by Córdova--Ohmori here for the reader's convenience.
\begin{proof}[Proof Sketch]
Give the unit disc $D^4$ its unique spin structure and regard it as a bordism $\varnothing\to S^3$, so that $Z(D^4)$ is a linear map $\C\to Z(S^3)$. Let $\lvert0\rangle \in Z(S^3)$ be the image of $1$ under this map. Since $S^4\cong D^4\cup{S^3}D_4$, the partition function $Z(S^4) = \langle 0\mid 0\rangle$, which is positive by unitarity. Moreover, \cite{CO1} proved that the partition function $Z(S^2\times S^2)\ne 0$.

Let $X$ be a closed spin $4$-manifold and $\check X\coloneqq X\setminus D$ where $D$ is a disc in $X$. Again regarding $X$ as a bordism $\varnothing\to S^3$, it defines a linear map $\C\to Z(S^3)$, which sends $1\mapsto Z(X)/Z(1)$ times $\lvert 0\rangle$. Using this, one sees that for closed spin $4$-manifolds $X$ and $Y$,
\begin{equation}\label{eq:Zconnected}
    Z(X\mathbin{\#}Y) = \frac{Z(X)Z(Y)}{Z(S^4)}~.
\end{equation}

Wall~\cite{Wal64} showed that, for any simply-connected smooth spin 4-manifold $X$, there exists an integer $\ell$ such that 
\begin{equation}\label{wall}
    X\mathbin{\#}(-X) \mathbin{\#} (S^2\times S^2)^{\#\ell}  \stackrel{\text{diff.}}{\cong}(S^2\times S^2)^{\#(\ell+\chi(X)-2)}~,
\end{equation}
where $-X$ denotes $X$ with the opposite orientation, $(S^2\times S^2)^{\#\ell}$ is the connected sum of $\ell$ copies of $S^2\times S^2$, and $\chi(X)$ is the Euler number of $X$. In general, the partition functions on $X$ and $-X$ are complex conjugates.  Hence \eqref{eq:Zconnected} and \eqref{wall} imply that for a simply connected spin manifold $X,$ the absolute value of the partition function depends only on its Euler number:
\begin{equation}\label{eq:Zsimpconn}
    |Z(X)|^2 = Z(S^2\times S^2)^{\chi(X)-2}Z(S^4)^{4-\chi(X)}.
\end{equation}
Since $Z(S^2\times S^2)$ and $S(Z^4)$ are both nonzero, as noted above, we conclude that the partition function on a simply-connected $X$ of any fermionic TQFT with a 1-dimensional Hilbert space on $S^{3}$ must satisfy
\begin{equation}
    Z(X) \neq 0~.
\end{equation}
In particular, K3 manifold should satisfy $Z(\mathrm{K3})\neq 0$.
\end{proof}
On the other hand, we want to analyze $Z(\mathrm{K3})$ in the presence of the anomaly $\alpha$. This means that we need to put $Z$ on the boundary of the IFT corresponding to the anomaly $\alpha$, i.e., $Z$ is the relative theory \cite{Freed:2012bs} relative to $\alpha$.

\begin{lem}
\label{it_is_even}
Let $X$ be a closed spin $4$-manifold and $\alpha$ be a 5d IFT of $(BG, s, \omega)$-twisted spin manifolds. Then the super vector space $\alpha(X)$ is one-dimensional, with even grading.
\end{lem}
\begin{proof}
That $\alpha(X)$ is a super vector space at all follows from the fact that our IFTs are valued in $4\sVect$, and $\Omega^4$ of this category is canonically equivalent to $\sVect$ as symmetric monoidal categories. That $\alpha(X)$ is one-dimensional is because $\alpha$ is invertible. Finally, we want to show that $\alpha(X)$ is even. Following~\cite[\S 4.3.5]{DG18}, this is equivalent to asking that that the image of $[X]\in\Omega_4^\Spin$ under the map $\pi_4(\alpha)\colon \Omega_4^\Spin\to\pi_0(\sVect^\times)\cong\Z/2$ is zero. Here, we may use spin bordism instead of twisted spin bordism because the twisted spin structure on $X$ is induced from a spin structure.

To prove this, act by $\eta\in\pi_1(\mathbb S)$ on both sides: since $\pi_4(\alpha)$ was induced from the map of spectra classifying $\alpha$, it commutes with the action by $\eta$, giving rise to a commutative diagram
\begin{equation}
\begin{tikzcd}
	{\Omega_5^\Spin} & {\pi_1(\sVect)\cong\C^\times} \\
	{\Omega_4^\Spin} & {\pi_0(\sVect)\cong\Z/2.}
	\arrow["{\pi_5(\alpha)}", from=1-1, to=1-2]
	\arrow["\eta", from=2-1, to=1-1]
	\arrow["{\pi_4(\alpha)}", from=2-1, to=2-2]
	\arrow["\eta", from=2-2, to=1-2]
\end{tikzcd}
\end{equation}
Since $\Omega_5^\Spin = 0$~\cite{Mil63} and $\eta\colon\pi_0(\sVect)\to\pi_1(\sVect)$ is isomorphic to the unique injection $\Z/2\hookrightarrow\C^\times$~\cite[\S 4.3.5]{DG18}, $\pi_4(\alpha)$ must vanish.
\end{proof}
Since $\alpha$ is an $(BG, s, \omega)$-twisted spin theory, $\alpha(\mathrm{K3})$  has a $G$-action coming from mapping cylinders of $G$ acting on the trivial principal $G$-bundle by left multiplication, as we explain below in \cref{maptorchar}. Assuming this for now, $Z(\mathrm{K3})$ is necessarily fixed under this action, so since it is nonzero, the whole 1-dimensional vector space $\alpha(\mathrm{K3})$ is invariant under the $G$-action. As such, \cref{it_is_even} prompts us to consider the partition function of $\alpha$ on mapping tori $\mathrm{K3}\times S^1$ with nontrivial holonomy around $S^1$. The invariance of $Z(K3)\neq 0$ suggests that the partition function of the mapping torus is always 1. However, we will show below that when $\alpha$ is a beyond-supercohomology IFT, the partition function can be nontrivial.

First, we reduce the case where $s=0$. In what follows, $P_{\mathrm{triv}}$ will denote a trivial principal bundle, where the group and base will be clear from context.

\begin{defn}
Let $X$ be a closed spin $4$-manifold. For any $g\in G$, let $W_{X,g}$ denote the bordism from $(X, P_{\mathrm{triv}})$ to itself which is $(X\times[0,1], P_{\mathrm{triv}})$ as a manifold with $G$-bundle, such that the incoming boundary is attached by $(\id_X, \id_{P_{\mathrm{triv}}})$, and the outgoing boundary is attached by $(\id_X, g)$.
\end{defn}
Thus $W_{X,g}$ has a canonical $(BG, 0, \omega)$-twisted spin structure for any $\omega$, and as $(BG, 0, \omega)$-structured manifolds, there is a diffeomorphism rel boundary
\begin{equation}
\label{compose_bord}
    W_{X,h}\cup_X W_{X,g}\cong W_{X,gh}.
\end{equation}
To regard $W_{X,g}$ as a bordism, we must describe how $(X, P_{\mathrm{triv}})$ is attached at the incoming and outgoing boundaries, which is more subtle than it initially appears. We will follow~\cite[\S 6.2]{DG18}. Fix a spin point $\pt$ and let $I_1$ denote the spin interval $[0,1]$, regarded as a spin bordism from $\pt$ to $\pt$, such that under the quotient $0\sim 1$, we obtain the \emph{nonbounding} spin circle. Then there is a diffeomorphism, as bordisms, $I_1\cup_{\pt} I_1\cong I_1$~\cite[Lemma 6.8]{DG18}. Thus, letting $W_{X,g}$ to be the bordism $X\times I_1$, with the attaching maps for the principal $G$-bundle described above, \eqref{compose_bord} upgrades to an identity in the bordism (higher) category:
\begin{equation}
\label{compose_bord_part_2}
    W_{X,h}\circ W_{X,g}\simeq W_{X,gh}\colon X\to X.
\end{equation}

\begin{defn}\label{maptorchar}
Let $\alpha$ be a 5d IFT of $(BG, 0, \omega)$-twisted spin manifolds and $X$ be a closed spin $4$-manifold. By \cref{it_is_even}, there is an isomorphism $\alpha(X)\cong\C$; fix one, then define the character $\chi_{\alpha,X}\in G^\vee$ by $\chi_{\alpha,X}(g) \coloneqq \alpha(W_{X,g})\colon\alpha(X)\to\alpha(X)$.
\end{defn}

\begin{prop}\label{lem_main_lemma}
When $s=0$, $\chi_{\alpha, K3} = 0$ if and only if $\alpha_{p+ip} = 0$ as characters of $G$.
\end{prop}

Assuming \cref{lem_main_lemma}, we immediately conclude that, if $\alpha_{p+ip}\ne 0$, there is nonzero vector that is invariant under $G$ action, a contradiction. Then we prove part~\eqref{SH:no} of \cref{existence_and_nonexistence} for $s=0$. 

When $s$ is nontrivial, we just need to restrict $G$ to the unitary subsymmetry $\widetilde{G}:=ker(s)$ and prove that the subsymmetry $\widetilde{G}$ is enough to obstruct the existence of a 4d TFT $Z$.


\begin{prop}
\label{first_reduction}
Suppose $s$ is nontrivial and $\alpha_{p+ip}$ is nontrivial in $H^1(BG; \mathbb{C}_s^\times)$, then $\alpha_{p+ip}|_{\widetilde{G}}$ is nontrivial in $H^1(B\widetilde{G}; \mathbb{C}_s^\times)$.
\end{prop}

\begin{proof}
Since $G$ is finite, there is a  natural isomorphism $H^1(BG;\mathbb{C}^\times_s)\cong H^2(BG;\Z_s)$ so we may as well use the image of the $p+ip$ layer in $H^2(\bl;\Z_s)$. Consider the $\Z$-coefficient Serre spectral sequence associated to the short exact sequence of groups,
\begin{equation}
    \shortexact*[][s]{\widetilde G}{G}{\Z/2}.
\end{equation}
On the $E_2$-page, we have 
\begin{sseqdata}[name=simpleserre, cohomological Serre grading, xrange={0}{3}, yrange={0}{3}, yscale=0.7, xscale=2.2, classes = {draw=none}]
\class["0"](0, 0)
\class["0"](1, 0)
\class["0"](1, 1)
\class["{H^0(B\Z/2;H^2(B \widetilde{G} ; \Z))}"](2, 0)
\class["\Z/2"](0, 1)
\class["0"](0, 2)
\end{sseqdata}
\begin{equation}\label{p2_SimpleSerre}
\begin{gathered}
\printpage[name=simpleserre, page=2]
\end{gathered}
\end{equation}
In particular, the $p = 1$ column is zero because $H^1(BG;\Z)$ is trivial for all finite groups $G$. $E_2^{0,2} = 0$ because $H^2(B\Z/2;\Z_s) = 0$ when $s$ is nontrivial~\cite[Lemma 1]{Cad99}.

Therefore, if $\alpha_{p+ip}\ne 0$, its image in the $E_\infty$-page of~\eqref{p2_SimpleSerre} must be in $E_\infty^{2,0}$, so its restriction to $\widetilde{G}$ must be nontrivial. 
\end{proof}


Now we just need to embark on the proof of \cref{lem_main_lemma}, which is the most technical part. It will be helpful to describe $\chi_{\alpha,X}$ in terms of bordism. First, recall that, thanks to the Atiyah--Hirzebruch spectral sequence,\footnote{If $\omega$ is not equal to $w_2(V)$ for a vector bundle $V\to BG$, then strictly speaking, the relevant Atiyah--Hirzebruch spectral sequence is not constructed. In this case one can instead use the \term{James spectral sequence}~\cite[Proposition 1]{Tei93}.}
there is an exact sequence
\begin{equation}\label{rtex}
    \rightexact[a][]{\Omega_1^\Spin}{\Omega_1^\Spin(BG, 0, \omega)}{H_1(BG;\Z)},
\end{equation}
where $a$ endows a closed spin $1$-manifold with the trivial $G$-bundle. Thus there is a canonical isomorphism $\psi\colon H_1(BG;\Z)\overset\cong\to\Omega_1^\Spin(BG, 0, \omega)/\mathrm{Im}(a)$. Let $h\colon G\to H_1(BG;\Z)$ be the Hurewicz map, and define
\begin{equation}
\label{char_again}
\widetilde\chi_{\alpha,X}\colon G\overset h\longrightarrow H_1(BG;\Z) \overset\psi\longrightarrow \Omega_1^\Spin(BG, 0, \omega)/\mathrm{Im}(a)\overset{\times X}{\longrightarrow} \Omega_5^{\Spin}(BG, 0, \omega) \overset{\alpha}{\longrightarrow}\C^\times.
\end{equation}
In~\eqref{char_again}, we have implicitly claimed that taking the product with $X$, as a map $\Omega_1^\Spin(BG, 0, \omega)\to \Omega_5^\Spin(BG, 0, \omega)$, factors through the quotient by $\mathrm{Im}(a)$. Well, let us explicitly claim it -- because $X$ is spin, the subgroup of $\Omega_5^\Spin(BG, 0, \omega)$ consisting of classes $C\times X$ with $C\in\mathrm{Im}(a)$ is exactly the image of $\Omega_5^\Spin$ in $\Omega_5^\Spin(BG, 0, \omega)$ -- and $\Omega_5^\Spin\cong 0$.
\begin{lem}
\label{bord_of_character}
$\chi_{\alpha,X} = \widetilde\chi_{\alpha,X}$.
\end{lem}
\begin{proof}
After unwinding the definitions of both characters, this follows from the standard fact in TFT that, for a TFT $F$, the partition function of the mapping torus $M_\varphi$ of a $\xi$-structure automorphism $\vp\colon X\to X$ is the trace of the map $F(\vp)$ induced by the mapping cylinder of $\vp$, thought of as a bordism from $X$ to $X$. Here $F = \alpha$, $\vp$ is the identity on $X$ but action by a group element on the principal $G$-bundle, so the mapping cylinder is $W_{X,g}$ and the mapping torus is $X\times S_{\mathit{nb}}^1$, with a principal $G$-bundle whose monodromy around the $S_{\mathit{nb}}^1$ factor is $g$.
\end{proof}

The first step of the proof of \cref{lem_main_lemma} is to immediately reduce to $(B\Z/p^k, 0, \omega)$.
\begin{lem}\label{is_cyclic}
A character $G^\vee$ is determined by its restriction to all cyclic subgroups of $G$.
\end{lem}

\begin{proof}
Since any element $g\in G$ of a finite group $G$ must have finite order, we can immediately conclude.
\end{proof}

So we may now assume $G\cong\Z/p^k$ for some $p$ and $k$, and that $s = 0$. This immediately brings us back to the familiar realm covered in e.g.~\cite{CO2,Cheng:2024awi,Hsi18,Wan:2025lad,Garcia-Etxebarria:2018ajm}. In this case, our options for $\omega$ are limited: unless $p = 2$, $H^2(B\Z/p^k;\Z/2) = 0$, and $H^2(B\Z/2^k;\Z/2)\cong\Z/2$; we will let $y$ denote the nonzero element. With these restrictions in hand, both of the characters in \cref{lem_main_lemma} can be described using the Atiyah--Hirzebruch spectral sequence computing $\mho_{\Spin}^*(BG, s, \omega)$.



\begin{itemize}
    \item Since $y$ equals $w_2$ of the standard representation of $\Z/p^k$ on $\C$ as the $(p^k)^{\mathrm{th}}$ roots of unity, \cref{pip_defns_agree} implies that \cref{notpip}, i.e.\ projection onto the line $q = 3$ in the AHSS, equals the $p+ip$ layer for these choices of $G$ and $\omega$.
    \item We want to compute the value of an IFT $\alpha$ on $C\times \mathrm{K3}$ for certain closed manifolds $C$ with $(BG, 0, \omega)$-twisted spin structures. As explained in TODO, this is the action of $v\coloneqq[\mathrm{K3}]\in\Omega_4^\Spin$ on $\mho_\Spin^*(BG, s, \omega)$ (ultimately coming from the $\MTSpin$-module structure on the Thom spectrum for $(BG, s, \omega)$-twisted spin bordism). The Atiyah--Hirzebruch spectral sequence is a spectral sequence of modules over $\Omega_*^\Spin$, so we can track the action of $v$ there too.
\end{itemize}
To illustrate this, we will first walk through the easiest remaining case, i.e.\ $p\ge 5$. For these $p$, and any $k$, the Atiyah--Hirzebruch spectral sequence for $\mho_\Spin^*$ collapses in degrees $7$ and below, with neither differentials nor hidden extensions, by work of Brown--Peterson~\cite{BP66} (see~\cite[\S 10.5]{DDHM24}).
\begin{prop}\label{p_is_5}
\Cref{lem_main_lemma} is true for $G = \Z/p^k$ if $p$ is a prime number greater than $3$.
\end{prop}
\begin{proof}
As noted above, $s$ and $\omega$ must both vanish.
We draw the Atiyah--Hirzebruch spectral sequence for $\mho_\Spin^*(B\Z/p^k)$; in this figure, $y$ represents the generator of $H^*(B\Z/p^k;\Z)\cong\Z[y]/(py)$, $\abs y = 2$, and $a\in\mho^3_\Spin\cong\Z$ is either of the generators, such as the one believed to represent the low-energy limit of the $p+ip$ insulator. Thus $a$ is dual to $v\in\Omega_4^\Spin$, meaning that $va$ is a generator of $\mho_\Spin^{-1}\cong\Z$. For degree reasons, $v^2a = 0$.

\begin{sseqdata}[name=largeprime, cohomological Serre grading, xrange={0}{6}, yrange={-1}{3}, scale=0.7, classes = {draw=none}]
\class["va"](0, -1)
\class["vay"](2, -1)
\class["vay^2"](4, -1)
\class["vay^3"](6, -1)
\class["a"](0, 3)
\class["ay"](2, 3)
\class["ay^2"](4, 3)
\class["ay^3"](6, 3)
\class["{\eta^2 a}"](0, 1)
\class["{\eta a}"](0, 2)
\circleclasses[rounded rectangle, gray!70!white](2, 3)(6, -1)
\end{sseqdata}
\begin{equation}\label{p5_AHSS}
\begin{gathered}
\printpage[name=largeprime, page=2]
\end{gathered}
\end{equation}
Since the spectral sequence collapses without differentials or extension problems, $\mho_\Spin^5(B\Z/p^k)\cong \Z/p^k\cdot ay \oplus \Z/p^k\cdot vay^3$~\cite[(2.22)]{Hsi18} and $\mho_\Spin^1(B\Z/p^k)\cong\Z/2\cdot\eta^2a \oplus \Z/p^k\cdot vay$. In particular, if $\alpha = \lambda_1 ay + \lambda_2 vay^3$, with $\lambda_1,\lambda_2\in\Z/p^k$, then $\alpha_{p+ip}$ is the component of $\alpha$ in $E_2^{2,3}$, namely $\lambda_1 ay$. And $v\alpha = \lambda_1 vay$. The IFT $\lambda_1 vay$ is in the image of $H^*\to\mho^*$, the ``in-cohomology'' IFTs (since this image is exactly the line $q = -1$), so its value on a closed $1$-manifold $C$ with principal $\Z/p^k$-bundle $P$ is $\int_C \lambda_1 \widetilde y(P)$, where $\widetilde y\in H^1(B\Z/p^k;\C^\times)$ is the preimage of $y$ under the Bockstein isomorphism $H^1(B\Z/p^k;\C^\times)\to H^2(B\Z/p^k;\Z)$. Thus, unwinding the definitions, $\alpha(\mathrm{K3}\times S^1, P_1)$, where $P_1$ is the bundle whose monodromy around the circle is a generator of $\Z/p^k$ is a $p^k$th root of unity to the $\lambda_1^{\mathrm{th}}$ power.

Thus, by inspection, $\alpha_{p+ip}\ne 0$ iff $\chi_{\alpha,\mathrm{K3}}\ne 0$ iff $\lambda_1 \ne 0$.
\end{proof}
The case $G = \Z/3^k$ is only a little different.
\begin{prop}\label{p_is_3}
\Cref{lem_main_lemma} is true for $G = \Z/3^k$.
\end{prop}
\begin{proof}
A lot of things are the same as for \cref{p_is_5} -- $s$ and $\omega$ are still $0$, and the Atiyah--Hirzebruch spectral sequence still lacks differentials for degree reasons, and the $E_\infty$-page is~\eqref{p5_AHSS}. However, this time there is a hidden extension -- $\mho_\Spin^5(B\Z/3^k)\cong\Z/3^{k+1}\oplus\Z/3^{k-1}$~\cite[(2.22)]{Hsi18}. Thus, we can choose $\beta_1,\beta_2\in\mho_\Spin^5(B\Z/3^k)$, such that $\beta_1$ generates a $\Z/3^{k+1}$ summand and $\beta_2$ generates a complementary $\Z/3^{k-1}$ summand, such that the images of $\beta_1$ and $\beta_2$ in the $E_\infty$-page are, respectively, $ay$ and $vay^3$. (If $k = 1$, we can set $\beta_2 = 0$.) In degree $1$, there is no extension issue: $\mho_\Spin^1(B\Z/3^k)\cong\Z/2\cdot \eta^2a \oplus \Z/3^k\cdot vay$.

Onward as before: a general $\alpha\in\mho_\Spin^5(B\Z/3^k)$ is of the form $\lambda_1 \beta_1 + \lambda_2 \beta_2$, where $\lambda_1\in\Z/3^{k+1}$ and $\lambda_2\in\Z/3^{k-1}$. The $p+ip$ layer is $\lambda_1 ay$, which is nonzero iff $\lambda_1\not\in 3^k\Z/3^{k+1}$, and $v\alpha = \lambda_1 vay$, which is also nonzero iff $\lambda_1\not\in 3^k\Z/3^{k+1}$. That these two conditions match is the heart of the proof -- the rest is exactly the same as in the proof of \cref{p_is_5}.
\end{proof}
\begin{prop}\label{p_is_2_untwisted}
\Cref{lem_main_lemma} is true for $(G,s,\omega) = (\Z/2^k, 0, 0)$.
\end{prop}
\begin{proof}
If $k = 1$, $\mho_\Spin^5(B\Z/2)\cong 0$, because $\Omega_5^\Spin(B\Z/2)\cong 0$~\cite{MM76}, and so the result is vacuously true. Thus, in the rest of the proof, we assume $k>1$. Now the Atiyah--Hirzebruch spectral sequence is nonzero on the lines $q = 1,2$, containing a copy of $H^*(B\Z/2^k;\Z/2)\cong\Z/2[\overline x, \overline y]/(\overline x{}^2)$, with $\abs{\overline x} =1$, $\abs{\overline y} = 2$, and $y\bmod 2 = \overline y$:
\begin{sseqdata}[name=evenprime, cohomological Serre grading, xrange={0}{6}, yrange={-1}{3}, yscale=0.7, xscale=1.2, classes = {draw=none}]
\class["va"](0, -1)
\class["vay"](2, -1)
\class["vay^2"](4, -1)
\class["vay^3"](6, -1)
\class["a"](0, 3)
\class["ay"](2, 3)
\class["ay^2"](4, 3)
\class["ay^3"](6, 3)
\class["{\eta^2 a}"](0, 1)
\class["{\eta^2 a\overline x}"](1, 1)
\class["{\eta^2 a\overline y}"](2, 1)
\class["{\eta^2 a\overline x\overline y}"](3, 1)
\class["{\eta^2 a\overline y{}^2}"](4, 1)
\class["{\eta^2 a\overline x\overline y{}^2}"](5, 1)
\class["{\eta^2 a\overline y{}^3}"](6, 1)
\class["{\eta a}"](0, 2)
\class["{\eta a\overline x}"](1, 2)
\class["{\eta a\overline y}"](2, 2)
\class["{\eta a\overline x\overline y}"](3, 2)
\class["{\eta a\overline y{}^2}"](4, 2)
\class["{\eta a\overline x\overline y{}^2}"](5, 2)
\class["{\eta a\overline y{}^3}"](6, 2)

\circleclasses[rounded rectangle, gray!70!white](2, 3)(6, -1)
\end{sseqdata}
\begin{equation}\label{p2_AHSS}
\begin{gathered}
\printpage[name=evenprime, page=2]
\end{gathered}
\end{equation}
Hsieh~\cite[(2.22)]{Hsi18} showed $\mho_\Spin^5(B\Z/2^k)\cong\Z/2^k\oplus\Z/2^{k-2}$ -- there are $2^{2k-2}$ elements. On the $E_2$-page in total degree $5$, we have $\Z/2^k\oplus\Z/2\oplus\Z/2\oplus \Z/2^k$ -- $2^{2k+2}$ classes. Thus all four groups in total degree $5$, which are $E_2^{2,3}$, $E_2^{3,2}$, $E_2^{4,1}$, and $E_2^{6,-1}$, must support a differential. This wipes out $E_2^{3,2}$ and $E_2^{4,1}$, quotients $E_2^{2,3}\cong\Z/2^k$ to $\Z/2^{k-1}$, and replaces $E_2^{6,-1}\cong\Z/2^k$ with its even subgroup, also isomorphic to $\Z/2^{k-1}$.

Thus, in total degree $5$ on the $E_\infty$-page, we have $ay$ generating a $\Z/2^{k-1}$ and $2vay^3$ generating another $\Z/2^{k-1}$, and there is a hidden extension between them. In degree $1$, there are once again neither differentials nor extension problems: $\mho_\Spin^1(B\Z/2^k)
\cong\Z/2\cdot \eta^2a \oplus\Z/2^k\cdot vay$. Thus the rest of the proof is the same as it was for $p =3$ in \cref{p_is_3}: there are $\beta_1,\beta_2\in\mho_\Spin^5(B\Z/2^k)$ generating the $\Z/2^k$, resp.\ $\Z/2^{k-2}$ summands, whose images in the $E_\infty$-page are $ay$, resp.\ $2vay^3$. Thus a general $\alpha\in\mho_\Spin^5(B\Z/2^k)$ is of the form $\lambda_1 \beta_1 + \lambda_2\beta_2$ with $\lambda_1\in\Z/2^k$ and $\lambda_2\in\Z/2^{k-2}$. The $p+ip$ layer of $\alpha$ is $\lambda_1 ay\in E_2^{2,3}\cong \Z/2^k$, nonzero iff $\lambda_1\not\in 2^k\Z/2^{k+1}$, and $v\alpha = \lambda_1 vay$, nonzero iff $\lambda_1\not\in 2^k\Z/2^{k+1}$, and can be detected with $H^1(\bl;\C^\times)$ as usual.
\end{proof}
\begin{prop}[{Córdova--Ohmori~\cite{CO2}}]
\label{p_is_2_twisted}
\Cref{lem_main_lemma} is true for $(G,s,\omega) = (\Z/2^k, 0, y)$.
\end{prop}
\begin{proof}
Once again draw the Atiyah--Hirzebruch spectral sequence. The $E_2$-page is isomorphic to the $E_2$-page for the case $y = 0$, as we saw in~\eqref{p2_AHSS}. (If $k = 1$, we can use the same picture, but this time $x^2 = y$ instead of $x^2 = 0$. This has no effect on the rest of the proof.) Thus there are again $2^{2k+2}$ elements of total degree $5$ on the $E_2$-page. By~\cite[(2.42), (2.43)]{Hsi18}, $\mho_\Spin^5(B\Z/2^k, 0, y)\cong\Z/2^{k+3}\oplus\Z/2^{k-1}$, which also has $2^{2k+2}$ elements, so there can be no differentials to or from classes in total degree $5$. Moreover, the orders of the summands in this group of IFTs uniquely constrain the nature of the hidden extensions: all three times that the extension could be nontrivial for degree reasons, it is nontrivial. The upshot is that we can choose $\beta_1,\beta_2\in\mho_\Spin^5(B\Z/2^k, 0, y)$ such that $\beta_1$ generates a $\Z/2^{k+3}$ summand and $\beta_2$ generates a complimentary $\Z/2^{k-1}$ summand, and such that the images of $\beta_1$ and $\beta_2$ in the $E_\infty$-page are $ay$, resp.\ $2vay^3$. (If $k = 1$, there is no second summand, so take $\beta_2 = 0$.)

This time around, $\mho_\Spin^1(BG,0,y)$ is more interesting: it is isomorphic to $\Z/2^{k+1}$, so there must be a hidden extension between $E_\infty^{0,1}\cong\Z/2$ and $E_\infty^{2,-1}\cong\Z/2^k$. In particular, $v\beta_1$ has image in the $E_\infty$-page equal to $vay$, hence is two times a generator of $\mho_\Spin^1(B\Z/2^k, 0, y)$, and has order $2^k$.

As before, let $\alpha\in\mho_\Spin^5(B\Z/2^k, 0, y)$, so that $\alpha = \lambda_1\beta_1 + \lambda_2\beta_2$, with $\lambda_1\in\Z/2^{k+3}$ and $\lambda_2\in\Z/2^{k-1}$. The $p+ip$ layer of $\alpha$ is thus $\lambda_1 ay\in E_2^{2,3}\cong\Z/2^k$, and is nonzero if and only if $\lambda_1\not\in 2^k\Z/2^{k+3}$. Similarly, $v\alpha = \lambda_1 v\beta_1$, which is nonzero if and only if $\lambda_1\not\in 2^k\Z/2^{k+3}$. The rest of the proof is essentially the same as the previous three cases; this time, we must be careful with the fact that $v\beta_1$ does not generate a direct summand of $\mho_\Spin^1(B\Z/2^k, 0, y)$. This ends up not being an issue: it \emph{does} generate the subgroup of IFTs in the image of $H^*\to\mho_\Spin^*$, and the inclusion of this image is dual to the quotient $\Omega_1^\Spin(BG, 0, y)/\Omega_1^\Spin$ we took in~\eqref{char_again}. Thus, $\lambda_1\not\in 2^k\Z/2^{k+3}$ if and only if $v\alpha\ne 0$ if and only if $v\alpha$ is nonzero on $P\to S^1$, where $P$ is the mapping torus of some element of $\Z/2^k$ acting on $\Z/2^k$ by translation.
\end{proof}
Thus, by combining \cref{is_cyclic,p_is_5,p_is_3,p_is_2_untwisted,p_is_2_twisted}, we have proven \cref{lem_main_lemma}. As we discussed above, this finishes the proof of \cref{existence_and_nonexistence}.

\section{Examples}\label{section:construction}
We now go to concrete examples and construct (3+1)d $G$-SETs with the symmetry structures given in \cref{ex:cyclic,ex:cyclicfermion,ex:timerev}. We do this by explaining how to trivialize generators of $\SH^5(BG,s,\omega)$ that parametrize the $G$-anomaly by pulling back to a larger group. This will establish \cref{thm:Znsymmetry,thm:trivializedGen,thm:timerev}. Once an element of $\SH^5(BG,s,\omega)$ has been trivialized, one can gauge a subgroup symmetry, following the Wang--Wen--Witten construction, and obtain a TQFT realizing that anomaly class. We collect all the calculations of the relevant supercohomology groups in Appendix~\ref{appendix:B} and focus on the calculation of the pullback in this section. To finish the computation, we will need to use the \textit{Smith long exact sequence}. We refer the reader to \cref{Smith_LES,LES_is_Smith} for background material and references about this long exact sequence.

To denote the relevant cohomology classes, we write down the relevant cohomology rings for the cyclic groups under consideration. The $\Z/2$ cohomology ring of $B\Z/2$ is given by
\begin{equation}
    H^*(B\Z/2; \bZ/2) = \Z/2[x], \quad |x| = 1 
\end{equation}
where $x$ is the nontrivial generator of $H^1(B\Z/2; \bZ/2)$. The $\Z/2$ cohomology ring of $B\Z/2^k,k\geq 2$ is given by
\begin{equation}
    H^*(B\Z/2^k; \bZ/2) = \Z/2[x, y]/(x^2), \quad |x| = 1, |y| = 2. 
\end{equation}
Finally, the integral cohomology ring of $B\Z/2^k$ with $k$ positive integer is given by 
\begin{equation}\label{eq:commoncoh_Z}
    H^*(B\Z/2^k;\Z) = \Z/2[\widehat y]/(2^k \cdot \widehat y),~|\widehat y|=2\,.
\end{equation}
We also need the following lemma regarding the image of the generators under the pullback.
\begin{lem}
\label{p_on_coh}
Consider the mod-$2^k$ reduction map $p\colon\Z/2^{k+1} \rightarrow\Z/2^k$.
\begin{enumerate}
    \item\label{Z2_pullback} For $k=1$, $p^*(x^2) = 0$. For $k\geq 2$, $p^*(y) = 0$. 
    \item\label{Z_pullback} $p^*(\widehat y) = 2\widehat y \in H^2(B\Z/2^{k+1};\Z)$.
\end{enumerate}
\end{lem}
\begin{proof}
Because $x^2$ or $y$ is the unique nontrivial class in $H^2(B\Z/2^k;\Z/2)$, it corresponds to a nonsplit central extension of $\Z/2^k$ by $\Z/2$, i.e.,
\begin{equation}
     1 \longrightarrow \Z/2\longrightarrow \Z/2^{k+1} \xrightarrow{p} \Z/2^k \longrightarrow 1\,.
\end{equation}
When we pull this extension back along $p$, $x^2$, resp.\ $y$ is tautologically trivialized, when $k = 1$, resp.\ $k>1$. This concludes the proof of part~\eqref{Z2_pullback}.

For part~\eqref{Z_pullback}, for a finite group $G$, consider the natural isomorphism 
\begin{equation}
H^2(BG;\Z) \cong H^1(BG;\C^\times)\,,
\end{equation}
given by the Bockstein homomorphism associated with the short exact sequence $0\to\Z\to\R\to\C^\times\to 0$. $H^1(BG;\C^\times)$ is naturally isomorphic to $\Hom(G, \C^\times)$, which is simply the abelian group for 1-dimensional complex representations of the group $G$. The generator $\widehat y$ corresponds to the representation which sends $1\in \Z/2^k$ to $\exp(2\pi i/2^k)$. We immediately see that under the pullback map $p^*(\widehat y) = 2\widehat y $.
\end{proof}

\subsection{Computations for $\Z/n\times\Z/2^F$ (\cref{ex:cyclic})}\label{subsub:cyclic}

\begin{proof}[Proof of \Cref{thm:Znsymmetry}]

Let $n= p^{k_1}_1 p^{k_2}_2 \ldots p^{k_\ell}_\ell$ be the prime factorization of $n$. Then the map $r_j\colon B\Z/n\to B\Z/p_j^{k_j}$ induced by the mod $p_j^{k_j}$ reduction map $\Z/n\to \Z/p_j^{k_j}$ is a homotopy equivalence after localizing at $p_j$, so all $p_j$-primary torsion in $\SH^5(B\Z/n)$ is in the image of $r_j^*$. Thus, it suffices to consider the case when $n = p^k$, as trivializations in these cases induce trivializations for all $n$.

When $n=p^k$ is odd, the AHSS is only nontrivial in the Dijkgraaf--Witten layer, because the $\Z/2$-valued cohomology of $B\Z/n$ vanishes in positive degrees. Therefore, there is a canonical isomorphism $\SH^5(B\Z/n)\cong H^5(B\Z/n;\C^\times)$, which is cannonically isomorphic to $\Z/n$. 
The generator of $H^5(B\Z/n;\C^\times)$ is the image of $x y^2\in H^5(B\Z/p^{k};\Z/p)$ under the exponential map $\Z/p \to\C^\times$ sending $\ell \mapsto \exp(2\pi i\ell/p)$, 
where $x \in H^1(B\Z/p^{k};\Z/p)\cong\Z/p$ and $y \in H^2(B\Z/p^{k};\Z/p)\cong\Z/p$ are the standard generators.
To trivialize the generator, consider the short exact sequence
\begin{equation}
    1 \rightarrow \Z/p \rightarrow \Z/p^{k+1} \rightarrow \Z/p^k\rightarrow 1\,.
\end{equation}
Similar to the proof of Part~\ref{Z2_pullback} of \cref{p_on_coh}, when we pull back to $H^5(B\Z/p^{k+1};\Z/p)$, $y$ trivializes, so $x y^2$ pulls back to $0$ as well. Therefore, we prove that we can use a $\Z/p$ gauge theory to realize the given anomalies. 

When $n = 2$, \Cref{prop:Z2notwist} shows that  $\SH^5(B\Z/2) = 0$. 

In the case when $G= \Z/2^k$ for $k\geq 2$, 
\Cref{prop:2knotwist} shows that $\SH^5(B\Z/2^k) = \Z/2^{k-1}$, and the generator is in the Dijkgraaf--Witten layer. This means that we can express it as a class in $H^5(B\Z/2^{k};\mathbb{C}^\times)\cong H^6(B\Z/2^{k};\mathbb{Z})$.  
We pull back the generator of $H^6(B\Z/2^{k};\Z)$ to $H^6(B\Z/2^{k+m};\Z)$ along the sequence 
\begin{equation}\label{eq:seq3repeat}
         1 \longrightarrow \Z/2^{m}\longrightarrow \Z/2^{k+m} \longrightarrow \Z/2^k \longrightarrow 1\,.
\end{equation}
From part~\ref{Z_pullback} of \cref{p_on_coh}, it is immediate to see that the generator ${\widehat y}^3 \in H^6(B\Z/2^k;\Z)$ pulls back to $8^{m}{\widehat y}^3 \in H^6(B\Z/2^{k+m};\Z)$, and it is trivialized in $SH^5(B\Z/2^{k+m})$ if $8^m {\widehat y}^3 \geq 2^{k+m-1}$, i.e., $m\geq \frac{k - 1}{2}$. Here, the extra 1 on the left-hand side of the inequality comes from the difference between supercohomology and regular cohomology. Therefore, we prove that we can use a $\Z/2^m$ gauge theory to realize the given anomalies. This establishes Theorem \ref{thm:Znsymmetry}.
\end{proof}

\subsection{Computations for $\Z/(2n)^F$ (\cref{ex:cyclicfermion})}\label{subsub:cyclicfermion}
Since the group structure in this example has the unitary $\Z/n$ symmetry mixing with fermion parity, the computations involve twisted supercohomology where the twist arises from a degree 2 cohomology class in $H^2(B\Z/n;\Z/2)$. In the case where $n$ is odd, the twist is trivial, and we can apply the same computation as those in Example \ref{ex:cyclic}. 

We first treat the case when  $n=2$, so that 
\begin{equation}
    g^2 = (-1)^F\,.
\end{equation}
This corresponds to giving spacetime a $G$-structure where $G = \Spin\times_{\set{\pm 1}}\Z/4$~\cite{Hason:2020yqf}, which is equivalent to a $(B\Z/2, 0, x^2)$-twisted spin structure: see~\cite[\S 7.8]{Cam17} and~\cite[Example 6.23]{debray2024smith}.
When $n = 2^k$, the story is similar: we obtain a $\Spin\times_{\set{\pm 1}}\Z/2^{k+1}$-structure, which is equivalent to a $(B\Z/2^k, 0, y)$-twisted spin structure: see~\cite{Hsi18} and~\cite[Example 6.23]{debray2024smith}.
These pass to the corresponding twists of supercohomology as described in Appendix \ref{subsection:twistedSH}. The obstructions associated to these symmetry structures are captured by $\SH^5(B\Z/2,0,x^2) \cong \Z/8$ and $\SH^5(B\Z/2^k,0,y) \cong \Z/2^{k+1} \oplus \Z/2$, respectively, which we calculated in \Cref{prop:Z2x2twist,prop:2kytwist}.

The easiest way to analyze the pullback for $n=2$ requires some knowledge of the Smith long exact sequence, which we present in detail below.

\begin{proof}[Proof of Part~\ref{exp2_2} of \Cref{thm:trivializedGen}]
To trivialize a generator for $\SH^5(B\Z/2,0,x^2)\cong\Z/8$, we first consider the symmetry extension sequence given by
\begin{equation}\label{eq:Z4Z2}
     1 \longrightarrow \Z/2\longrightarrow \Z/4 \xrightarrow{p} \Z/2 \longrightarrow 1\,.
\end{equation}
According to part~\ref{Z2_pullback} of \cref{p_on_coh}, $p^*(x^2)$ is zero. Thus the twist $(0, x^2)$ over $B\Z/2$ pulls back to the trivial twist $(0, 0)$ over $B\Z/4$. Therefore the map $p$ in~\cref{eq:Z4Z2} induces a pullback map $\SH^5(B\Z/2,0,x^2)\to\SH^5(B\Z/4,0,0)$. We now show that this map sends $1\mapsto 1$ and trivializes the $\Z/4$ subgroup $2\Z/8\subset \Z/8$. 

To perform this computation in supercohomology, we will study two associated Smith long exact sequences. The goal is to analyze this pullback by reducing it to a lower-degree pullback that is easier to study and already known in the literature. 
\begin{subequations}\label{bothSmith}
\begin{gather}\label{eq:smith1}
\begin{tikzcd}[ampersand replacement=\&, column sep=0.4cm]
    \ldots \arrow[r]\& \SH^4(B\Z/2) \arrow[r]\& \SH^4(B\Z/4,x,0)\arrow[r,"\mathrm{sm}_\sigma"] \&\SH^5(B\Z/4)\arrow[r]\&
    \SH^5(B\Z/2) \arrow[r]\&\ldots
    \end{tikzcd}\\
 \label{eq:smith2}
  \begin{tikzcd}[ampersand replacement=\&, column sep=0.4cm]
    \ldots \arrow[r]\& \SH^4(\mathrm{pt}) \arrow[r]\& \SH^4(B\Z/2,x,x^2)\arrow[r,"\mathrm{sm}_\sigma"] \& \SH^5(B\Z/2,0,x^2)\arrow[r]\&
    \SH^5(\mathrm{pt}) \arrow[r]\&\ldots\,
      \end{tikzcd}
 \end{gather}
 \end{subequations}
The first one is constructed in~\cite[(A.29)]{Debray:2023iwf},\footnote{$(B\Z/4, x, 0)$-twisted spin bordism, sometimes called \term{epin bordism}, is also studied in~\cite{BG97, BY99, BY00, WWZ20}. See also~\cite{CHZ24} for a closely related symmetry type in a physics application.} and the second is constructed in~\cite[Theorem 3.1]{Gia73}, where it is attributed to Stong.\footnote{This long exact sequence is studied more systematically, as part of a family of related Smith long exact sequences, in~\cite{HS13, KTTW15, TY19, Hason:2020yqf, WWZ20, BR23, debray2024smith}.} According to \cref{prop:Z2notwist} and degree considerations, we have $\SH^\ell(B\Z/2, 0, 0) = 0$ and $\SH_\ell(\pt) = 0$ for $\ell = 4,5$. Therefore, the horizontal Smith homomorphisms in~\eqref{eq:smithsquareko} are all isomorphisms. 

Now consider the commuting square constructed from the two Smith homomorphisms:

\begin{equation}\label{eq:smithsquareko}
    \begin{tikzcd}[column sep=1cm, row sep=1cm]
\underbracket[0.15ex]{\SH^4(B\Z/4,x,0)}_{\cong\Z/2} \arrow[r,"\mathrm{sm}_\sigma"] \arrow[dr, phantom, very near start] & \underbracket[0.15ex]{\SH^4(B\Z/4,0,0)}_{\cong\Z/2~\eqref{prop:2knotwist}} \\
\underbracket[0.15ex]{\SH^4(B\Z/2,x,x^2)}_{\cong\Z/8} \arrow[u] \arrow[r, "\mathrm{sm}_\sigma"] & \underbracket[0.15ex]{\SH^5(B\Z/2,0,x^2)}_{\cong\Z/8~\eqref{prop:Z2x2twist}} \arrow[u] \,.
\end{tikzcd}
\end{equation}
The left vertical map can be analyzed by going back to the spin bordism. Consider the following map of short exact sequences, where the middle map is determined in \cite[Proposition A.35 (1)]{Debray:2023iwf}, and the subscript $x$ of $\C^\times$ indicates the nontrivial twist of cohomology.
{\footnotesize
    \begin{gather}
\begin{tikzcd}[ampersand replacement=\&]
	0 \& {\SH^4(B\Z/2,x,x^2) \cong \Z/8} \& {\mho^4(B\Z/2,x,x^2)\cong\Z/16} \& {H^1(B\Z/2,\C_x^\times) \cong \Z/2} \& 0 \\
	0 \& {\SH^4(B\Z/4,x,0) \cong \Z/2} \& {\mho^4(B\Z/4,x,x^2)\cong\Z/4} \& {H^1(B\Z/2,\C_x^\times) \cong \Z/2} \& 0
	\arrow[from=1-1, to=1-2]
	\arrow["{1\mapsto 2}", from=1-2, to=1-3]
	\arrow["{p^*}"', from=1-2, to=2-2]
	\arrow["{\bmod 2}", from=1-3, to=1-4]
	\arrow["{p^*}", from=1-3, to=2-3]
	\arrow[from=1-4, to=1-5]
	\arrow["{1\mapsto 1}", from=1-4, to=2-4]
	\arrow[from=2-1, to=2-2]
	\arrow["{1\mapsto 2}", from=2-2, to=2-3]
	\arrow["{\bmod 2}", from=2-3, to=2-4]
	\arrow[from=2-4, to=2-5]
\end{tikzcd}
\end{gather}}
The commutativity of the diagram then forces $p^*\colon \SH^4(B\Z/2,x,x^2)\rightarrow \SH^4(B\Z/4,x,0)$ to be a projection map sending $1\mapsto 1$. As a result, the right vertical map, which is the map we want to analyze, also sends $1\mapsto 1$.

Therefore, extending $G=\Z/2$ by $\Z/2$ is not enough to trivialize the whole $\SH^5(B\Z/2,0,x^2)\cong\Z/8$. Still, using the result of \cref{thm:Znsymmetry}, the generator can be trivialized by considering a further extension
     \begin{equation}
      1 \longrightarrow \Z/2\longrightarrow \Z/8 \longrightarrow \Z/4 \longrightarrow 1\,.
 \end{equation}
 Therefore, we prove that we can use a $\Z/4$ gauge theory to realize the given anomalies. 
\end{proof}

Now we go to $n=2^k,k\geq 2$, which also requires the construction of Smith long exact sequences that are slightly more involved.

\begin{proof}[Proof of Part~\ref{its_m} of \Cref{thm:trivializedGen}]

The last case to study in this example is $n = 2^k$ with $k \geq 2$. We first consider the symmetry extension sequence given by
\begin{equation}
     1 \longrightarrow \Z/2^{}\longrightarrow \Z/2^{k+1} \xrightarrow{p} \Z/2^k \longrightarrow 1\,,
\end{equation} 
which induces the pullback
\begin{equation}
    p^*\colon \SH^5(B\Z/2^{k+1}, 0, 0) \longrightarrow \SH^5(B\Z/2^k, 0, y)\,.
\end{equation}

We will use two Smith long exact sequences in this proof, again in order to analyze this pullback using another pullback in lower degrees. First take \Cref{Smith_LES} with $E_* = \tau_{\leq 2}\ko_*$, $X = B\Z/2^k$, and $V=W=V_\rho$. There is a homotopy equivalence $S(V_\rho)\simeq S^1$, which stably splits as $\mathbb S\vee\Sigma\mathbb S$~\cite[Example 7.28]{debray2024smith}. Thus we have a long exact sequence

\begin{equation}\label{eq:smithwithV}
     \ldots \rightarrow (\tau_{\leq 2}\ko)_n(\mathbb{S}\vee \Sigma \mathbb S) \rightarrow (\tau_{\leq 2}\ko)_n((B\Z/2^k)^{V_\rho}) \xrightarrow{\mathrm{sm}_{V_\rho}} (\tau_{\leq 2}\ko)_{n-2}(B\Z/2^k)\rightarrow \ldots.
\end{equation}

\begin{lem}
\label{pullback_rho}
Under the map $p\colon B\Z/2^{k+1}\rightarrow B\Z/2^k$, the bundle $V_\rho \to B\Z/2^k$ pulls back to $V_\rho\otimes V_\rho \rightarrow B\Z/2^{k+1}$.
\end{lem}
\begin{proof}
It suffices to show this at the level of representations of $\Z/2^{k+1}$; since this is a cyclic group, it suffices to check on a generator. Specifically, for both $p^*(V_\rho)$ and $V_\rho\otimes V_\rho$, it is straightforward to see that $1\in\Z/2^{k+1}$ acts by $e^{2\pi i/k}$.
\end{proof}

The other Smith long exact sequence we need uses $E = \tau_{\le 2}\ko$ again; this time $X = B\Z/2^{k+1}$, $V = 0$, and $W = V_\rho\otimes V_\rho$:
\begin{equation}\label{eq:smithVV}
    \begin{tikzcd}
    \dotsb \rightarrow (\tau_{\leq 2} \ko)_n (S(V_\rho\otimes V_\rho))  \rightarrow (\tau_{\leq 2} \ko)_n (B\Z/2^{k+1}) \xrightarrow{\text{sm}_{V_\rho\otimes V_\rho}} (\tau_{\leq 2} \ko)_{n-2} (B\Z/2^{k+1}) \rightarrow \dotsb\,,
\end{tikzcd}
\end{equation}
A priori the rightmost term in~\eqref{eq:smithVV} is a $(B\Z/2^{k+1}, V_\rho\otimes V_\rho)$-twisted $\tau_{\le 2}\ko$-homology group, but $V_\rho\otimes V_\rho$ has a canonical spin structure, so we obtain untwisted $\tau_{\le 2}\ko$-homology. In a little more detail, a spin structure on a complex line bundle is equivalent to a choice of square root with respect to tensor product~\cite{Ati71}, and for $V_\rho\otimes V_\rho$, we have the square root $V_\rho$.

The sphere bundle $S(V_\rho\otimes V_\rho)$ fits in the  following diagram, where both squares are pullback squares:
\begin{equation}\label{eq:doublepullback}
    \begin{tikzcd}
        S(V_\rho\otimes V_\rho) \arrow[r] \arrow[d] & S(V_\rho) \arrow[r] \arrow[d] &E \mathbb{C}^\times \arrow[d] \\
        B\Z/2^{k+1} \arrow[r] & B\Z/2^k \arrow[r] & B\mathbb{C}^\times\,.
    \end{tikzcd}
\end{equation}
We identified $S(V_\rho)\simeq S^1 = B\Z$ above and need to compute $S(V\otimes V)$.
\begin{lem}
There is a homotopy equivalence $S(V_\rho\otimes V_\rho)\simeq B\Z\times B\Z/2$ under which
\begin{enumerate}
    \item the map $S(V_\rho\otimes V_\rho)\to B\Z/2^{k+1}$ in~\eqref{eq:doublepullback} is $B$ of the map $\Z\times \Z/2\to \Z/2^{k+1}$ sending $(c, d)\mapsto c + 2^kd$, and
    \item the map $S(V_\rho\otimes V_\rho)\to S(V_\rho)$ is identified with the map $B\Z\times B\Z/2\to B\Z$ which is projection onto the first factor.
\end{enumerate}
\end{lem}
\begin{proof}
Conveniently, $S(V_\rho\otimes V_\rho)$ is the pullback of the diagram $B\Z/2^{k+1}\to B\Z/2^k\gets B\Z$, which is the result of applying the classifying space functor to the following diagram of groups:
\begin{equation}\label{gp_pull}
\begin{tikzcd}
	& \Z \\
	{\Z/2^{k+1}} & {\Z/2^k}
	\arrow["{\bmod 2^k}", from=1-2, to=2-2]
	\arrow["{\bmod 2^k}"', from=2-1, to=2-2]
\end{tikzcd}
\end{equation}
The bar construction model for the classifying space functor preserves pullbacks, so $S(V_\rho)$ is homotopy equivalent to the classifying space of the group which is the pullback of~\eqref{gp_pull}. In the category of groups, there is an explicit formula for the pullback of the diagram $H\overset f\to G\overset g\gets K$~\cite[\href{https://stacks.math.columbia.edu/tag/0020}{Tag 0020}]{stacks-project}, namely
\begin{equation}
    H\times_G K\cong \set{(h, k)\in H\times K: f(h) = g(k)}.
\end{equation}
The maps to $H$ and $K$ are projection onto the first, resp., second factor.

Applying this to~\eqref{gp_pull}, we see that the pullback group is $\Z\times\Z/2$, with the map to $\Z/2^{k+1}$ sending $(c,d)\mapsto c + 2^k d$ and the map to $\Z$ sending $(c,d)\mapsto c$. Applying the classifying space functor, we have $S(V_\rho\otimes V_\rho)\cong B\Z\times B\Z/2$ as well as the maps to $S(V_\rho)$ and to $B\Z/2^{k+1}$.
\end{proof}

The Smith long exact sequence~\eqref{eq:LESsmith} is by construction natural in the data $X$, $V$, and $W$, so from~\eqref{eq:doublepullback} we obtain the following commutative diagram, whose rows are exact.

{\footnotesize \begin{equation*}\label{map_Smith_LES}
    \begin{tikzcd}[column sep=0.4cm]
      {(\tau_{\leq 2}\ko)_5(B\Z \times B\Z/2)} \arrow[r] \arrow[d]& {(\tau_{\leq 2} \ko)_5(B\Z/2^{k+1},0,0)} \arrow[r,"\mathrm{sm}_{V_\rho \otimes V_\rho}"] \arrow[d]& {(\tau_{\leq 2}  \ko)_3(B\Z/2^{k+1},0,0)} \arrow[r] \arrow[d]& {(\tau_{\leq 2}\ko)_4(B\Z \times B\Z/2)}\arrow[d] \\
      {(\tau_{\leq 2}\ko)_5(\mathbb{S} \vee 
       \Sigma \mathbb{S})} \arrow[r] &  {(\tau_{\leq 2} \ko)_5(B\Z/2^k,0,y)} \arrow[r,"\mathrm{sm}_{V_\rho}"] & {(\tau_{\leq 2}  \ko)_3(B\Z/2^{k},0,0)} \arrow[r] & {(\tau_{\leq 2}\ko)_4(\mathbb{S} \vee 
       \Sigma \mathbb{S})}\,.
    \end{tikzcd}
\end{equation*}}
Applying the $I_{\C^\times}$, we have the corresponding sequence in supercohomology, and cross-references indicate that we have already determined some of the entries in this diagram.
{\small \begin{equation}\label{map_Smith_LES_dual}
    \begin{tikzcd}[column sep=0.4cm]
      \underbracket[0.15ex]{\SH^4(B\Z \times B\Z/2)}_{\Z/8} \arrow[r] & \underbracket[0.15ex]{\SH^3(B\Z/2^{k+1},0,0)}_{\cong\Z/2^{k+2}\oplus\Z/2~\eqref{prop:Zknotwist3}} \arrow[r,"\mathrm{sm}_{V_\rho \otimes V_\rho}"] & \underbracket[0.15ex]{\SH^5(B\Z/2^{k+1},0,0)}_{\cong\Z/2^k~\eqref{prop:2knotwist}} \arrow[r] & {\SH^5(B\Z \times B\Z/2)} \\
      {\SH^4(\mathbb{S} \vee 
       \Sigma \mathbb{S})} \arrow[u] \arrow[r] &  \underbracket[0.15ex]{\SH^3(B\Z/2^k,0,0)}_{\cong\Z/2^{k+1}\oplus\Z/2~\eqref{prop:Zknotwist3}} \arrow[u, "p^*"] \arrow[r,"\mathrm{sm}_{V_\rho}"] & \underbracket[0.15ex]{\SH^5(B\Z/2^{k},0,y)}_{\cong\Z/2^{k+1}\oplus\Z/2~\eqref{prop:2kytwist}} \arrow[u, "p^*"] \arrow[r] & {\SH^5(\mathbb{S} \vee 
       \Sigma \mathbb{S})}\arrow[u] \,.
    \end{tikzcd}
\end{equation}}

Thus we would like to find $\SH^\ell(X)$ for $\ell = 4,5$ and $X = \Sph\vee\Sigma\Sph$ and $B\Z\times B\Z/2$. A straightforward Atiyah--Hirzebruch spectral sequence calculation shows $\SH^\ell(\mathbb{S} \vee \Sigma \mathbb{S})\cong 0$ whenever $\ell\ge 4$. From the standard Adams spectral sequence, we see that $\SH^4(B\Z\times B\Z/2) = \Z/8$ and $\SH^5(B\Z\times B\Z/2)=0$. Therefore, from the exactness of both rows, $\mathrm{sm}_{V_\rho\otimes V_\rho}$ is surjective and $\mathrm{sm}_{V_\rho}$ is an isomorphism. Moreover, choosing a fixed isomorphism $\SH^3(B\Z/2^k,0,0)\cong\Z/2^{k+1}\oplus \Z/2$ for all $k$, we can choose the isomorphism $\SH^5(B\Z/2^k,0,y)\cong \Z/2^{k+1}\oplus \Z/2$ such that $\mathrm{sm}_{V_\rho}$ sends 
\begin{equation}
(1,0)\mapsto (1,0)\,,\quad (0,1)\mapsto(0,1)\,.
\end{equation}
Similarly, we can choose the isomorphism $\SH^5(B\Z/2^{k+1},0,0)\cong \Z/2^k$ such that $\mathrm{sm}_{V_\rho\otimes V_\rho}$ sends
\begin{equation}
(1,0)\mapsto 1\,,\quad (0,1)\mapsto 2^{k-1}\,.
\end{equation}

Now consider the left vertical $p^*$.
\begin{lem}
\label{h_in_C11}
There are choices of the isomorphisms $\SH^3(B\Z/2^k)\cong \Z/2^{k+1}\oplus\Z/2$ from \Cref{prop:Zknotwist3} such that, with respect to those isomorphisms, the map $p^*\colon\SH^3(B\Z/2^k)\to \SH^3(B\Z/2^{k+1})$ in~\eqref{map_Smith_LES_dual} is given by the matrix $\begin{bsmallmatrix}4&0\\0&1\end{bsmallmatrix}$.
\end{lem}

\begin{proof}[Proof of \cref{h_in_C11}]
We begin by calculating the effect of $p^*$ on the $E_\infty$-page of the supercohomology AHSSes. We computed these $E_\infty$-pages in~\eqref{eq:E3SH_Z4}. 
Let ${}^{k}\! E_\infty^{p,q}$ denote the $E_\infty^{p,q}$ entry of the AHSS computing $\SH^*(B\Z/2^k)$, $^k\! E_\infty^{\bullet, 3-\bullet}$ consists of the following three summands:
\begin{itemize}
    \item $^k\! E_\infty^{3,0}\cong H^3(B\Z/2^k;\C^\times)\cong\Z/2^k$,
    \item $^k\! E_\infty^{2,1}\cong H^2(B\Z/2^k;\Z/2)\cong\/\Z/2$ generated by $y$, and
    \item $^k\! E_\infty^{1,2}\cong H^1(B\Z/2^k;\Z/2)\cong\/Z/2$ generated by $x$.
\end{itemize}
It is straightforward to check that in mod $2$ cohomology, $p$ pulls back $y\mapsto 0$ and $x\mapsto x$. Thus $p^*\colon ^k\!E_\infty^{3-j,j}\to ^{k+1}\!E_\infty^{3-j,j}$ is zero for $j=1$ and an isomorphism for $j = 2$. On $E_\infty^{3,0}$, \cref{p_on_coh}, part~\eqref{Z_pullback} and the isomorphism $H^3(B\Z/2^k;\C^\times)\cong H^4(B\Z/2^k;\Z)$ shows that $p^*$ gives $1\mapsto 4$.

To finish, we need to lift from the $E_\infty$-page to the actual supercohomology groups. Before this, we have an extension problem to resolve for $\SH^3(B\Z/2^k)$, where $k\ge 2$: $\Z/2^k$ in the Dijkgraaf--Witten layer, $\Z/2$ in the Gu--Wen layer, and $\Z/2$ in the Majorana layer combine to $\Z/2^{k+1}\oplus\Z/2$ (\cref{prop:Zknotwist3}). Thus, we have a nonsplit extension of either the Gu--Wen layer or the Majorana layer by the Dijkgraaf--Witten layer. In fact, the extension is between the DW and GW layers; to see this, first note that this is equivalent to the corresponding extension in \emph{restricted} supercohomology $\mathit{rSH}^3(B\Z/2^k)$ being nonsplit, just as in the proof of \cref{bonus_AHSS_lemma}. Gu--Wen~\cite[(F12)]{PhysRevB.90.115141} computed $\mathit{rSH}^3(B\Z/2^k)\cong\Z/2^{k+1}$ for $k\ge 2$, implying a nonsplit extension in restricted supercohomology, and therefore an extension between the GW and DW layers in supercohomology.

Since $p^*$ is an isomorphism on the Majorana layer, and the Majorana layer splits off for all $k\ge 2$, choose any splitting of the Majorana layer off of the GW and DW layers for $k = 2$; for $k>2$, inductively choose the splitting that makes the pullback map $p^*$ diagonal. Thus we have chosen isomorphisms $\SH^3(B\Z/2^k)\cong \mathit{rSH}^3(B\Z/2^k)\oplus\Z/2$ such that $p^*$ is a diagonal matrix whose $(2, 2)$ entry is $1$ and whose $(1, 1)$ entry is to be determined.

We have a commutative diagram of short exact sequences

\begin{equation}\label{one_more_xtn}
\begin{tikzcd}
	0 & {^kE_\infty^{3,0} \cong \Z/2^k} & {\mathit{rSH}^3(B\Z/2^k)\cong\Z/2^{k+1}} & {^kE_\infty^{2,1} = \Z/2} & 0 \\
	0 & {{}^{k+1}E_\infty^{3,0}\cong\Z/2^{k+1}} & {\mathit{rSH}^3(B\Z/2^{k+1})\cong\Z/2^{k+2}} & {^{k+1}E_\infty^{2,1} \cong \Z/2} & 0
	\arrow[from=1-1, to=1-2]
	\arrow["{1\mapsto 2}", from=1-2, to=1-3]
	\arrow["{1\mapsto 4}"', from=1-2, to=2-2]
	\arrow["{\bmod 2}", from=1-3, to=1-4]
	\arrow["{p^*}", from=1-3, to=2-3]
	\arrow[from=1-4, to=1-5]
	\arrow["{1\mapsto 0}", from=1-4, to=2-4]
	\arrow[from=2-1, to=2-2]
	\arrow["{1\mapsto 2}", from=2-2, to=2-3]
	\arrow["{\bmod 2}", from=2-3, to=2-4]
	\arrow[from=2-4, to=2-5]
\end{tikzcd}\end{equation}
which is the map induced by $p^*$ between the extensions of the Dijkgraaf--Witten and Gu--Wen layers. The leftmost and rightmost vertical arrows follow from our calculation of $p^*$ applied to $E_\infty^{3,0}$ and $E_\infty^{2,1}$; commutativity then forces the middle vertical arrow to send $1\mapsto 4$.\footnote{Here we have been cavalier about the choice of isomorphism $\mathit{rSH}^3(B\Z/2^k)\cong\Z/2^{k+1}$, but this is easily fixed: choose any such isomorphism for $k = 2$, then inductively define it for larger $k$ so that the middle vertical arrow in~\eqref{one_more_xtn} sends $1\mapsto 4$.}
\end{proof}

Now we can return to the commutative diagram \eqref{map_Smith_LES_dual}. The commutativity suggests that for the right vertical $p^*$, after choosing the generators of the corresponding groups in a suitable way, we must have 
\begin{equation}\label{eq:map_pre}
(1, 0)\mapsto 4,\quad (0, 1)\mapsto 2^{k - 1}.
\end{equation}
Therefore, we again see that extending $G=\Z/2^k$ by $\Z/2$ is not enough to trivialize $\SH^5(B\Z/2^k,0,y)$. 

Now consider an even larger extension
\begin{equation}
1\longrightarrow \Z/2^m \longrightarrow \Z/2^{k+m}\xrightarrow{p} \Z/2^k\rightarrow 1\,,
\end{equation}
with $m\geq 1$. The induced pullback $p^* \colon\SH^5(B\Z/2^k,0,y)\rightarrow \SH^5(B\Z/2^{k+m})$ is given by the composition
\begin{equation}
    \SH^5(B\Z/2^k,0,y)\longrightarrow
    \SH^5(B\Z/2^{k+1}, 0, 0) \longrightarrow
    \SH^5(B\Z/2^{k+m}, 0, 0)\,.
\end{equation}
Therefore, composing the result of \eqref{eq:map_pre} and part~\ref{Z_pullback} of \cref{p_on_coh}, it sends the generator $(1,0)$, which we denote as $\alpha_{\mathrm{GW}}$, to $4\cdot 8^{m-1}$, and the generator $(0,1)$, which we denote as $\alpha_{\mathrm{Maj}}$, to $2^{k-1}\cdot 8^{m-1}$. In order for them to trivialize in $\SH^5(B\Z/2^{k+m})$. They must be greater than or equal to $2^{k+m-1}$. Straightforward algebra shows that for $\alpha_{\mathrm{GW}}$, the corresponding $m$ satisfies $m\geq \frac{k}{2}$, while for $\alpha_{\mathrm{GW}}$, the corresponding $m$ satisfies $m\geq 2$. Therefore, we prove that for anomalies generated by $\alpha_{\mathrm{GW}}$, we can use a $\Z/2^m$ gauge theory to realize the given anomalies with $m\geq \frac{k}{2}$, and for anomalies generated by $\alpha_{\mathrm{Maj}}$, we can use a $\Z/4$ gauge theory to realize the given anomalies. 

\end{proof}

\begin{rem}\label{rem:refJuven}
In \cite{Cheng:2024awi}, Cheng--Wang--Yang explicitly construct the TQFT state which realizes the anomaly corresponding to $1\in \Z/8$ for $k=1$ or $(1,2^{k-2}) \in \Z/2 \oplus \Z/(2^{k+1})$ for $k\geq 2$, with the help of the crystalline equivalence principle. Our results match their results. In particular, we also confirm that $\Z/2$ gauge theory is not enough and the minimal gauge group $K$ has to be $\Z/4$. 
\end{rem}

\begin{rem}\label{rem:refJuven2}

Following the submission of our work to arXiv, we became aware of the independent studies by \cite{WW25} and \cite{ZJSNew}, which demonstrate significant overlap with the calculations presented in this section. We note that our findings are in excellent agreement with their reported results using slightly different methods; this independent convergence reinforces the robustness of the calculations and conclusions drawn herein.

\end{rem}

\section{Conclusion and discussion}\label{section:conclusion}

Our goal in this project is to construct (3+1)d fermionic TQFTs that realize prescribed anomalies. This is motivated by the challenge of understanding the IR phases of strongly coupled UV gauge theories (Question~\ref{question:main}), and we focused on symmetries and anomalies in such theories. We expect that the framework developed here will also shed light on the IR behavior of fermionic lattice systems and provide new insights into beyond-Standard-Model theories from the perspective of symmetries and anomalies. Our construction is based on the framework of fusion 2-categories that classifies $(3+1)$d symmetry-enriched topological orders. Using this framework, we extended the symmetry extension procedure to unitary symmetries in the fermionic setting. 

Our construction crucially relies on twisted supercohomology $\SH^5(BG, s, \omega)$. This choice is motivated by the fact that supercohomology is directly related to the data of fusion 2-categories and the structure theorems presented in \S\ref{section:always_works}. We prove that, in (3+1)d fermionic systems, any supercohomology anomaly can be saturated by an appropriate fermionic topological order, whereas any beyond-supercohomology anomaly cannot be saturated by any fermionic topological order. This establishes a clear dichotomy of (3+1)d anomalies for finite symmetries and answers Córdova--Ohmori's question~\cite{CO2} in this setting. We discuss some physical implications of our results in the companion paper \cite{DYY3}.

We then go to specific examples corresponding to cyclic group symmetries, which may be trivially or nontrivially extended by fermion parity (\cref{ex:cyclic,ex:cyclicfermion}).  For each of these examples, 
we explicitly computed the corresponding supercohomology groups $\SH^5(BG, s, \omega)$ that classify anomalies. Then we found group extensions $H \twoheadrightarrow G$ that trivialize the anomaly corresponding to the generators of $\SH^5(BG, s, \omega)$, as detailed in \cref{thm:Znsymmetry,thm:trivializedGen}. As explained in \S\ref{subsection:TO}, this construction demonstrates a concrete path to realizing these anomalous gapped phases.

A significant technical contribution announced in this paper is the development of a \textit{hastened Adams spectral sequence} for computing supercohomology groups. This tool made the computations in \cref{ex:cyclic,ex:cyclicfermion} tractable. Further development and refinement of these spectral sequence techniques will be essential for classifying anomalies of more complex symmetry groups, such as non-abelian groups or symmetries with non-trivial $s$ and $\omega$ twists, thus expanding the reach of the fermionic symmetry extension procedure \cite{DYY2}. We also make heavy use of the recently developed \textit{Smith long exact sequences}, making these calculations that have not been done in the literature tractable.

In summary, our results provide a systematic, mathematically rigorous path for constructing candidate IR topological orders that can saturate a given UV anomaly, offering a powerful tool for studying strongly-coupled fermionic systems.

Several natural future directions are in order. 
\begin{enumerate}
\item A natural extension of our construction is to systems with Lie group symmetries. In particular, it would be interesting to determine (i) how to construct anomalous TQFTs that realize a prescribed anomaly associated with a Lie group symmetry, possibly along the lines of a generalized symmetry-extension framework, and (ii) under what conditions a given Lie group anomaly necessarily implies symmetry-enforced gaplessness. It is already known that if the anomaly IFT $\alpha$ of a (3+1)d theory is nontopological (so that its anomaly polynomial is nonzero, a fact which can be seen perturbatively), no topological order can have anomaly $\alpha$, so it would be especially interesting to attack the remaining topological IFTs.
\item It is also of interest to investigate whether our construction can be generalized to incorporate time-reversal symmetries, potentially relying on a more refined categorical understanding of anti-unitary symmetries.
\item In another direction, we also want to understand whether we can display the symmetry actions on the point-like and string-like excitations in a (3+1)d fermionic topological order in an explicit manner. This may involve disentangling the underlying fusion 2-category data into the explicit numeric data of pentagonator, braiding, etc., like those displayed in \cite{Huang:2024ivi,Huang:2025hpk}.
\item Finally, it would be worthwhile to explore applications of our results to concrete physical systems, including strongly-coupled quantum chromodynamics in four dimensions, Weyl semimetals, and scenarios relevant to physics beyond the Standard Model. We will reserve a more extensive discussion of these matters from a physical point of view in a companion paper \cite{DYY3}.
\end{enumerate}

\subsection*{Acknoweledgements}
It is a pleasure to thank
Clay Córdova,
Thibault Décoppet,
Jaume Gomis,
Weizhen Jia,
Theo Johnson-Freyd,
Ryohei Kobayashi,
Cameron Krulewski,
Tian Lan,
Miguel Montero,
Lukas Müller,
Kantaro Ohmori,
Luuk Stehouwer,
Chong Wang,
Juven Wang,
and Rui Wen
for helpful conversations. We especially thank Theo Johnson-Freyd for sharing with us his insights into the anomalies of topological orders, which helped shape Appendix \ref{section:anomaliesObstructions}. We would also like to thank Zheyan Wan and Juven Wang for sharing their recent work on very similar topics \cite{WW25,ZJSNew}, which overlaps with some of our results in \Cref{tab:results} and \cref{mainthmB}.

WY was supported by the Natural Sciences and Engineering Research Council of Canada (NSERC) and the European Commission under the Grant Foundations of
Quantum Computational Advantage. MY is supported by the EPSRC Open Fellowship EP/X01276X/1.

\appendix

\section{Review on generalized cohomology theory}\label{subsec:generalized_coho}

Classical cohomology theories like singular cohomology, de Rham cohomology, and sheaf cohomology share common axiomatic properties but capture different topological and geometric information. Generalized cohomology theory provides a unifying framework that encompasses these classical theories while allowing for new, exotic cohomology theories with applications throughout mathematics and physics. In this appendix, we review the necessary aspects of generalized cohomology theories that we use in this paper. A standard textbook for generalized cohomology theory is \cite{adams1974stable}.

We start with the definition of a spectrum, which is a homotopical object representing a generalized homology or cohomology theory. An \textit{$\Omega$-spectrum}\footnote{There are many different yet equivalent ways to define spectra; see for example~\cite{MMSS01}. We use $\Omega$-spectra because they tend to appear in physics applications: see, for example, \cite{Kit13, Freed:2014iua, Kit15, GJF19, CGT25, 2025arXiv250721209T}.}
$E$ is a sequence of pointed topological spaces $\{E_n\}_{n \in \mathbb{Z}}$ together with structure maps, which are homotopy equivalences:
\begin{equation}
\sigma_n : E_n \overset\simeq\to \Omega E_{n+1}
\end{equation}
where $\Omega$ denotes the based loop space functor. These structure maps encode the fundamental relationships between different degrees of the cohomology theory. Given an  $\Omega$-spectrum $E = \{E_n, \sigma_n\}$, we can define a \textit{generalized cohomology theory} by setting: 
\begin{equation}
E^n(X) := [X, E_n]
\end{equation}
where $[X, E_n]$ denotes the set of homotopy classes of pointed maps from $X$ to $E_n$.

One can think of a spectrum as encoding ``stable'' homotopy-theoretic information. While individual spaces $E_n$ may have complicated unstable behavior, the spectrum captures what remains after we  have ``stabilized'' by taking suspensions. The structure maps $\sigma_n$ induce \textit{suspension isomorphisms}
\begin{equation}
E^n(X) \overset{\cong}{\longrightarrow} E^{n+1}(\Sigma X)\,.
\end{equation}
This is the key property that makes the theory ``stable'' and gives it the structure of a cohomology theory. We will also use $\Sigma^k E$ to denote the suspension of the theorem, and the corresponding generalized cohomology theory is simply the original cohomology shifted by degree $k$.

The generalized cohomology theory $E^*$ satisfies the Eilenberg-Steenrod axioms~\cite{ES45} except the dimension axiom. The \textbf{coefficient groups} $E^n(\mathrm{pt})$ are the cohomology groups of a point, and these can be computed as:
\begin{equation}
E^n(\mathrm{pt}) = \pi_{-n}(E) := \operatorname{colim}_k \pi_{k-n}(E_k)\,,
\end{equation}
where $\pi_{-n}(E)$ denotes the $n$-th stable homotopy group of the spectrum $E$. From these coefficient groups, we have the Atiyah--Hirzebruch spectral sequence (AHSS) for a generalized cohomology theory $E^*$ 
\begin{equation}
E_2^{p,q}\cong H^p\bigl(X;E^q(\mathrm{pt})\bigr)\quad\Longrightarrow\quad E^{p+q}(X),
\end{equation}
with differentials $d_r:E_r^{p,q}\to E_r^{p+r,q-r+1}$. This gives us a practical tool for calculation and also has nice physical interpretations \cite{Wang:2017moj,PhysRevX.10.031055}. 

Some standard examples of spectra and their related generalized cohomology theories are as follows:

\begin{enumerate}
\item \textbf{Eilenberg--Mac Lane spectrum and ordinary cohomology}: The Eilenberg--Mac Lane spectrum is built from $E_n = K(A, n)$, the Eilenberg--Mac Lane spaces, where $A$ can be any Abelian group. This spectrum is denoted as $HA$ in the literature, and gives us the usual singular cohomology $H^*(X; A)$ with coefficient $A$.

\item \textbf{$\KO$-spectrum and the connective cover $\ko$}: The spectrum $\KO$ is defined by the sequence of spaces
\begin{equation*}
\KO_0 = \mathrm{BO},\quad 
\KO_1 = \mathrm{O},\quad
\KO_2 = \mathrm{O}/\mathrm{U},\quad
\KO_3 = \mathrm{U}/\mathrm{Sp},\quad
\KO_{n+8} \simeq \KO_n,
\end{equation*}
together with structure maps inducing the loop-space identifications
$\Omega \KO_{n+1} \simeq \KO_n$.  This spectrum represents real $K$-theory:
for any space $X$, the generalized cohomology groups $\KO^*(X) = [X, \KO_*]$ classify real vector bundles over $X$ and their formal differences up to
stable isomorphism.  The 8-fold Bott periodicity of the $\KO$ spectrum
implies that the coefficient groups $\KO^n(\mathrm{pt})$ repeat with
period~8, i.e., $\Z$, $\Z/2$, $\Z/2$, 0, $\Z$, 0, 0, 0 as $n$ increases from 0. This generalized cohomology theory plays an important role in the classification of free-fermion SPTs \cite{FH,2025arXiv250708694S}.

The connective cover $\ko$ is obtained by truncating $\KO$ below degree 0, and yields the connective real $K$-theory $\ko^*(X)$. The geometric meaning of $\ko$-theory is less obvious than that of $\KO$, but it provides a computationally convenient approximation to the spectrum of (interacting) fermionic SPTs \cite{DDHM24}.

\item \textbf{Thom spectrum, Madsen-Tillmann spectrum and cobordism theories}: The Thom construction maps a vector bundle $E \to X$ to a pointed space $\mathrm{Th}(E)$ via the one-point compactification of the total space. For a stable group $G$, e.g. $\SO$, $\Spin$, etc., the Thom spectrum $MG$ is constituted by the spaces $MG_n := \mathrm{Th}(\gamma_n)$, where $\gamma_n$ denotes the universal vector bundle over the classifying space $BG(n)$. These individual spaces assemble to form spectra representing various cobordism theories. Notably, the spectrum $MSpin$, which classifies spin cobordism, is the sequence of Thom spaces formed by the universal bundles over the spaces $BSpin(n)$. Relatedly, we also have the Madsen-Tillmann spectrum, which is constituted by the spaces $MTG_n := \mathrm{Th}(-\gamma_n)$. These theories are of particular importance in classifying fermionic SPTs.
\end{enumerate}

In an $n$-dimensional quantum system with global symmetry $G$, the anomalies we consider are classified by $E^*(BG)$, where the relevant degree depends on the spacetime dimension, and the generalized cohomology theory $E$ is twisted by the data $(s,\omega)$. In \S\ref{section:anomaliesObstructions}, we will examine three different generalized cohomology theories, corresponding to three different spectra. 
\begin{itemize}
\item The first generalized cohomology theory is associated to the spectrum $I_{\Z}\MSpin$, the Anderson dual of the Madsen-Tillmann spectrum $\MSpin$, and classifies reflection-positive invertible field theories (IFTs)  \cite{FH, Gra23} or fermionic SPTs in the language of condesned matter. By anomaly inflow, such IFTs can be used to cancel the anomalies for fermionic QFTs with $G$-symmetry in one dimension lower. 

\item The second generalized cohomology theory is related to what we call a \emph{categorical obstruction}, which represents the obstruction for the $G$-crossed extension of the underlying category. This gives a mathematically rigorous definition of anomalies of topological orders, based on their categorical description using higher fusion categories in this work. The associated spectrum is closely related to the super-Witt group $s\mathcal{W}$ \cite{DNO}, and hence will be denoted by $\SW$. 

\item Supercohomology $\SH$.
Supercohomology is first proposed in \cite{PhysRevB.90.115141} for classifying fermionic SPTs. For our purpose, supercohomology theory was defined in two equivalent ways in Appendix \ref{subsection:twistedSH}, and the corresponding spectrum can be thought of as the spectrum of $I_{\Z}\MSpin$ truncated to only degrees 0, 1, 2. 
Similar truncations also appear in e.g.\ classifying mixed-state SPTs \cite{2025PhRvX..15b1062M}.

\end{itemize}

\section{Twisted supercohomology in two ways}\label{subsection:twistedSH}
There are two ways of realizing supercohomology that will be important for this work. The first is as the Pontryagin dual of the spectrum $\tau_{\leq 2}\ko$, which fits into the following fiber sequence:
\begin{equation}
    \tau_{\geq 4} \ko \longrightarrow \ko \longrightarrow \tau_{\leq 2}\ko \,.
\end{equation}
The homotopy groups and $k$-invariants of $\SH$ can thus be read off of those of $\ko$: see~\cite[Proof of Lemma 5.6]{ABP67} for the latter. In particular, we see that the homotopy groups are given by 
\begin{equation}\label{eq:SHgroups}
    \pi_{-2}(\SH) =\Z/2,
    \quad \pi_{-1}(\SH) =\Z/2,
    \quad \pi_0 (\SH) = \mathbb{C}^\times\,.
\end{equation}
The $k$-invariant connecting the two copies of $\Z/2$ is
\begin{equation}
    \Sq^2 \colon H^*(-;\Z/2) \to H^{*+2}(-;\Z/2)
\end{equation}
and the $k$-invariant connecting $\Z/2$ with $\mathbb{C}^\times$ is
\begin{equation}
    (-1)^{\Sq^2} \colon H^*(-;\Z/2) \to H^{*+2}(-;\mathbb{C}^\times)\,.
\end{equation}

It is straightforward to introduce $s$ and $\omega$ twists from the homotopical point of view. In particular, the map $\ko\to\tau_{\le 2}\ko$ induces a map of twisting data, so we can use twists of $\ko$-theory to twist $\SH$. Given a space $X$, choose $s\in H^1(X;\Z/2)$ and $\omega\in H^2(X;\Z/2)$. The data $(s, \omega)$ defines a twist of $\ko$-theory over $X$~\cite{ABG10}, hence also define a twist of $\SH$ over $X$. We denote by $\SH^n(X,s,\omega)$ the corresponding degree-$n$ twisted supercohomology group. When $s = 0$, this is sometimes written as $\SH^{n+\omega}(X)$, e.g.\ in~\cite{D11,Decoppet:2024htz,Teixeira:2025qsg}.

The second realization of supercohomology is in terms of the Picard $2$-groupoid $\tsVect^\times$. This perspective on twisted supercohomology makes natural contact with applications in the fusion $2$-categories literature~\cite{Johnson-Freyd:2020twl,JFY2,DY23a,D11,Decoppet:2024htz,DY2025, Teixeira:2025qsg}. Specifically, the homotopy groups of this Picard $2$-groupoid are
\begin{align}
  \pi_0  \, \tsVect^\times = \Z/2, \quad \pi_1  \, \tsVect^\times = \Z/2, \quad  \pi_2  \, \tsVect^\times = \mathbb{C}^\times,
\end{align}
with the unique nontrivial Postnikov invariants connecting the groups~\cite{Fre12,DG18}. Therefore the spectrum corresponding to $\tsVect^\times$ under the stable homotopy hypothesis~\cite{GJO19,MOPSV22} is $I_{\C^\times}(\tau_{\le 2}\ko) = SH$, as it has isomorphic homotopy groups and Postnikov invariants. 
Thus, the abelian group of homotopy classes of maps 
\begin{equation}
    X \longrightarrow B^{n-2} \tsVect^\times 
\end{equation}
is naturally isomorphic to $\SH^{n}(X)$.

Like for twisted ordinary cohomology, we will use automorphisms of $\tsVect^\times$ to twist supercohomology. The automorphisms of interest to us are:
\begin{description}
    \item[Fermion parity] tensor a $1$-morphism with the odd line. This defines a $B\Z/2$-action.
    \item[Duality] send objects, $1$-morphisms, and $2$-morphisms to their duals. This is reminiscent of the time-reversal action of $\sVect$ and almost defines a $\Z/2$-action.
\end{description}
The Koszul sign rule means that duality does not square to the identity, but rather participates in an abelian $2$-group extension with fermion parity:
\begin{equation}
\label{2gpext}
    0 \to B\Z/2 \to \mathbb A \to \Z/2 \to 0.
\end{equation}
$2$-group extensions of the form~\eqref{2gpext} are classified by $H^3(B\Z/2;\Z/2)\cong\Z/2$~\cite[Theorem 1]{SP11}, so the extension $\mathbb A$ of duality by fermion parity is uniquely specified up to isomorphism by the fact that it is non-split.

Thus, given a space $X$ and a map $f\colon X\to B\mathbb A$, we can form the associated bundle
\begin{equation}
\label{supercoh_assoc}
\begin{tikzcd}
        {(B^{n-2}\tsVect^\times) \times_{\mathbb A} f^*(E\mathbb A)} 
        \ar[d]\\
        X.
    \end{tikzcd}
\end{equation}
Then $\SH^{n+f}(X)$ is the abelian group of homotopy classes of sections of~\eqref{supercoh_assoc}.

Though $\mathbb A$ is not split, there is a homotopy equivalence of spaces $B\mathbb A\simeq B\Z/2\times B^2\Z/2$, so we will identify a twist of supercohomology by a triple $(X, s, \omega)$, where $s\in H^1(X;\Z/2)$ and $\omega\in H^2(X;\Z/2)$, matching the homotopical definition of twisted supercohomology.

Hence if $X$ is a space equipped with a map $\omega \colon X\to  B^2\Z/2$, the $\omega$-twisted $n$-th supercohomology of $X$ is the group of homotopy classes of
    $B\Z/2$-equivariant maps from $X$ to $B^{n-2}\tsVect^\times$.
    In the companion paper~\cite{DYY2}, we show that the two notions of $(X, s, \omega)$-twisted supercohomology that we have introduced are naturally isomorphic.

    \begin{rem}
    \label{s_nonzero}
        In the context of fusion 2-categories the $s$-twist in the first definition of twisted supercohomology has
		not previously appeared in the literature. One reason for this is because the TQFTs that fusion 2-categories construct are oriented \cite{douglas2018fusion}. It would be interesting to have a definition of fusion 2-categories with a unitary structure that parallels what exists for fusion 1-categories; symmetries of unitary fusion $2$-categories could potentially correspond to twists with $s\ne 0$.\footnote{See \cite{ferrer2024dagger,chen2024manifestly, SS24, Ste24, Bar25, MS23, Muller:2025ext} for recent progress towards unitary higher categories.}
    \end{rem}

\begin{rem}
The space of homotopy equivalences $\phi\colon B\mathbb A\overset\simeq\to B\Z/2\times B^2\Z/2$ is not connected, implying there is an ambiguity in how we identified the data $(s,\omega)$ with a twist of supercohomology. There are a few ways to address this, which we will discuss in more detail in~\cite{DYY2}. We choose the (standard) convention that, if $a\in H^1(B\Z/2;\Z/2)$ denotes the unique nonzero class, $(B\Z/2, a, 0)$-twisted supercohomology maps to the twist of spin cobordism that is isomorphic to \pinm cobordism, rather than \pinp cobordism.

This ambiguity does not affect twists $(X,s,\omega)$ for which $s = 0$, so it will not play a major role in this paper.

\end{rem}

For computations of (twisted) supercohomology groups, it is sometimes helpful to use an explicit cochain description of twisted supercohomology.
We now present the following conditions that the cochains of \S\ref{subsection:fermionicsupercoh}
must satisfy in the twisted setting:
  \begin{itemize}
      \item the cochain $a  \in C^{n-2}(BG;\Z/2)$ solves $da =0$,
       \item the cochain $b \in C^{n-1}(BG;\Z/2)$ solves $db  = (\Sq^2+\omega) a$, and
       \item the cochain $c \in C^n(BG;\mathbb{C}^\times)$ solves $dc = (-1)^{(\Sq^2+\omega)b} \cdot f_\omega(a)$.
  \end{itemize}
Here, the cochain $f_\omega(a)$ represents the failure of $(\Sq^2+\omega)b$ to be closed, and represents the \textit{secondary cohomology operation} \cite{mosher2008cohomology} corresponding to the relation $\Sq^1 \Sq^2\Sq^2 = 0$. The formula of $f_\omega(a)$ up to $n=4$ is detailed in \cite{PhysRevX.10.031055,2025arXiv251225069N}. Based on the homotopical definition of supercohomology, this is equivalent to the information of the Atiyah--Hirzebruch spectral sequence (AHSS).\footnote{To obtain the full group structure from the AHSS, we also need to solve the extension problem or obtain the \emph{stacking rules} of the AHSS. Stacking rules of supercohomology up to $n=3$ written in terms of the explicit cochain descriptions are detailed in \cite{2024PhRvB.110w5117R}.} 
In~\cite{DYY2}, we also develop a complementary tool, the \term{hastened Adams spectral sequence} (HASS), that helps resolve many extensions in the AHSS. Importantly, for almost all the examples in \Cref{tab:results} we will need to use the hastened Adams spectral sequence to compute the value of the degree 5 supercohomology. 

\section{The $p+ip$ layer in full generality}
\label{s:p+ip}
In \cref{notpip}, we gave a simple definition of the $p+ip$ layer of a 5d $(BG, s, \omega)$-twisted spin reflection positive IFT $\alpha$ in terms of the Atiyah--Hirzebruch spectral sequence. To use this spectral sequence, the spectrum $I_\Z\MTSpin(BG, s, \omega)$ must split as $\MTSpin$ smash some other spectrum $X$. It suffices to assume $(s,\omega) = (w_1(V), w_2(V))$ for a vector bundle $V\to BG$, so that $X = (BG)^{V - \mathrm{rank}(V)}$, and for these choices of $(G, s, \omega)$, \cref{notpip} is a valid definiton of the $p+ip$ layer.
However, there are finite groups $G$ and data $(s,\omega)$ for which no such vector bundle $V$ exists, by work of Gunarwardena--Kahn--Thomas~\cite[\S 2]{GKT89}. In these cases, we have to give a more complicated definition of the $p+ip$ layer, and the goal of this appendix is to do so. As a bonus, we will be able to provide alternate definitions of the layers of supercohomology (Majorana, Gu--Wen, and Dijkgraaf--Witten).

To be clear, we are not really doing anything new: these layers are the pieces of the associated graded of the Postnikov filtration on $I_\Z\MTSpin$. \Cref{notpip} implicitly takes this view, as the AHSS is exactly the spectral sequence induced by this filtration~\cite{Mau63}. Our more general definition below simply applies the Postnikov filtration to an $\MTSpin$-module Thom spectrum, in the language of Ando--Blumberg--Gepner--Hopkins--Rezk~\cite{ABGHR14a, ABGHR14b}, then uses a Thom isomorphism theorem for ordinary cohomology. This is a common theme when working with non-vector-bundle twists: though the homotopical prerequisites are higher, essentially the same theorems are true for these more general twists, and for broadly similar reasons.
\begin{lem}[{Lurie~\cite[Proposition 7.1.1.13]{HA}}]
Let $R$ be a connective, $E_\infty$-ring spectrum and $M$ be an $R$-module spectrum. Then the Postnikov $t$-structure on the $\infty$-category of spectra lifts across the forgetful functor $\cat{Mod}_R\to\cat{Sp}$. Thus, the Postnikov truncation maps $\tau_{\le n}\colon M\to \tau_{\le n}M$ and $\tau_{\ge n}\colon \tau_{\ge n}M\to M$ canonically acquire the structure of $R$-module maps.
\end{lem}
We will use $R = \MTSpin$.
Thus we have a map of $\MTSpin$-module spectra
\begin{equation}
\label{did_map}
    di+d\colon \Sigma^4 H\Z \simeq \tau_{\le 4}\tau_{\ge 4}\MTSpin \overset{\tau_{\ge 4}}{\longrightarrow} \tau_{\le 4}\MTSpin.
\end{equation}
Given $G, s, \omega$ as above, let $\mathit{MT\xi}(G, s, \omega)$ denote the Madsen--Tillmann spectrum for $(G, s, \omega)$-twisted spin structures, so that $\mho_\Spin^k(BG, s, \omega)$ is by definition $\pi_0(\Sigma^k I_{\C^\times}\mathit{MT\xi}(G, s, \omega))$.

Hebestreit--Joachim~\cite[Corollary 3.3.8]{HJ20}\footnote{Hebestreit--Joachim state their result in terms of bordism groups, but their proof goes through at the level of spectra: see~\cite[Remark 1.28]{DY23}.} identified $\mathit{MT\xi}(G, s, \omega)$ with the $\MTSpin$-module Thom spectrum associated to the map
\begin{equation}
\label{MSpin_Thom}
    BG \overset{(s,\omega)}{\longrightarrow}
        K(\Z/2, 1)\times K(\Z/2, 2) \overset{\phi}{\underset{\simeq}{\longrightarrow}} B\O/B\Spin \overset{T}{\longrightarrow} B\GL_1(\MTSpin)
\end{equation}
in the sense of Ando--Blumberg--Gepner--Hopkins--Rezk~\cite{ABGHR14a, ABGHR14b}; the equivalence $\phi$ is proven in~\cite[Proposition 1.37]{DY23},\footnote{See also Beardsley--Luecke--Morava~\cite[Propositions 4.1 and 5.19]{BLM23} and Carmeli--Luecke~\cite[Theorem C]{CL24} for splitting results for spaces closely related to $B\O/B\Spin$.} and the map $T$ is constructed by May--Quinn-Ray~\cite[Lemma IV.2.6]{MQRT77}. See~\cite[\S 1.2.3]{DY23} for more information on this perspective on $\mathit{MT\xi}(G, s, \omega)$. Now smash $\mathit{MT\xi}(G, s, \omega)$ with the map $di+d$~\eqref{did_map}, over the base $\MTSpin$:
\begin{equation}
\label{did_tensored}
    \Sigma^4 H\Z\wedge_{\MTSpin} \mathit{MT\xi}(G,s,\omega)
        \xrightarrow{(d+id)\wedge \id_{\mathit{MT\xi}(G, s, \omega)}} \tau_{\le 5}\MTSpin \wedge_{\MTSpin} \mathit{MT\xi}(G,s,\omega) \longrightarrow
        \tau_{\le 5}(\mathit{MT\xi}(G, s, \omega)),
\end{equation}
where the rightmost map exists because Postnikov truncation is lax symmetric monoidal (see, e.g., \cite[Example/Proposition 3.12]{HNP25}). Lax monoidality of Postnikov truncation also implies that $H\Z$ is an $E_\infty$-$\MTSpin$-algebra, so the base change $\Sigma^4 H\Z\wedge_{\MTSpin} \mathit{MT\xi}(G,s,\omega)$ is $\Sigma^4$ of the $H\Z$-module Thom spectrum of the composition
\begin{equation}
    BG \underset{\eqref{MSpin_Thom}}{\longrightarrow} B\GL_1(\MTSpin) \longrightarrow B\GL_1(H\Z) \overset\simeq \longrightarrow K(\Z/2, 1),
\end{equation}
and by~\cite[Lemma 1.8 and \S 1.2.1]{DY23} this Thom spectrum can be identified with $H\Z\wedge (BG)^{\sigma-1}$, where $\sigma\to BG$ is the real line bundle with $w_1(\sigma) = s$. Thus we can rephrase~\eqref{did_tensored} as
\begin{equation}
    \Sigma^4 H\Z\wedge (BG)^{\sigma-1} \longrightarrow \tau_{\le 5} \mathit{MT\xi}(G, s, \omega).
\end{equation}
Now apply $\Sigma^5 I_{\C^\times}$:
\begin{subequations}
\begin{equation}
\label{almost_there}
    (\text{--})_{p+ip}\colon \Sigma^5 I_{\C^\times}(\tau_{\le 5} \mathit{MT\xi}(G, s, \omega)) \longrightarrow
    \Sigma^5 I_{\C^\times}(\Sigma^4 H\Z\wedge (BG)^{\sigma-1}).
\end{equation}
The universal property of $I_{\C^\times}$ allows us to rewrite the domain and codomain of~\eqref{almost_there} as follows:
\begin{equation}
\label{pip_spectra}
    (\text{--})_{p+ip}\colon \tau_{\ge 0} (\Sigma^5 I_{\C^\times}\mathit{MT\xi}(G, s, \omega)) \longrightarrow
    \Sigma H\C^\times \wedge (BG)^{\sigma-1}.
\end{equation}
\end{subequations}
Since $\pi_0$ is an isomorphism on connective covering maps, there is a canonical isomorphism
\begin{equation}
\label{connective_isom}
    \pi_0(\tau_{\ge 0} (\Sigma^5 I_{\C^\times}\mathit{MT\xi}(G, s, \omega)))\overset\cong\longrightarrow \pi_0(\Sigma^5 I_{\C^\times}\mathit{MT\xi}(G, s, \omega)) =: \mho_\Spin^5(BG, s, \omega),
\end{equation}
so we can (finally!) evaluate~\eqref{pip_spectra} on 5d IFTs.
\begin{defn}
\label{pip_defn}
The \term{$p+ip$ layer} of $\alpha\in\mho_\Spin^5(BG, s, \omega)$ is the class $\alpha_{p+ip}\in H^1(BG; \C^\times_s)$ obtained by evaluating $\pi_0$ of the map~\eqref{pip_spectra} on $\alpha$ via the isomorphism~\eqref{connective_isom}.
\end{defn}
\begin{lem}
\label{pip_defns_agree}
Suppose there is a vector bundle $V\to  BG$ such that $s = w_1(V)$ and $\omega = w_2(V)$. Then \cref{notpip,pip_defn} agree.
\end{lem}
\begin{proof}
For convenience, let $n\coloneqq\mathrm{rank}(V)$.
The lemma assumption implies that the twist $(s,\omega)\colon BG\to B\GL_1(\MTSpin)$ that we built in~\eqref{MSpin_Thom} factors through $B\GL_1(\Sph)\to B\GL_1(\MTSpin)$; see~\cite[\S 1]{DY23}. Thus, as discussed in~\cite[\S 1.2]{ABGHR14b}, there is an $\MTSpin$-module equivalence $\mathit{MT\xi}(G, s, \omega)\simeq\MTSpin\wedge (BG)^{V-n}$. The next step in constructing the $p+ip$ layer, as in~\eqref{did_tensored}, is to smash with $di+d$. Since $R\wedge_R X\simeq X$ for an $R$-module $X$, we see that $di+d$ smashed with $\id_{\mathit{MT\xi}(G,s,\omega)}$ coincides up to $\MTSpin$-module equivalence with the result of smashing $di+d$, over $\Sph$, with $(BG)^{V-n}$. After applying $\Sigma^5 I_{\C^\times}$, this is the Postnikov $5$-truncation of the connective cover of $\MTSpin$ smashed with $(BG)^{V-n}$. Since the Atiyah--Hirzebruch spectral sequence is the spectral sequence induced from the Postnikov filtration~\cite{Mau63}, this map can be identified with the projection onto the line $q = 4$ as in \cref{notpip}.
\end{proof}
Now we fulfill the promise that the $p+ip$ layer is the complete obstruction to realizing a class in $\mho_\Spin^5$ as a supercohomology class.
\begin{proof}[Proof of \cref{cofib_pip}]
The cofiber of~\eqref{did_map} is by definition $\tau_{\le 3}\colon \tau_{\le 5}\MTSpin\to \tau_{\le 3}\MTSpin$; after applying $I_{\C^\times}$, this is (the connective cover of $\Sigma^5$ of) the usual map $\SH\to\mho_\Spin^5$. The rest of the proof follows by tracing this fact through the rest of the construction of the $p+ip$ layer.
\end{proof}
In much the same way we can construct the Majorana, Gu--Wen, and Dijkgraaf--Witten layers of a twisted supercohomology class. As mentioned above, lax monoidality of Postnikov truncation implies that $\tau_{\le n}\ko$ is an $E_\infty$-ring spectrum, and that $\tau_{\le n}\colon\ko\to\tau_{\le n}\ko$ is an $E_\infty$-ring map. Therefore, given $(s,\omega)$ as usual, we can compose the map $f_{s,\omega}\colon BG\to B\GL_1(\MTSpin)$ defined by $(s,\omega)$ from~\eqref{MSpin_Thom} with the map $B\GL_1(\MTSpin)\to B\GL_1(\tau_{\le n}\ko)$ induced by the $E_\infty$-ring maps
\begin{equation}
    \MTSpin\longrightarrow\ko\overset{\tau_{\le n}}{\longrightarrow} \tau_{\le n}\ko,
\end{equation}
where the first map is the Atiyah--Bott--Shapiro orientation~\cite{ABS64}, realized as an $E_\infty$-ring map by Joachim~\cite{Joa04},
to obtain a map $f_{s,\omega,n}\colon BG\to B\GL_1(\tau_{\le n}\ko)$. Let $Mf_{s,\omega,n}$ denote the corresponding $\tau_{\le n}\ko$-module Thom spectrum.

For $n = 1,2$, imitate~\eqref{did_map} to define the following maps of $\tau_{\le n}\ko$-module spectra:
\begin{equation}
    \phi_n\colon \Sigma^n H\Z/2 \simeq \tau_{\ge n}\tau_{\le n} \ko \overset{\tau_{\ge n}}{\longrightarrow} \tau_{\le n}\ko.
\end{equation}
Then smash with $Mf_{s,\omega,n}$ over $\tau_{\le n}\ko$ and apply $\Sigma^{n+1}I_{\C^\times}$. As for the $p+ip$ layer, we can rephrase the result as a map
\begin{equation}\label{GW_Maj_interm}
    \Sigma^{n+1} I_{\C^\times}(\tau_{\le n} Mf_{s,\omega,n}) \longrightarrow \Sigma^{n+1}I_{\C^\times}(\Sigma^nH\Z/2\wedge (BG)_+).
\end{equation}
The chief difference to the $p+ip$ layer case is that we map the twist to its image in $[BG, B\GL_1(H\Z/2)]$, but $B\GL_1(H\Z/2)$ is contractible, so all twists of $BG$ over $H\Z/2$ are trivial, with Thom spectrum $\Sigma^\infty (BG)_+$. Take homotopy classes of maps and pass to the connective cover, like in~\eqref{connective_isom}, resulting in maps
\begin{subequations}
\begin{gather}
    (\bl)_{\mathrm{Maj}}\colon \SH^k(BG, s, \omega) \longrightarrow H^{k-2}(BG;\Z/2)\\
     (\bl)_{\mathrm{GW}}\colon \mathit{rSH}^k(BG, s, \omega) \longrightarrow H^{k-1}(BG;\Z/2).
\end{gather}
\end{subequations}
In other words, associated to any supercohomology class is its Majorana layer, and associated to any restricted supercohomology class, we obtain a Gu--Wen layer. Since the fiber of the Majorana layer map of spectra is restricted supercohomology, by a proof analogous to that of \cref{cofib_pip}, a supercohomology class with trivial Majorana layer has a Gu--Wen layer. Continuing in this way, a class with trivial Majorana and Gu--Wen layers has a Dijkgraaf--Witten layer in integral cohomology.

\section{Anomalies of topological orders from obstruction theory}\label{section:anomaliesObstructions}

In much of the literature, an anomaly is shorthand for a ’t Hooft anomaly, understood as an “obstruction to gauging’’ in a quantum field theory (QFT) and classified via SPTs in one higher dimension through anomaly inflow. However, as we have seen in \S\ref{section:always_works}, supercohomology anomalies play a distinguished role for $(3+1)$d fermionic topological orders. In this final section, we offer an alternative perspective on anomalies of topological orders based on obstruction theory in higher category theory and reformulate our main conjecture within this framework. We expect that this perspective will prove useful in broader settings, which we leave to future work.

From this viewpoint, there is no single, universal notion of anomaly that applies uniformly across all physical contexts. Rather, the appropriate classification framework depends on the setting at hand—be it continuum QFTs, lattice systems, or categorical formulations of topological orders. Indeed, it is reminiscent of the different notions of anomalies in lattice systems that are recently explored in the literature \cite{2014PhRvB..90w5137E,2020PhRvB.101v4437E,2024arXiv240102533K,2025arXiv250721209T,Kapustin:2025nju,Kapustin:2025rhp}. 

In the context of a topological order, we anticipate that the new perspective is purely based on its categorical/algebraic data and offers a mathematically well-defined notion of anomaly. Previous literatures that discuss anomalies of topological orders from this perspective include \cite{ENO2,JF,JFY2,2024JHEP...11..111L,Antinucci:2025fjp,Stockall:2025ppu}. As we see in \S\ref{section:always_works}, the classification that emerges from this definition differs from the familiar notions of 't Hooft anomaly in a continuum QFT. Nevertheless, the two are related by a natural map between their underlying spectra, which we will discuss in detail. 

Motivated by this, we propose that a broad class of these different notions of anomalies, particularly those connected to the 't Hooft anomaly of a QFT, can be systematically organized using the language of generalized cohomology theory, reviewed in Appendix~\ref{subsec:generalized_coho}. Furthermore, physical processes may give rise to maps between generalized cohomology theories, such as renormalization group flow connecting theories described by algebraic data in higher category theory, i.e.\ topological order, to a TQFT.

\subsection{Categorical obstructions of topological orders and 't Hooft anomalies}\label{subsection:obst}

In the context of continuum QFT, a theory is said to have a 't Hooft anomaly for a symmetry group $G$ if, when coupled to background $G$ gauge fields, the partition function fails to be invariant under $G$ gauge transformations, even after accounting for possible local counterterms. Based on the hypothesis of anomaly inflow \cite{Freed:2014iua}, 't Hooft anomalies are said to be classified by IFTs in one higher dimension. Then Freed--Hopkins~\cite{FH} and Grady~\cite{Gra23} showed that, for fermionic theories, these fermionic IFTs are classified by some generalized cohomology theory with the relevant spectrum in question being $I_\Z \MSpin$. In particular, for fermionic $G$-symmetry, the classification of $n$-dimensional 't Hooft anomalies are given by $I_\Z \MSpin^{n+1}(BG)$.

While this perspective is sufficiently general across different quantum systems, it may not be able to capture all the algebraic information of the underlying quantum system associated to symmetries, where more handwaving concepts like ``gauging'' or ``anomaly inflow'' can be defined in a much more precise manner. For fermionic topological orders in (3+1)d which have a fusion 2-categorical description \cite{JF,DY2025}, we may be able to define anomalies purely in terms of the interaction of the symmetry and the categorical data. To distinguish the anomalies defined in this new perspective, we define the \emph{categorical obstruction} for a $G$-action on a $\tsVect$-enriched nondegenerate braided fusion $2$-category $\fB$ (which was called the categorical $G$-obstruction in the main text), to be the failure to construct a $\tsVect$-enriched nondegenerate faithfully graded $G$-crossed braided fusion 2-category extending the $G$-action on $\fB$. As explained in \cite[Section 4.4]{DY2025}, this perspective of anomaly is equivalent to an anomaly for a (3+1)d fermionic $G$-SET. This gives the full algebraic data that characterizes the interplay between symmetries and the underlying categorical data. We will see that it is classified by $\SW^5(BG)$, where $\SW$ denotes the \textit{super-Witt spectrum}.

To be more specific, let us first specialize to (2+1)d, and review the classical 1-categorical result in \cite{ENO2}. In \cite{ENO2},  Etingof--Nikshych--Ostrik--Meir constructed faithfully graded $G$-crossed braided extensions of a braided fusion 1-category $\cB$.
Such $G$-crossed braided extensions are parametrized by the homotopy classes of maps
\begin{equation}
    BG \longrightarrow B\sPic(\cB),
\end{equation}
where $\sPic(\cB)$ is the Picard groupoid of $\cB$, given by the space of invertible $\cB$-modules. See \cite[\S 2.2]{Bhardwaj:2024xcx} for a physical introduction and an example of how the extension theory proceeds.  

One way to think of this extension is to imagine a specific case when $\cB$ is nondegenerate and represents a (2+1)d TQFT. If $\cB$ has a $G$-symmetry, i.e.\ a map $\rho\colon G \rightarrow \sAut^{br}(\cB)$, then to form a $G$-crossed braided extension of $\cB$ is to insert $G$-defects into $\cB$ such that the fusion and associativity relations respect
the group multiplication of $G$ \cite{Barkeshli:2014cna}. The result is a (2+1)d $G$-SET, i.e. a nondegenerate $G$-crossed braided fusion 1-category, that incorporates extra data such as the $F$-symbols of objects in the $G$-crossed braided extension including original objects in $\cB$ as well as the extra $G$-defects. The different $G$-crossed extensions parametrize SET phases. 

We define the \textit{categorical obstruction} to be the complete obstruction, in the sense of obstruction theory in algebraic topology, to the existence of a lift
\begin{equation}
    \begin{tikzcd}
        & B\sPic(\cB) \arrow[d]\\
        BG \arrow[r] \arrow[ru,dotted]& B\sAut^{br}(\cB)\,.
    \end{tikzcd}
\end{equation}
In other words, the obstruction corresponds to the inability to define a topological phase in which symmetry fractionalization is non-anomalous and the $G$-crossed braided consistency conditions, like the heptagon equations in \cite{Barkeshli:2014cna}, are satisfied. Maps to $B\sAut^{br}(\cB)$ that factor through $B\sPic(\cB)$ are precisely those $G$-actions on $\cB$ that are non-anomalous.

We can generalize the obstruction to higher dimensional theories, and obtains the classifications of anomalies from this perspective. Let $\mathbf{C}$ be a fusion $n$-category, which can be loosely defined inductively via delooping and Karoubi completing as in \cite{Gaiotto:2019xmp}.\footnote{See \cite[Section 3.1]{Bhardwaj:2024xcx} for an explanation of a crucial technical assumption, that must be made with our current understanding of condensation, in order for the inductive construction to be valid at for all values of $n$. For the contents of this paper, we will not require those assumptions. See \cite[Section 4.1]{StockallCondensation} for a  treatment of higher fusion categories in
terms of Cauchy completion.} There is a fiber sequence of spaces given in \cite[Theorem 5.2.24]{Bhardwaj:2024xcx}, which follows from unpublished work by Jones--Reutter: 
\begin{equation}\label{eq:obstructionSES}
   B \mathbf{C}^\times \longrightarrow B\sAut^\otimes(\mathbf{C}) \longrightarrow B\mathbf{Bimod}(\mathbf{C})^\times\,,
\end{equation}
where $(\text{--})^\times$ denotes only taking the invertible parts of a symmetric monoidal category. The rightmost entry parametrizes obstructions to lifting a map $X \rightarrow B\sAut^\otimes(\mathbf{C})$ to $X \rightarrow B \mathbf{C}^\times$.
There is an analogous sequence in the fermionic case, when each entry is a category enriched in super ($n$)-vector spaces:\footnote{Analogously to the construction of higher fusion categories, we obtain super 
($n$)-vector spaces via condensation completion, beginning with the fusion 1-category of super vector spaces  $\mathbf{sVect}$.}
\begin{equation}\label{eq:superobstructionSES}
   B \mathcal{S}\mathbf{C}^\times \longrightarrow B\mathcal{S}\sAut^\otimes(\mathbf{C}) \longrightarrow B\mathcal{S}\mathbf{Bimod}(\mathbf{C})^\times\,.
\end{equation}

\begin{example}\label{ex:bosonicobstruction}
    Let $\mathbf{C}$  in \eqref{eq:obstructionSES} be a connected fusion 2-category of the form $\Mod(\cB)$ where $\cB$ is a nondegenerate braided fusion 1-category. For background on the foundations of fusion 2-categories, we recommend \cite{douglas2018fusion}. Then we get the sequence
    \begin{equation}
        B \sPic(\cB) \longrightarrow B\sAut^{br}(\cB) \longrightarrow B\Bimod(\Mod(\cB))^\times\,,
    \end{equation}
    and hence the 3-groupoid $B\Bimod(\Mod(\cB))^\times$ parametrizes obstructions, which is isomorphic to $B\mathscr{W}itt: = B3\Vect^\times$. The homotopy groups of $\mathscr{W}itt$ are simply \cite{ENO2}
    \begin{equation}
\begin{aligned}
    \pi_{0} \,\mathscr{W}itt= &\mathcal{W}itt, \quad \pi_{1} \, \mathscr{W}itt= \pi_{2} \, \mathscr{W}itt= \pi_{3} \, \sWitt= 0, \quad \pi_{4} \, \mathscr{W}itt= \mathbb{C}^\times,
\end{aligned}
\end{equation}
    where $\mathcal{W}itt$ is the Witt group \cite{brochier2021invertible} of \textit{nondegenerate} braided fusion categories. Let $\mathscr{W}^*$ denote the generalized cohomology theory corresponding to the spectrum whose $n$-th space is $B^{n-4}\mathscr{W}itt$. Then from our perspective, the obstruction should take values in $\mathscr{W}^4(BG)$, and we have a natural comparison map\footnote{For any pointed space $X$, we have the natural inclusion and retraction ${\mathrm{pt}} \rightarrow X \rightarrow {\mathrm{pt}}$. For any generalized cohomology theory $E$, this gives the natural split $E^*(X) = E^*(\mathrm{pt}) \oplus \widetilde{E}^*(X)$, where $\widetilde{E}^*$ is the \textit{reduced} generalized cohomology theory for $E$.}
    \begin{equation}\label{eq:sequence}
        H^4(BG;\mathbb C^\times) \rightarrow \mathscr{W}^4(BG) \rightarrow H^0(BG;\mathcal{W}itt).
    \end{equation}
    The $H^4(BG;\mathbb C^\times)$ part is what is commonly referred to as the ``$G$-anomaly'' for bosonic topological orders in (2+1)d, and is the obstruction to the associativity of the extension \cite{ENO2}. The $H^0(BG;\mathcal{W}itt)$ part simply encodes the information of the Witt class of the underlying nondegenerate braided fusion category under consideration. Therefore, our perspective aligns with the more common perspective. Yet it provides a unifying framework that can be generalized to higher dimensions and more complicated settings.
    
    Interestingly, when $G$ does not contain any anti-unitary symmetry, $\mathscr{W}^4(BG)$ is canonically isomorphic to $H^4(BG;\mathbb C^\times) \oplus H^0(BG;\mathcal{W}itt)$. It is very interesting to investigate whether the split holds when anti-unitary symmetries or $s$-twist are present.

    We also conjecture that there is a natural map from $\mathscr{W}^4(BG)\rightarrow I_\Z \MTSO^4(BG)$, and the image should give the 't Hooft anomaly of a $(2+1)$d bosonic topological order given by the calculation in e.g. \cite{Bulmash:2020flp,Ye:2022bkx}.
\end{example}

As discussed in \S\ref{subsection:TO}, (3+1)d fermionic topological orders are described by nondegenerate $\tsVect$-enriched braided fusion 2-categories $\mathfrak{B}$. When one takes $\mathbf{C} = \Mod(\fB)$ in \eqref{eq:obstructionSES}, then the categorical obstruction to performing a faithfully $G$-crossed  braided extension is parametrized by the 4-groupoid $B \sWitt: = B\mathcal{S}\Bimod(\Mod(\fB))^\times = B\mathbf{4sVect}^\times$.  The details of the enrichment over $\tsVect$ and the appearance of this groupoid are presented in \cite[Section 4]{DY2025}.

The homotopy groups  of $\sWitt$ were computed in \cite{DY2025}, and given by: 
\begin{equation}
\begin{aligned}
    \pi_{0} \,\sWitt= &s\mathcal{W}, \quad \pi_{1} \, \sWitt= 0, \quad \pi_{2} \, \sWitt= \Z/2, \\
     &\pi_{3} \, \sWitt= \Z/2, \quad \pi_{4} \, \sWitt= \mathbb{C}^\times,
\end{aligned}
\end{equation}
where $\sW$ is the super-Witt group of  braided fusion categories $\cB$ with Müger center $\sVect$, given in \cite{DNO}. Such categories are also referred to as slightly degenerate braided fusion categories.  By \cite[Proposition 5.18]{DNO} we have
\begin{equation}\label{eq:swittdecompose}
    \sW = \sW_{\pt} \oplus \sW_2 \oplus \sW_{\infty}\,,
\end{equation}
where $\sW_{\pt}$ is generated by
the Witt classes of Abelian super MTCs, $\sW_2$ is an elementary Abelian 2-group, and $\sW_{\infty}$ is a free group of countable rank.
Determining the $k$-invariants of the space $\sWitt$ is an important open question, especially in the context of this work for computing categorical obstructions.
\begin{defn}
Let $\SW^*$ denote the generalized cohomology theory corresponding to the spectrum whose $n$-th space is $B^{n-4}\sWitt$.
\end{defn}
Thus $\SW^n(BG)$ parametrizes homotopy classes of maps
\begin{equation}
    BG \rightarrow B^{n-4} \sWitt,
\end{equation}
and is exactly the categorical obstruction that we are seeking for.

In the relevant range for our applications, $\SW$ has the following homotopy groups:
\begin{equation}
\begin{aligned}
    \pi_{-4} \,\SW= &s\mathcal{W}, \quad \pi_{-3} \, \SW= 0, \quad \pi_{-2} \, \SW= \Z/2, \\
     &\pi_{-1} \, \SW= \Z/2, \quad \pi_0 \, \SW= \mathbb{C}^\times.
\end{aligned}
\end{equation}
The categorical obstruction given by $\SW$ resembles the more familiar 't Hooft anomalies that are classified by $I_\Z \MSpin$, which in the relevant range, has homotopy groups \begin{equation}
\begin{aligned}
    &\quad \pi_{-4} \,I_\Z \MSpin= \Z, \quad \pi_{-3} \, \,I_\Z \MSpin= \Z/2, \\  \pi_{-2} \, I_\Z &\MSpin= \Z/2, \quad \pi_{-1} \, I_\Z \MSpin= 0, \quad \pi_0 \, I_\Z \MSpin= \Z.
\end{aligned}
\end{equation}
By comparing $\SW$ with $I_\Z\MSpin$, it is conjectured \cite{Theo_private} that there exists a map from the categorical obstruction to the  't Hooft anomaly, i.e.\ a map 
\begin{equation}\label{eq:SWtoMSpin}
p\colon \SW \rightarrow \Sigma I_\Z \MSpin.
\end{equation} 
which maps nondegenerate braided fusion ($n$)-categories enriched in super ($n$)-vector spaces, to reflection positive invertible spin TQFTs. \textbf{In the rest of \S\ref{section:anomaliesObstructions}, we assume this conjecture.}
We summarize a heuristic construction for part of this map due to what we learned in \cite{Theo_private}.
Comparing the homotopy groups of the spectrum $\SW$ and the spectrum of $I_\Z \MSpin$, we have
\begin{center}
\begin{tabular}{c|c|c}
$\pi_*$ & $\SW$ & $\Sigma I_\Z \MSpin$ \\
\hline
$+1$ & $0$       & $\mathbb{Z}$ \\
$0$  & $\mathbb{C}^\times$ & $0$ \\
$-1$ & $\mathbb{Z}/2$ & $\mathbb{Z}/2$ \\
$-2$ & $\mathbb{Z}/2$ & $\mathbb{Z}/2$ \\
$-3$ & $0$       & $\mathbb{Z}$ \\
$-4$ & $s\mathcal{W}$ & $0$
\end{tabular}
\end{center}
In degrees $-2,\dotsc,+1$, the two spectra are determined (noncanonically) by their homotopy groups together with the fact that the Postnikov $k$-invariants of consecutive homotopy groups are all nontrivial \cite[Section 5]{GJF19}. In this range of degree, $\SW$ looks like  $\Sigma I_{\Z}\MSpin$, except that $\C^\times$ is replaced with $\Z$ in one degree higher. Indeed, after truncating to degrees $-2$ and above, the map~\eqref{eq:SWtoMSpin} ``is'' the cofiber of the exponential map $\C\to\C^\times$, in that the fiber of~\eqref{eq:SWtoMSpin} is the Eilenberg--Mac Lane spectrum $H\C$. This says that the map~\eqref{eq:SWtoMSpin} is very close to being an equivalence: in degrees $-1$ and below, it is an isomorphism on homotopy groups, and in degrees $0$ and $1$, it is a Bockstein.

In these degrees, it is possible to describe the map~\eqref{eq:SWtoMSpin} field-theoretically: in principle, this map describes how every invertible object of $\Omega^2\mathbf{4sVect}^\times\simeq \mathbf{sAlg}^\times$, the Morita $2$-category of superalgebras, gives rise to a two-dimensional reflection-positive invertible spin TFT. This is standard: the unit in $\mathbf{sAlg}^\times$ gives rise to the trivial theory, and the unique nontrivial Morita class, represented by the Clifford algebra $\mathit{C\ell}_1$, gives rise to the Arf theory~\cite{Gun16}.

It remains to address the maps in degrees $-3$ and $-4$. The existence of such a map was communicated to us in \cite{Theo_private}, and progress on mapping the torsion part of $\sW$ to the degree $-3$ entry in $\Sigma I_\Z \MSpin$ has been announced in \cite{reutter:youtube}.

\subsection{Relationship between $\SH$, $\SW$, $\mho_\Spin$ in degree 5 and the main conjecture}
In the setting of our paper, we would like to understand the relationship between $\SH$, $\SW$, and $\mho_\Spin$ in degree $5$ when applied to $BG$ for a finite group $G$. Thus consider the maps
\begin{equation}
    \SH^5(BG)\overset{\mathcal I}{\longrightarrow} \SW^5(BG) \overset p\longrightarrow \mho^5_\Spin(BG)\,,
\end{equation}
where  the map $\mathcal I\colon \SH \to \SW$ is the Postnikov $(-3)$-connected cover.
\begin{lem}
    If the map $(p\circ\mathcal I)_*\colon \SH^5(BG)\rightarrow \mho^5_\Spin(BG)$ is an isomorphism for a given group $G$, then there is a subgroup $A$ of $H^1(BG;\sW)$ and a splitting $\SW^5(BG) \cong \SH^5(BG) \oplus A$ of the map $\mathcal I_*\colon\SH^5(BG)\to\SW^5(BG)$.
\end{lem}
\begin{proof}
The map $\mathcal I\colon \SH\to\SW$ of spectra is an isomorphism on homotopy groups in all degrees $-3$ and above, so its cofiber is the Postnikov quotient $\tau_{\le(-4)}\SW$. As this spectrum has only one nonzero homotopy group $\pi_4(\tau_{\le(-4)}\SW)\cong \sW$, it must be an Eilenberg--Mac Lane spectrum: $\tau_{\le(-4)}\SW\simeq \Sigma^{-4}H\sW$. That is, we have a fiber sequence
\begin{equation}
\label{SH_SW_fib}
    \SH\overset{\mathcal I}{\longrightarrow} \SW \overset{\tau_{\le(-4)}}\longrightarrow \Sigma^{-4}H\sW.
\end{equation}
Combining the induced long exact sequence from~\eqref{SH_SW_fib} with the data from the lemma statement, we have the following commutative diagram, where the top row is exact:
\begin{equation}\label{I_tau_LES}
\begin{tikzcd}
	\dotsb & {H^0(BG; \sW)} & {\SH^5(BG)} & {\SW^5(BG)} & {H^1(BG; \sW)} & \dotsb \\
	&&& {\mho_\Spin^5(BG)}
	\arrow[from=1-1, to=1-2]
	\arrow["{\delta^0}", from=1-2, to=1-3]
	\arrow["{\mathcal I}", from=1-3, to=1-4]
	\arrow["\cong"{description}, from=1-3, to=2-4]
	\arrow["{\tau_{\le(-4)}}", from=1-4, to=1-5]
	\arrow["p", from=1-4, to=2-4]
	\arrow["{\delta^1}", from=1-5, to=1-6]
\end{tikzcd}\end{equation}
Since $p\circ\mathcal I\colon\SH^5(BG)\to\mho_\Spin^5(BG)$ is an isomorphism by hypothesis, it provides a section of $\mathcal I\colon \SH^5(BG)\to\SW^5(BG)$. Thus $\mathcal I$ is a split injection. Because the sequence in~\eqref{I_tau_LES} is exact, $A\coloneqq \ker(\delta^1)\subset H^1(BG;\sW)$ is a complementary summand to the image of $\mathcal I$, which finishes the proof.

\end{proof}

Furthermore, if $\SH^5(BG) \rightarrow I_\Z \MSpin$ is surjective then $\SH^5(BG)$ captures a subgroup of $\SW^5(BG)$.

We now summarize the arguments for and against using these two types of obstructions, as well as supercohomology, to build a (3+1)d topological order:
\begin{itemize}
    \item  While the correct obstruction to  $G$-crossed braided extensions of a $\tsVect$-enriched nondegenerate braided fusion 2-category are classes in $\SW^5(BG)$, computing this group is hard because we do not know about higher differentials in the Atiyah--Hirzebruch spectral sequence that computes $\SW^5(BG)$; in particular, this is tied to the fact that the $k$-invariants of $B\sWitt$ are not fully determined.  Even assuming the existence of the map $\SW \to I_\Z \MSpin$ does not mean we can necessarily pull back differentials, as some elements may be sent to zero.
    \item Spin cobordism is usually tractable to compute by the Adams spectral sequence. It also classifies anomalies for continuous quantum field theories. However, it is less directly related to the categorical obstructions. Like $\SW^5(BG)$, spin cobordism also has a contribution that goes beyond cohomology starting in dimension 4, which we do not know of a good cocycle description for.\footnote{See Brumfiel--Morgan~\cite{BM16, BM18} for cocycle descriptions of $I_{\C^\times}\MSpin$ in lower degrees.}
    \item Using the hastened Adams spectral sequence for supercohomology that we develop in~\cite{DYY2}, supercohomology is roughly as computable as $I_\Z \MSpin$; see Appendix~\ref{appendix:B}. Supercohomology also has a cocycle description \cite{PhysRevB.90.115141,Wang:2017moj}. Therefore  it is both possible in theory and tractable in practice to apply the fermionic Wang--Wen--Witten construction on supercohomology, with the hopes of writing down a state-sum that generalizes \cite{Kobayashi:2019lep}.
    Furthermore, we know that (3+1)d fermionic topological orders are classified by degree 4 supercohomology classes \cite[Corollary V.4]{JF}.
     Supercohomology is an approximation not only to $\SW$, but also to spin cobordism in low degrees using the first definition of supercohomology in Appendix \ref{subsection:twistedSH}.  Hence $\SH^5(BG)$ may contain classes which map to 0 in $\mho^5_\Spin(BG)$, however the two may at times also coincide. In the case when they do coincide, $\SH^5(BG)$ really does have an interpretation in terms of classifying fermionic $G$-SPTs. 
See \Cref{ex:cyclic} for an example when the two groups coincide, and \cref{ex:timerev} for an example where the two groups do not coincide. 

    We do not know exactly how much $\SH^5(BG)$ misses of the full categorical anomaly given by $\SW^5(BG)$. To fully answer this question we would need to understand how to compute $\SW^5(BG)$, which is a difficult open problem. Finding a cocycle description of this group is expected to be even harder. Thus, we will ignore the bottom layer with $\sW$ in our approximation to the categorical obstruction.\footnote{It would be interesting to expand the definition of fusion 2-categories to incorporate unitarity, and compare if the analogous obstructions with and without restriction from unitarity.}
\end{itemize}

\begin{rem}
    There is the natural question of what it actually means to give a state sum construction for a TQFT whose Lagrangian description involves a class in $\SW^5(BG)$, which contains the group $\sW$. We believe this question to be related to realizing discrete invertible phases with ``SPT index'' valued in $\mho^5_\Spin(BG)$. In spacetime  dimension three or lower, one could define an SPT index valued in $\mho^3_\Spin(BG)$ via the cocycles $(\alpha,\beta,\gamma)$ of supercohomology. But it is not known how to go to higher dimensions. In particular, one should provide an answer for how to work with  a ``cocycle'' valued in $\sW$. Such a cocycle should have the interpretation as the super Witt class of a (2+1)d topological order with a $G$-symmetry. Such Witt classes are defined in \cite[Definition 5.2.3]{Bhardwaj:2024xcx}. Trivializing a cocycle upon pulling back to a group $H$ would mean that the (2+1)d topological order with a $G$-symmetry is Witt trivial in the class of (2+1)d topological order with a $H$-symmetry.
\end{rem}

We now discuss how these three obstructions come together in an example involving (2+1)d fermionic TQFTs. 

\begin{example}
    In analogy to \cref{ex:bosonicobstruction}, the categorical obstruction for a $G$-crossed braided extension of a slightly degenerate braided fusion category $\cA$ is given by an element in $\SH^4(BG)$, as shown in \cite{Decoppet:2024htz}. However, this again misses the anomaly given by the Witt class $[\cA] \in \sW$.  Taking the anomaly from the Witt class into account would make this example line up with the conjecture that there is a map from $\SW \to \Sigma I_\Z \MSpin$ with properties as described above. In the case where $G$ is a unitary symmetry, we have a match between $\SH^4(BG)$ and $\widetilde{\mho}^4_\Spin(BG)$, where the latter denotes reduced spin cobordism, and $\SW^4(BG)$ splits as $\SH^4(BG) \oplus \sW$.
\end{example}

\section{Spectral sequence computations}\label{appendix:B}
In this appendix, we provide the technical computations involving the hastened Adams and Atiyah--Hirzebruch spectral sequences used in \S\ref{section:construction} to prove the main theorems. 

Throughout this appendix, we make a technical assumption: that for all $(X, a, b)$-twisted supercohomology groups that we consider, there is a vector bundle $V\to X$ such that $w_1(V) = a$ and $w_2(V) = b$. This is true, and straightforward to verify, for all examples appearing in this paper.\footnote{\label{nonvb_footnote}This assumption is not true in general: see~\cite{GKT89, RWG14, TJF19, Kuhn20, Speyer22, DY24} for counterexamples where $X$ is the classifying space of a compact Lie group. For the (hastened) Adams spectral sequence, this assumption is unnecessary~\cite{DY23,DYY2}; for the Atiyah--Hirzebruch spectral sequence, this assumption is used to prove the formulas~\eqref{dform} for differentials. We conjecture that these formulas hold even without this assumption, but this is not in the literature to our knowledge.} 

First, we provide details about the AHSS for the groups we consider in this paper. For a fermionic symmetry group given by $(G, s, \omega)$ such as in \Cref{tab:results}, the entries of the AHSS on the $E_2$-page are given by:
\begin{equation}\label{eq:layout}
\resizebox{\linewidth}{!}{$
E^{i,j}_2=
\begin{array}{c|cccccccc}
     j\\
     \\
     2& H^0(BG; \bZ/2) & H^1(BG; \bZ/2) & H^2(BG; \bZ/2) & H^3(BG; \bZ/2) &  \ldots \\
     1& H^0(BG; \bZ/2) & H^1(BG; \bZ/2) & H^2(BG; \bZ/2) & H^3(BG; \bZ/2) & H^4(BG; \bZ/2) &  \ldots \\
     0 & H^0(BG; \mathbb C^\times _ s) & H^1(BG; \mathbb C^\times _ s) & H^2(BG; \mathbb C^\times _ s) & H^3(BG; \mathbb C^\times _ s) & H^4(BG; \mathbb C^\times _ s) & H^5(BG; \mathbb C^\times _ s) &  \ldots \\
     \hline 
     & 0 & 1 & 2 & 3 & 4 & 5  & i\,
\end{array}
$}
\end{equation}

The rows for $j=2$ and $j=1$ come from the $\Z/2$ coefficient ring of the space $BG$, and we will write the entries/generators in terms of generators of the $\Z/2$ coefficient ring. The subscript $s$ in the $j=0$ row is an indicator that $G$ has nontrivial action on the $\mathbb{C}^\times$ module determined by $s$. We can write the elements in the $j=0$ row using elements in $H^*(BG;\Z/2)$ with the help of the map $\bZ/2 \rightarrow \mathbb{C}^\times$. This allows us to present the entries as $(-1)^x$ where $x$ is an element in $\bZ/2$ cohomology. When it is not possible, we will only write down the explicit group of the entry without giving a name to the corresponding generator of the entry.

Since supercohomology is a Postnikov truncation of the Pontryagin dual of $\ko$, the $d_2$ differentials in the supercohomology Atiyah--Hirzebruch spectral sequence follow from the respective differentials in the $\ko$-AHSS, which were computed by Bott~\cite{Bot69}.\footnote{\label{AHSS_diffs}For an explicit statement of these differentials, see Anderson--Brown--Peterson~\cite[Proof of Lemma 5.6]{ABP67}. In addition, see~\cite[Lemma A.23]{Debray:2023iwf} for the details on passing the differentials through Anderson duality.}
First, define twisted Steenrod squares acting on $H^*(X;\Z/2)$ by
\begin{subequations}
\label{twisted_Sq}
\begin{align}
    \Sq_s^1(x) &\coloneqq \Sq^1(x) + sx\\
    \Sq_{s,\omega}^2(x) &\coloneqq \Sq^2(x) + s \Sq^1(x) + \omega x.
\end{align}
\end{subequations}
Then the Atiyah--Hirzebruch $d_2$s have the formula
\begin{subequations}
\label{dform}
\begin{align}
    d_2 &: E^{i,2}_2  \to E^{i+2,1}_2 && X \mapsto \Sq_{s,\omega}^2 (X) \,, \\
     d_2 &: E^{i,1}_2 \to E^{i+2,0}_2 \quad&& X \mapsto (-1)^{\Sq_{s,\omega}^2 (X)} \,.
\end{align}
\end{subequations}
There is also potentially a nontrivial $d_3$ differential 
\begin{equation}
    d_3\colon E_3^{i, 2} \to E_3^{i + 3, 0}.
\end{equation}
We do not know an explicit formula for general $i$.\footnote{Results for low degrees based on physical constructions can be found in \cite{PhysRevX.10.031055,2022PhRvB.105w5143B,2025arXiv251225069N}; see also~\cite{MO_KO_d3}.} On the $E_\infty$-page we must resolve potential extension problems, i.e., resolve how different entries in different rows are assembled together to give the full supercohomology group. This can be especially tricky, especially when the total degree is higher than 3.

At this point, it is traditional in the mathematical physics literature to turn to the \term{Adams spectral sequence}, whose structure makes many extension problems easier, and which admits a remarkable simplification for computing twisted spin bordism (see~\cite{BC18}). However, we need to compute twisted supercohomology, for which the standard Adams spectral sequence is messier. Instead, we use a variant called the \term{hastened Adams spectral sequence} (HASS). HASSes were introduced in~\cite{BHHM08} and systematized in~\cite{BR21} associated to the general data of a map of spectra; in a companion paper~\cite{DYY2} we apply this to supercohomology and study a number of examples. Here, we give an overview of the HASS for $\tau_{\le 2}\ko$, then apply it in several examples.

Let $\cA(1)$ denote the subalgebra $\langle \Sq^1, \Sq^2\rangle$ inside the Steenrod algebra $\cA$ of mod $2$ stable cohomology operations, and let $H_{s,\omega}^*(X;\Z/2)$ be the $\cA(1)$-module whose underlying graded vector space is $H^*(X;\Z/2)$, but where $\Sq^1$ acts by $\Sq_s^1$ and $\Sq^2$ acts by $\Sq_{s,\omega}^2$ (see~\eqref{twisted_Sq}).\footnote{It is not immediately obvious that $\Sq_s^1$ and $\Sq_{s,\omega}^2$ satisfy the Adem relations and thus define an $\cA(1)$-action; this was shown in~\cite[Lemma 2.38(3)]{DY23}.} Then the input data to the Adams spectral sequence computing $(X, s, \omega)$-twisted spin bordism is
\begin{equation}
\label{Adams_sketch}
    E_2^{s,t} = \Ext^{s,t}_{\cA(1)}(H_{s,\omega}^*(X;\Z/2), \Z/2),
\end{equation}
where $\Ext$ is a functor classifying extensions of $\cA(1)$-modules of different lengths. In this paper, when we write $\Ext(M)$ we mean $\Ext_{\cA(1)}^{*,*}(M, \Z/2)$.

In the hastened Adams spectral sequence for (the dual of) supercohomology, most of~\eqref{Adams_sketch} is the same, but $\Ext$ is replaced with a different functor $\mathcal Q$, which one can think of as a ``difference of two Exts.'' The following theorem makes this precise.
\begin{thm}[{\cite[Proposition 12.33]{BR21}, \cite{DYY2}}]
\label{cQthm}
Let $\uQ$ denote the $\cA(1)$-module $\cA(1)/(\Sq^1, \Sq^2\Sq^3)$.
\begin{enumerate}
    \item\label{g4_defn}
    There is a map of $\Z^2$-graded $\Ext(\Z/2)$-modules
    \begin{equation}
        g_4\colon\Ext_{\cA(1)}^{s,t}(\uQ, \Z/2) \longrightarrow
            \Ext_{\cA(1)}^{s+3, t+2}(\Z/2, \Z/2)
    \end{equation}
    which is induced from the Postnikov cover map $\tau_{\ge 4}\ko\to\ko$.
    \item \label{cQ_defn}
    There is a functor $\mathcal Q^{*,*}$ from $\cA(1)$-modules to $\Z^2$-graded $\Ext(\Z/2)$-modules which commutes with direct sums and such that for all $\cA(1)$-modules $M$, there is a long exact sequence
    \begin{equation}
    \label{cQLES}
    \dotsb\to \Ext_{\cA(1)}^{s,t}(\uQ\otimes M, \Z/2) \overset{g_4}{\longrightarrow}
            \Ext_{\cA(1)}^{s+3, t+2}(M, \Z/2)
    \longrightarrow
    \mathcal Q^{s,t}(M) \longrightarrow
    \Ext_{\cA(1)}^{s+1,t}(\uQ\otimes M, \Z/2)\overset{g_4}{\to}\dotsb
    \end{equation}
    \item\label{HASS_defn}
    Let $X$ be a space of finite type,\footnote{The finite-type hypothesis appears for technical reasons and holds in all circumstances one might reasonably encounter in mathematical physics.} $s\in H^1(X;\Z/2)$, and $\omega\in H^2(X;\Z/2)$. Then the HASS for $\tau_{\le 2}\ko_*(X, s, \omega)$ converges strongly and has signature
    \begin{equation}
        E_2^{s,t} = \mathcal Q^{s,t}(H_{s,\omega}^*(X; \Z/2)) \Longrightarrow \tau_{\le 2}\ko_{t-s}(X, s, \omega)_2^\wedge.
    \end{equation}
    The map $\ko_*(X, s, \omega)\to\tau_{\le 2}\ko_*(X, s, \omega)$ lifts to a map from the ordinary Adams spectral sequence to the HASS.
\end{enumerate}
\end{thm}
Because $\mathcal Q$ commutes with direct sums and fits into the sequence~\eqref{cQLES}, it is straightforward to compute it on $\cA(1)$-modules of interest. In~\cite{DYY2}, we compute $\mathcal Q$ on many common $\cA(1)$-modules, and we use this to compute the $E_2$-pages of the HASSes we use below. Once we have done this, running the HASS is just as in the usual Adams spectral sequence.

\subsection{Example: $\SH^5(B\Z/2)$}
We first compute $\SH^5(B\Z/2)$, which we use in \cref{ex:cyclic}. The $\Z/2$ cohomology ring of $B\Z/2$ is given by
\begin{equation}\label{Z2Z2coh}
    H^*(B\Z/2; \bZ/2) = \Z/2[x], \quad |x| = 1 
\end{equation}
where $x$ is the nontrivial generator of $H^1(B\Z/2; \bZ/2)$.

\begin{prop}\label{prop:Z2notwist}
The group $\SH^\ell(B\Z/2)=0$ for $\ell = 4,5$.
\end{prop}
For $\ell = 4$ this is due to Décoppet~\cite[Example 4.13]{D11}; for $\ell = 5$ this is new.
\begin{proof}
The AHSS which computes this group has the following $E_2$-page:
\begin{equation}\label{eq:SHZ2}
E^{i,j}_2=\begin{array}{c|cccccccc}
     j\\
     \\
     2 & 1 & x & x^2 & x^3 &  \ldots \\
     1 & 1 & x & x^2 & x^3  & x^4 & \ldots  \\
     0 & \mathbb{C}^\times & (-1)^x & 0 & (-1)^{x^3}& 0 & (-1)^{x^5} &0 & \ldots \\
     \hline 
     & 0 & 1 & 2 & 3 & 4 & 5 & 6 & i\,
\end{array}
\end{equation}
The $d_2$ differentials are given by:
\begin{subequations}
\begin{align}\label{untwistedd2}
    d_2 &: E^{i,2}_2  \to E^{i+2,1}_2 && X \mapsto \Sq^2 X \,, \\ \label{seconduntwisted2}
     d_2 &: E^{i,1}_2 \to E^{i+2,0}_2 && X \mapsto (-1)^{\Sq^2 X} \,.
\end{align}
\end{subequations}
After resolving the $d_2$ differential, the $E_3$-page is given as follows:
\begin{equation}\label{eq:E3SH_Z2}
E^{i,j}_3=\begin{array}{c|cccccccc}
     j\\
     \\
     2 & 1 & x & 0 & 0 &  \ldots \\
     1 & 1 & x & x^2 & 0 & 0 &\ldots  \\
      0 & \mathbb{C}^\times & (-1)^x & 0 & (-1)^{x^3} & 0 & 0 & 0 & \ldots \\
     \hline 
     & 0 & 1 & 2 & 3 & 4 & 5 & 6 & i\,
\end{array}
\end{equation}
In particular, $\SH^\ell(B\Z/2)$ is trivial for $\ell =4,5$.
\end{proof}
See Wang--Gu~\cite[Table II]{Wang:2017moj}, Gaiotto--Johnson-Freyd~\cite[\S 4]{GJF22}, and Yu~\cite[\S 2.8]{Yu21} for $\SH^\ell(B\Z/2)$ when $\ell<4$.
\subsection{Example: $\SH^5(B\Z/2^k), k\geq 2$}
We now compute  $\SH^5(B\Z/2^k)$ for $k\geq 2$. This is also relevant in Example \ref{ex:cyclic}. The $\Z/2$ cohomology ring of $B\Z/2^k$ is given by
\begin{equation}
    H^*(B\Z/2^k; \bZ/2) = \Z/2[x, y]/(x^2), \quad |x| = 1, |y| = 2. 
\end{equation}

\begin{prop}\label{prop:2knotwist}
    For $k\geq 2$, the group $\SH^5(B\Z/2^k) \cong \Z/2^{k-1}$, with generator in the Dijkgraaf--Witten layer.
\end{prop}

\begin{proof}
The $E_2$-page of the AHSS that compute this group is given as follows:
\begin{equation}\label{eq:SHZ2k}
E^{i,j}_2=\begin{array}{c|cccccccc}
     j\\
     \\
     2 & 1 & x & y & xy &  \ldots \\
     1 & 1 & x & y & xy  & y^2  &\ldots  \\
      0 & \mathbb{C}^\times & \Z/2^k & 0 & \Z/2^k & 0 & \Z/2^k &0 & \ldots \\
     \hline 
     & 0 & 1 & 2 & 3 & 4 & 5 & 6 & i\,
\end{array}
\end{equation}
with the same $d_2$ differential as in Equations \eqref{untwistedd2} and \eqref{seconduntwisted2}. After resolving the $d_2$ differential, the $E_3$-page is given as follows:
\begin{equation}\label{eq:E3SH_Z4}
E^{i,j}_3=\begin{array}{c|cccccccc}
     j\\
     \\
     2 & 1 & x & 0 & 0 &\ldots \\
     1 & 1 & x & y & 0 & 0 &\ldots  \\
      0 & \mathbb{C}^\times & \Z/2^k & 0 & \Z/2^k & 0 & \Z/2^{k-1} &0 & \ldots \\
     \hline 
     & 0 & 1 & 2 & 3 & 4 & 5 & 6 & i\,
\end{array}
\end{equation}
We see that $\SH^5(B\Z/2^k) = \Z/2^{k-1}, k\geq 2$ and the whole group is in the Dijkgraaf--Witten layer.
\end{proof}
See Wang--Gu~\cite[Table II]{Wang:2017moj} for $\SH^\ell(B\Z/2^k)$ when $\ell<4$ and Décoppet~\cite[Example 4.13]{D11} for $\ell = 4$. Specifically, Wang--Gu's work resolves the extension question in total degree $3$ in~\eqref{eq:E3SH_Z4}, which we will need to use later in this article.
\begin{prop}[{Wang--Gu~\cite[Table II]{Wang:2017moj}}]
\label{prop:Zknotwist3}
For $k\ge 2$, $\SH^3(B\Z/2^k)\cong \Z/2^{k+1}\oplus\Z/2$.
\end{prop}
\subsection{Example: $\SH^5(B\Z/2,0,x^2)$}\label{subsub:twistedappendix}
We now combine the hastened Adams spectral sequence and the Atiyah--Hirzebruch spectral sequence to compute $\SH^5(B\Z/2,0,x^2)$ and determine the filtration of its generators. We will use these results in \cref{ex:cyclicfermion}. Recall $H^*(B\Z/2;\Z/2)$ from~\eqref{Z2Z2coh}.

\begin{prop}\label{prop:Z2x2twist}
$\SH^\ell(B\Z/2,0,x^2)$ is isomorphic to $\Z/2$ for $\ell = 4$ and $\Z/8$ for $\ell = 5$. In the latter case, there is a generator of the $\Z/8$ residing in the Majorana layer.
\end{prop}
The case $\ell = 4$ verifies a prediction of Décoppet~\cite[Example 4.13]{D11}.
\begin{proof}
The $E_2$-page is the same as Equation \eqref{eq:SHZ2}, with the twisted $d_2$ differentials given by
\begin{subequations}
\begin{align}
   d_2 &: E^{i,2}_2  \to E^{i+2,1}_2 && X \mapsto \Sq^2 X+x^2 X \,, \\ 
     d_2 &: E^{i,1}_2 \to E^{i+2,0}_2 \quad&& X \mapsto (-1)^{\Sq^2 X + x^2 X} \,.
\end{align}
\end{subequations}
After resolving the $d_2$ differentials, the $E_3$-page is given as follows:
\begin{equation}\label{eq:SH}
E^{i,j}_3=\begin{array}{c|cccccccc}
     j\\
     \\
     2 & 0 & 0 & x^2 & x^3 &  \ldots \\
     1 & 1 & 0 & 0 & 0 & x^4 &\ldots  \\
      0 & \mathbb{C}^\times & (-1)^x & 0 & 0 & 0 & (-1)^{x^5} &0 & \ldots \\
     \hline 
     & 0 & 1 & 2 & 3 & 4 & 5 & 6 & i\,
\end{array}
\end{equation}

As we do not know about potential $d_3$ differentials or extensions in this spectral sequence, we turn to the hastened Adams spectral sequence. In this and future HASS arguments, we assume some background with the ordinary Adams spectral sequence; Beaudry--Campbell's article~\cite{BC18} is an excellent introduction covering everything we assume.

In Figure \ref{fig1:HASS}, left, we display the $\cA(1)$-module $R_1\coloneqq H^*_{0,x^2}(B\Z/2;\Z/2)$, which was calculated by Campbell~\cite[Figure 7.2]{Cam17}. By \cref{cQthm}, part~\ref{HASS_defn}, the $E_2$-page of the HASS for $(B\Z/2,0,x^2)$-twisted $\tau_{\le 2}\ko$-homology is $\mathcal Q^{*,*}(H^*_{0,x^2}(B\Z/2; \Z/2))$. Using the long exact sequence~\eqref{cQLES}, we calculate this $E_2$-page in~\cite{DYY2}, and we display it in \cref{fig1:HASS}, right. In degree 5, there is  room for a $d_2$ differential $d_2(m_2) = \mu_2$. In fact, though, this differential vanishes, because $m_2$ is in the image of the map of Adams spectral sequences induced by $\ko\to\tau_{\le 2}\ko$~\cite{DYY2}, and in the Adams spectral sequence for the corresponding twist of $B\Z/2$ over $\ko$, $d_2(m_2) = 0$~\cite[\S 7.8]{Cam17}.
Therefore on the $E_\infty$-page we have  $(\tau_{\le 2}\ko)_4(B\Z/2,0,x^2) = \Z/2$ and $(\tau_{\le 2}\ko)_5(B\Z/2,0,x^2) = \Z/8$. The corresponding twisted supercohomology is the Pontryagin dual group.
Thus the $d_3$ mentioned previously in the AHSS vanishes. 

\begin{figure}[H]
\begin{subfigure}[c]{0.4\textwidth}
\begin{tikzpicture}[scale=0.4]
	\foreach \y in {0, ..., 8} {
		\tikzpt{0}{\y}{}{};
	}
	\foreach \y in {1, 3, 5, 7} {
		\sqone(0, \y);
	}
	\sqtwoL(0, 0);
	\sqtwoL(0, 4);
	\sqtwoR(0, 1);
	\sqtwoR(0, 5);
	\begin{scope}
		\clip (-1, 6.8) rectangle (1, 8.5);
		\sqtwoL(0, 8);
	\end{scope}
\end{tikzpicture}
\end{subfigure}
\begin{subfigure}[c]{0.4\textwidth}
\begin{sseqdata}[name=QR1, classes = fill, scale=0.4, xrange={0}{8}, yrange={0}{4},
x tick step = 2,
y tick step = 2,
x tick gap = 0.25cm,
y tick gap = 0.25cm,
x axis tail = 0.2cm,
y axis tail = 0.2cm,
y axis clip padding = 0cm,
x axis clip padding = 0cm,
class labels = {below = 0.07em, font=\tiny },
]
\class["\mu_1"](0, 0)\AdamsTower{}
\class(1, 1)\structline(0, 0)(1, 1)
\class["m_1"](1, 0)\structline

\class["\mu_2"](4, 2)
\class["m_2"](5, 0)
\class(5, 1)\structline
\class(5, 2)\structline
\class["m_1'"](8, 2)
\end{sseqdata}
\printpage[name=QR1, page=2]
\end{subfigure}
\caption{Left: The $\cA(1)$-module structure on $R_1 \cong H_{0,x^2}^*(B\Z/2; \Z/2)$~\cite[Figure 7.2]{Cam17}. Right: $\mathcal{Q}(R_1)$, computed in~\cite{DYY2}.}
\label{fig1:HASS}
\end{figure}
Combining with the Atiyah--Hirzebruch spectral sequence, we see that the generator of $\Z/8$ lies in the Majorana layer. 
\end{proof}
See Wang--Gu~\cite[Table III]{PhysRevX.10.031055} and Zhang--Wang--Yang--Qi--Gu~\cite{ZWYQG20} for $\SH^\ell(B\Z/2, 0, x^2)$ for $\ell<4$.
\subsection{Example: $\SH^5(B\Z/2^k,0,y), k\geq 2$}
We compute $\SH^5(B\Z/2^k,0,y)$ for $k\geq 2$, which we use in \cref{ex:cyclicfermion}.
\begin{prop}\label{prop:2kytwist}\hfill
\begin{enumerate}
    \item For $k\geq 2$, the group $\SH^\ell(B\Z/2^k,0,y)$ is isomorphic to $\Z/2$ for $\ell = 4$ and to $\Z/2^{k+1}\oplus \Z/2$ for $\ell = 5$.
    \item The isomorphism $\SH^5(B\Z/2^k, 0, y)\cong\Z/2^{k+1}\oplus\Z/2$ may be chosen so that the class $\alpha_{\mathrm{GW}}$ mapping to $(1, 0)$ has image in the Gu--Wen layer of the $E_\infty$-page of the AHSS, and the class $\alpha_{\mathrm{Maj}}$ mapping to $(0, 1)$ has image in the Majorana layer. 
\end{enumerate}
\end{prop}
The case $\ell = 4$ verifies a prediction of Décoppet~\cite[Example 4.13]{D11}.
\begin{proof}
This will again require the HASS. The $E_2$-page of the AHSS is given by Equation \eqref{eq:SHZ2k}, with the twisted $d_2$ differentials given by
\begin{subequations}
\begin{align}
   d_2 &\colon E^{i,2}_2  \to E^{i+2,1}_2 && X \mapsto \Sq^2 X+y X \,, \\ 
     d_2 &\colon E^{i,1}_2 \to E^{i+2,0}_2 \quad&& X \mapsto (-1)^{\Sq^2 X + y X} \,.
\end{align}
\end{subequations}
After resolving the $d_2$ differentials, the $E_3$-page is given as follows:
\begin{equation}\label{eq:SHZ2kE3}
E^{i,j}_3=\begin{array}{c|cccccccc}
     j\\
     \\
     2 & 0 & 0 & y & x y &  \ldots \\
     1 & 1 & 0 & 0 & 0 & y^2 &\ldots  \\
      0 & \mathbb{C}^\times & \Z/2^k & 0 & \Z/2^{k-1} & 0 & \Z/2^k &0 & \ldots \\
     \hline 
     & 0 & 1 & 2 & 3 & 4 & 5 & 6 & i\,
\end{array}
\end{equation}
There is room for a nontrivial $d_3$ differential in total degree 5, and hence we turn to the hastened Adams spectral sequence.

The input to the (usual or hastened) Adams spectral sequence is the $\cA(1)$-module $H_{0,y}^*(B\Z/2^k; \Z/2)$ from~\cite[Definition 2.31(3)]{DY23}. Let $V\to B\Z/2^k$ be the vector bundle associated to the rotation representation of $\Z/2^k$ on $\R^2$. Then $(0, y) = (w_1(V), w_2(V))$, (i.e.\ this is a \term{vector bundle twist} of supercohomology, in the language of~\cite{DY23}), so there is an $\cA(1)$-module isomorphism (see~\cite{DY23})
\begin{equation}
    H_{0,y}^*(B\Z/2^k; \Z/2)\cong H^*((B\Z/2^k)^{V-2}; \Z/2),
\end{equation}
where for a space $X$ with bundle $V\to X$ with rank $r_V$, we denote by $X^{V-r_V}$ the associated \textit{Thom spectrum}, which is the suspension spectrum of
the Thom space.

The $\cA(1)$-module structure on $H^*((B\Z/2^k)^{V-2}; \Z/2)$ is computed in~\cite{Cam17, DL21, DDHM24}.\footnote{The references~\cite{Cam17, DL21} appear to use a different vector bundle than $V$, but this is a typo.} Given an $\cA(1)$-module $M$, let $\Sigma^k M$ denote the same $\cA(1)$-module with grading increased by $k$; we let $\Sigma M\coloneqq\Sigma^1 M$. For example, define $C\eta\coloneqq\Sigma^{-2}\widetilde H^*(\CP^2;\Z/2)$. Then there is an $\cA(1)$-module isomorphism
\begin{equation}
\label{spZ8_A1}
    H^*((B\Z/2^k)^{V-2}; \Z/2) \cong
        \textcolor{BrickRed}{C\eta} \oplus
        \textcolor{RedOrange}{\Sigma C\eta} \oplus
        \textcolor{Goldenrod!67!black}{\Sigma^4 C\eta} \oplus
        \textcolor{Green}{\Sigma^5 C\eta} \oplus F,
\end{equation}
where $F$ is concentrated in degrees $8$ and above (and thus we may ignore it). The $\Sigma^{2k}C\eta$ summand is spanned by $Uy^k$ and $Uy^{k+1}$, where $U$ is the Thom class, and the $\Sigma^{2k+1}C\eta$ summand is spanned by $Uxy^k$ and $Uxy^{k+1}$.

By \cref{cQthm}, $\mathcal Q$ commutes with direct sums and suspensions and vanishes in topological degrees below the minimum degree of a bounded-below $\cA(1)$-module, so we can ignore $F$ and only need $\mathcal Q(C\eta)$. We compute this in~\cite{DYY2} and give the result in \cref{first_HASS}, left (compare $\Ext_{\cA(1)}(C\eta)$, displayed in~\cite[Figure 22]{BC18}). Using this, we can draw the $E_2$-page of the HASS in \cref{first_HASS}, center. Differentials can be computed by comparing to the corresponding Adams spectral sequence for twisted spin bordism, as in~\cite[\S 7.9]{Cam17} or~\cite[\S 13.2]{DDHM24}: except for on the $E_k$-page, all differentials vanish. Thus we obtain the $E_{k+1} = E_\infty$-page in \cref{first_HASS}, right. As in the usual Adams spectral sequence, vertical lines represent $h_0$-multiplication, which lifts to multiplication by $2$, so we deduce that $(\tau_{\le 2}\ko)_5(B\Z/2^k, 0, y)\cong \textcolor{Green}{\Z/{2^{k+1}}}\oplus\textcolor{Orange}{\Z/2}$, and the corresponding twisted supercohomology is the Pontryagin dual group. Thus the Atiyah--Hirzebruch $d_3$ mentioned above vanishes. 

Comparing the $E_\infty$-page of the AHSS with the answer we found by the HASS, we see there is a hidden extension in total degree $5$ in the AHSS. It must be an extension of the Dijkgraaf--Witten layer by either the Gu--Wen layer or the Majorana layer.

\begin{figure}[H]
\begin{subfigure}[c]{0.33\textwidth}
\begin{sseqdata}[name=QB, classes = fill, scale=0.58, xrange={0}{6}, yrange={0}{4},
x tick step = 2,
y tick step = 2,
x tick gap = 0.25cm,
y tick gap = 0.25cm,
x axis tail = 0.2cm,
y axis tail = 0.2cm,
y axis clip padding = 0cm,
class labels = {right = 0.07em, font=\tiny },
]
\class(0, 0)
\class["h_0"](0, 1)\structline
\AdamsTower{}
\class["b"](2, 1)\AdamsTower{}
\class["b^2"](4, 2)
\end{sseqdata}
\printpage[name=QB, page=2]
\end{subfigure}
\begin{subfigure}[c]{0.32\textwidth}
\begin{sseqdata}[name=spZ, classes = fill, scale=0.58, xrange={0}{6}, yrange={0}{4},
x tick step = 2,
y tick step = 2,
x tick gap = 0.25cm,
y tick gap = 0.25cm,
x axis tail = 0.2cm,
y axis tail = 0.2cm,
y axis clip padding = 0cm,
class labels = {right = 0.07em, font=\tiny },
>=stealth, Adams grading
]
\begin{scope}[BrickRed]
    \class(0, 0)
    \class(0, 1)\structline
    \AdamsTower{}
    \class(2, 1)\AdamsTower{}
    \class(4, 2)
\end{scope}
\begin{scope}[RedOrange]
    \class(1, 0)
    \class(1, 1)\structline
    \AdamsTower{}
    \class(3, 1)\AdamsTower{}
    \class(5, 2)
\end{scope}
\begin{scope}[draw=none, fill=none]
    \foreach \x in {4, 5} {
        \foreach \y in {0, 1, 3, 4, 5} {
            \class(\x, \y)
        }
    }
    \foreach \x in {1, 3, 5} {
        \foreach \y in {6, 7} {
            \class(\x, \y)
        }
    }
\end{scope}
\begin{scope}[Goldenrod!67!black]
    \class(4, 0)
    \class(4, 1)\structline
    \AdamsTower{}
    \class(6, 1)\AdamsTower{}
\end{scope}
\begin{scope}[Green]
    \class(5, 0)
    \class(5, 1)\structline
    \AdamsTower{}
    \class(7, 1)\AdamsTower{}
\end{scope}
\begin{scope}[gray]
    \d3(2, 1)
    \d3(4, 0, -1)
    \d3(6, 1)(5, 4, -1)
\end{scope}
\begin{scope}[draw=none]
    \foreach \y in {2, 3, 4} {
        \d3(4, \y, -1)
        \d3(2, \y)
        \d3(6, \y)
    }
    \d3(4, 1, -1)
\end{scope}
\end{sseqdata}
\printpage[name=spZ, page=3]
\end{subfigure}
\begin{subfigure}[c]{0.33\textwidth}
\printpage[name=spZ, page=4]
\end{subfigure}
\caption{Left: $\mathcal Q(C\eta)$, computed in~\cite{DYY2}. Center: the $E_k$-page of the HASS computing $\tau_{\le 2}\ko(B\Z/2^k, 0, y)$ (here $k = 3$). Right: the $E_{k+1} = E_\infty$-page.}
\label{first_HASS}
\end{figure}

\begin{lem}
\label{bonus_AHSS_lemma}
The hidden extension in degree $5$ of the AHSS is between the Dijkgraaf--Witten and Gu--Wen layers; thus, the isomorphism $\phi\colon \SH^5(B\Z/2^k, 0, y)\overset\cong\to\Z/2^{k+1}\oplus\Z/2$ may be chosen so that $\alpha_{\mathrm{GW}}\coloneqq \phi^{-1}(1, 0)$ has image in the $E_\infty$-page of the AHSS in the Gu--Wen layer and $\alpha_{\mathrm{Maj}}\coloneqq \phi^{-1}(0, 1)$ has image in the Majorana layer.
\end{lem}
\begin{proof}
Let $\iota\colon\Z/2\to\Z/2^k$ be the map sending $1\mapsto 2^{k-1}$; we will also let $\iota$ denote the induced map on classifying spaces. Recall that $\iota^*(y) = x^2\in H^2(B\Z/2;\Z/2)$,\footnote{Because $y$ is $w_2$ of the standard rotation representation $\rho$ of $B\Z/2^k$, it suffices to show that restricting $\rho$ to $\Z/2$ yields the representation $2\sigma$; then $w_2(2\sigma) = x^2$ by the Whitney sum formula.} so we have a map $\SH^*(B\Z/2, 0, x^2)\to \SH^*(B\Z/2^k, 0, y)$, and therefore a map of AHSSes computing these supercohomology groups. This map is compatible with the extension problems on the $E_\infty$-pages in the following sense: each extension is a short exact sequence from a group on the $E_\infty$-page to the corresponding quotient of supercohomology, and the map $\iota$ induces a commutative diagram of short exact sequences. Thus, in particular, if $\mathit{rSH}^n(X, a, b)$ denotes the quotient of $(X, a, b)$-twisted supercohomology by the Majorana layer, so that $\mathit{rSH} \simeq I_{\C^\times}(\tau_{\le 1}\ko)$,\footnote{$\mathit{rSH}$ is Gu--Wen restricted supercohomology~\cite{Fre08, PhysRevB.90.115141}.}, then $\SH^n(X, a, b)$ is an extension of $\mathit{rSH}^n(X, a, b)$ by the Majorana layer $E_\infty^{2,n-2}$, and specializing to the map $\iota$ we get a commutative diagram of short exact sequences
\begin{equation}
\label{filt_comm}
\begin{tikzcd}
	0 & {^{k}\!E_\infty^{5,0}} & {\mathit{rSH}^5(B\Z/2^k, 0, y)} & {{}^{k}\!E_\infty^{4,1}} & 0 \\
	0 & {^{1}\!E_\infty^{5,0}} & {\mathit{rSH}^5(B\Z/2, 0, x^2)} & {{}^{1}\!E_\infty^{4,1}} & 0,
	\arrow[from=1-1, to=1-2]
	\arrow[from=1-2, to=1-3]
	\arrow["{\iota^*}"', from=1-2, to=2-2]
	\arrow[from=1-3, to=1-4]
	\arrow["{\iota^*}"', from=1-3, to=2-3]
	\arrow[from=1-4, to=1-5]
	\arrow["{\iota^*}"', from=1-4, to=2-4]
	\arrow[from=2-1, to=2-2]
	\arrow[from=2-2, to=2-3]
	\arrow[from=2-3, to=2-4]
	\arrow[from=2-4, to=2-5]
\end{tikzcd}
\end{equation}
where ${}^\ell\! E_r^{p,q}$ denotes the AHSS for the twisted supercohomology of $B\Z/2^\ell$. To prove the lemma, it would suffice to show that the upper central term of~\eqref{filt_comm}, $\mathit{rSH}^5(B\Z/2^k, 0, y)$, is isomorphic to $\Z/2^{k+1}$, as this plus the HASS computation would force the $\mathit{rSH}$-to-Majorana extension to split. Therefore our next task is to fill in the entries of~\eqref{filt_comm}. We computed ${}^1\! E_\infty^{5-j,j}$ in~\eqref{eq:SH} (there we claim it is the $E_3$-page, but in the proof of \cref{prop:Z2x2twist} we show that $d_3$ vanishes going to or from total degree $5$), and we computed ${}^{k}\!E_\infty^{5-j,j}$ in~\eqref{eq:SHZ2kE3} (again, this was the $E_3$-page, and we used the HASS to show this equals $E_\infty$ in degree $5$). Because $y$ pulls back to $x^2$, the map $\iota^*\colon {}^{k}\!E_\infty^{4,1}\to {}^{1}\!E_\infty^{4,1}$ is an isomorphism $\Z/2\overset\cong\to\Z/2$. The map on $E_\infty^{5,0}$ can be computed by finding its image under the Bockstein $H^5(\bl;\C^\times)\to H^6(\bl;\Z)$; there it is a map
\begin{equation}
    \iota^*\colon \Z/2^k\cong H^6(B\Z/2^k;\Z) \longrightarrow H^6(B\Z/2;\Z)\cong\Z/2.
\end{equation}
To show this map is nonzero (which uniquely determines it), use the universal coefficient theorem to show that it suffices to show that the image in mod $2$ cohomology is nonzero; there we already know the map sends $y^3\mapsto x^6$, hence is nonzero.

We have thus filled in most of~\eqref{filt_comm}; only the middle column remains. Because $\SH^5(B\Z/2, 0, x^2)\cong\Z/8$ (\cref{prop:Z2x2twist}) and $\mathit{rSH}^5(B\Z/2, 0, x^2)$ is a quotient of this $\Z/8$ by the $\Z/2$ in the Majorana layer, we have $\mathit{rSH}^5(B\Z/2, 0, x^2)\cong\Z/4$. Thus~\eqref{filt_comm} becomes the following commutative diagram of short exact sequences:
\begin{equation}
\begin{tikzcd}
	0 & {\Z/2^k} & {\mathit{rSH}^5(B\Z/2^k, 0, y)} & {\Z/2} & 0 \\
	0 & {\Z/2} & {\Z/4} & {\Z/2} & 0,
	\arrow[from=1-1, to=1-2]
	\arrow[from=1-2, to=1-3]
	\arrow["{1\mapsto 1}"', two heads, from=1-2, to=2-2]
	\arrow[from=1-3, to=1-4]
	\arrow["{\iota^*}"', from=1-3, to=2-3]
	\arrow[from=1-4, to=1-5]
	\arrow["\cong"', from=1-4, to=2-4]
	\arrow[from=2-1, to=2-2]
	\arrow[from=2-2, to=2-3]
	\arrow[from=2-3, to=2-4]
	\arrow[from=2-4, to=2-5]
\end{tikzcd}\end{equation}
and one can quickly check that this is only possible when $\mathit{rSH}^5(B\Z/2^k, 0, y)\cong\Z/2^{k+1}$. As noted above, this finishes the proof of the lemma.
\end{proof}

Looking at the $E_\infty$-page of the HASS  (\cref{first_HASS}, right), we also see that $\SH^4(B\Z/2^k, 0, y)\cong\Z/2$.

\end{proof}
See Wang--Gu~\cite[Table III]{PhysRevX.10.031055} and Zhang--Wang--Yang--Qi--Gu~\cite{ZWYQG20} for $\SH^\ell(B\Z/2^k, 0, y)$ for $\ell<4$.

\section{Example: \texorpdfstring{$\SH^5(B\Z/2 \times B\Z/2^k,x_1,y), k\geq 2$}{SH5(BZ/2*B/Z/2k,x1,y), k>=2}}\label{appendix:timerev}
In this appendix, we perform the computations supporting \cref{ex:timerev}, with the fermionic group $(B\Z/2\times B\Z/2^k, x_1, y)$ -- an example with time-reversal symmetry. In \S\ref{subsubsec:TR_adams}, we compute $\SH^5(B\Z/2 \times B\Z/2^k,x_1,y),k\geq 2$; then in \S\ref{subsubsec:cons}, we discuss how to perform symmetry extension. This example was not addressed by Wan--Wang~\cite{WW25}, and so we do not know whether our choice of cover $G\to \Z/2\times \Z/2^k$ has minimal degree; we would be interested in learning whether this is the case.
\subsubsection{The twisted supercohomology computation}
\label{subsubsec:TR_adams}
The $\Z/2$ cohomology ring of $B\Z/2 \times B\Z/2^k$ is 
\begin{equation}
    H^*(B\Z/2 \times B\Z/2^k; \Z/2) \cong \Z/2[x_1, x, y]/(x^2),\quad |x_1| = |x| =1,~ |y|=2.
\end{equation}
There is an isomorphism $\SH^5(B\Z/2 \times B\Z/2^k,x_1,y)\cong \SH^5(B\Z/2\times B\Z/2^k, x_1, y+x_1^2)$, because the twists $(x_1, y)$ and $(x_1, y+x_1^2)$ are related by an automorphism of $B\Z/2\times B\Z/2^k$. Namely, the automorphism is given by
\begin{equation}
   f\colon \Z/2 \times \Z/2^k \rightarrow \Z/2 \times \Z/2^k, \quad \quad \quad (1, 0) \mapsto (1, 2^{k - 1}),~(0, 1) \mapsto (0, 1), 
\end{equation}
under which we have $f^*(x_1) = x_1 $, $f^*(x) = x$ and $f^*(y) = y + x_1^2$.

\begin{prop}\label{prop:22kxy}
There is an isomorphism $\SH^5(B\Z/2\times B\Z/2^k, x_1, y)\cong\Z/4\oplus\Z/2\oplus\Z/2$ such that the generators $\alpha_{\mathrm{GW}}$, $\alpha_{\mathrm{DW}}$, and $\alpha_{\mathrm{Maj}}$, corresponding to $(1, 0, 0)$, $(0, 1, 0)$, and $(0, 0, 1)$ respectively, have the following properties.
\begin{enumerate}
    \item The images of $\alpha_{\mathrm{DW}}$,  $\alpha_{\mathrm{GW}}$, and $\alpha_{\mathrm{Maj}}$  in the $E_\infty$-page of the AHSS are in the Dijkgraaf--Witten, Gu--Wen, and Majorana layers, respectively.
    \item The kernel of the map $\SH^5\to\mho_\Spin^5$ is spanned by $\alpha_{\mathrm{Maj}}$.
\end{enumerate}
\end{prop}

\begin{lem}
\label{AHSS_analysis}
The following hold for the $E_3$-page of the Atiyah--Hirzebruch spectral sequence computing $\SH^*(B\Z/2\times B\Z/2^k, x_1, y)$.
\begin{enumerate}
    \item\label{4_16} There are exactly $16$ classes in total degree $4$.
    \item\label{E3_5} A basis for total degree $5$ consists of $(-1)^{x_1^4x}$ (DW layer), $(-1)^{xy^2}$ (DW layer), $x_1^3 x$ (GW layer), $x_1^3 + x_1y$ (Majorana layer), and $xy$ (Majorana layer).
    \item\label{mspin_d3} In the corresponding spectral sequence for $\mho_\Spin^*$-cohomology, $x_1^3 + x_1y\in E_2^{3,2}$ is in the image of $d_2$.
\end{enumerate}
\end{lem}
\begin{proof}
As usual, the twisted $d_2$ differentials are given by
\begin{align}
   d_2 &\colon E^{i,2}_2  \to E^{i+2,1}_2 && X \mapsto \Sq^2_{x_1,y}(X) \coloneqq \Sq^2 X+x_1\Sq^1 X + y X \,, \\ 
     d_2 &\colon E^{i,1}_2 \to E^{i+2,0}_2 \quad&& X \mapsto (-1)^{\Sq^2 X +x_1\Sq^1 X  + y X} \,.
\end{align}
We assemble these ingredients and give the $E_2$-page as follows:
{\scriptsize \begin{equation*}
E^{i,j}_2=\begin{array}{c|cccccccc}
     j\\
     \\
    2 & 1 & x_1, x & x_1^2, x_1 x, y & x_1^3, x_1^2 x, x_1 y, x y &   \ldots \\
    1 & 1 & x_1, x & x_1^2, x_1 x, y & x_1^3, x_1^2 x, x_1 y, x y & x_1^4, x_1^3 x, x_1^2 y, x_1 x y, y^2 &  \ldots \\
      0 & -1 & (-1)^{x} & (-1)^{x_1^2}, (-1)^{y} & (-1)^{x_1^2 x}, (-1)^{x y} & (-1)^{x_1^4}, (-1)^{x_1^2 y}, (-1)^{y^2} & (-1)^{x_1^4 x}, (-1)^{x_1^2 x y}, (-1)^{x y^2} & \ldots \\
     \hline 
     & 0 & 1 & 2 & 3 & 4 & 5  & i\,
\end{array}
\end{equation*}}
After resolving the $d_2$ differentials, the $E_3$-page is given by:
{\footnotesize \begin{equation*}\label{eq:E3timerev}
E^{i,j}_3=\begin{array}{c|ccccccccc}
     j\\
     \\
    2 & 0 & 0 & y & x_1^3 + x_1 y,  x y &  \ldots \\
    1 & 0 & x_1 & x_1 x &   x_1^3 &  x_1^3 x &  \ldots \\
      0 & -1 & (-1)^{x} & (-1)^{x_1^2} & (-1)^{x_1^2 x} &  (-1)^{x_1^4}, (-1)^{y^2} & (-1)^{x_1^4 x}, (-1)^{x y^2} & (-1)^{x_1^6}, (-1)^{x_1^2y^2}\ldots \\
     \hline 
     & 0 & 1 & 2 & 3 & 4 & 5  & 6 & i\,
\end{array}
\end{equation*}}
This proves items~\eqref{4_16} and~\eqref{E3_5} of the lemma statement. There could be nontrivial $d_3$ differentials $d_3\colon E_3^{2, 2} \rightarrow E_3^{5, 0}$ and  $d_3\colon E_3^{3, 2} \rightarrow E_3^{6, 0}$, as well as potentially a hidden extension between different layers in degree $5$; we will in a moment turn to the HASS to solve these problems.

Lastly we prove part~\eqref{mspin_d3}. In this spectral sequence, $d_2\colon E_2^{i,3}\to E_2^{i+2, 2}$ is identified with the map $H^i(\bl;\Z)\to H^{i+2}(\bl;\Z/2)$ which is reduction modulo $2$ followed by $\Sq^2$~\cite{Bot69}. (See also Footnote~\ref{AHSS_diffs}.)

Let $\widetilde e\in H^1(B\Z/2\times B\Z/2^k; \Z_{x_1})$ be the twisted Euler class of $\sigma_1$, the tautological line bundle over $B\Z/2$, pulled back to the product $B\Z/2\times B\Z/2^k$ (see~\cite[Lemma 1]{Cad99}). Then $\widetilde e\bmod 2 = w_1(\sigma_1) = x_1$ (\textit{ibid.}), so 
\begin{equation}
    d_2(\widetilde e) = \Sq^2_{x_1,y}(x_1) = x_1^3 + x_1y,
\end{equation}
which proves part~\eqref{mspin_d3}.
\end{proof}
\begin{lem}
\label{long_HASS_analysis}
The following facts hold for the $E_\infty$-page of the HASS computing $(B\Z/2\times B\Z/2^k, x_1, y)$-twisted $\tau_{\le 2}\ko$-homology.
\begin{enumerate}
    \item\label{Adams_4line} There are exactly $16$ classes in topological degree $4$.
    \item\label{Adams_5line} There are classes $b$, $c$, and $e$ in topological degree $5$ such that $\set{b,c,h_0c, e}$ is a basis for topological degree $5$.
    \item\label{cokernel} The cokernel of the map of $E_\infty$-pages from the $\ko$-Adams SS to the $\tau_{\le 2}\ko$-HASS in topological degree $5$ is $\Z/2$, spanned by $e$.
    \item\label{no_ext} There are no hidden extensions in topological degree $5$, so $\tau_{\le 2}\ko_5(B\Z/2\times B\Z/2^k, x_1, y)\cong\Z/4\oplus (\Z/2)^{\oplus 2}$.
\end{enumerate}
\end{lem}
We will postpone the proof of \cref{long_HASS_analysis} in order to first see how it helps us.
\begin{proof}[Proof of \cref{prop:22kxy}, assuming \cref{long_HASS_analysis}]
Comparing \cref{AHSS_analysis,long_HASS_analysis} in total degree $5$, there appears to be a discrepancy: there are $32$ classes in $E_3$ of the AHSS and $16$ in $E_\infty$ of the HASS. (These two spectral sequences compute $\SH$-cohomology, resp.\ $\tau_{\ge 2}\ko$-homology, which are Pontryagin dual and therefore abstractly isomorphic whenever they are finite.) This means that there must be a $d_r$, $r\ge 3$, in the AHSS that kills some class in total degree $5$. Since the numbers of elements in total degree $4$ match between these two spectral sequences, this differential must go from total degree $5$ to total degree $6$. The only option for this differential is $d_3\colon E_3^{3,2}\to E_3^{6,0}$. Moreover, this differential will be preserved by the map into the $\mho_\Spin^*$-AHSS, where it must vanish on $x_1^3 + x_1y$, so that $d_3^2 = 0$; thus $d_3(xy)\ne 0$. For degree reasons there can be no more nonzero differentials in total degree $5$ for the AHSS, so we know that the $E_\infty$-page is spanned by $(-1)^{x_1^4x}$, $(-1)^{xy^2}$, $x_1^3x$, and $x_1^3 + x_1y$, in the DW, DW, GW, and Majorana layers respectively.

To finish, we resolve the extensions on the $E_\infty$-page of the AHSS. The HASS calculations in \cref{long_HASS_analysis} imply we must answer the following two questions,
\begin{enumerate}
    \item\label{AHSS_extn} $\SH^5(B\Z/2\times B\Z/2^k, x_1, y)\cong\Z/4\oplus (\Z/2)^{\oplus 2}$, but the $E_\infty$-page of the AHSS has four $\Z/2$ summands in total degree $5$. Where is the hidden extension?
    \item\label{which_killed} What is the filtration in the AHSS of the class that is killed when one maps to $\mho_\Spin^*$?
\end{enumerate}
Question~\eqref{which_killed} is easier: by the previous paragraph, $x_1^3 + x_1y\in E_\infty^{3,2}$ is killed by the map of AHSSes to $\mho_\Spin^*$. This class lifts to a class $\alpha_\mathrm{Maj}$ which is killed when one passes to $\mho_\Spin^*$.

Now~\eqref{AHSS_extn}. The HASS analysis implies that the hidden extension is between two classes that are not in the kernel of the map to $\mho_\Spin^*$, and these two classes are necessarily in two different layers of the AHSS filtration. This uniquely forces it to be an extension of a $\Z/2$ subgroup of the Dijkgraaf--Witten layer by the unique $\Z/2$ in the Gu--Wen layer (spanned by $x_1^3 x$), giving a generator $\alpha_{\mathrm{GW}}$ in the Gu--Wen layer generating a $\Z/4$. A complementary subgroup to the image of $2\alpha_{\mathrm{GW}}$ in $E_\infty^{5,0}$ lifts to the generator $\alpha_{\mathrm{DW}}$.
\end{proof}

\begin{proof}[Proof of \cref{long_HASS_analysis}]

As in the previous example (and all examples in this paper), the twist $(x_1, y)$ is a vector bundle twist: it is $(w_1(W), w_2(W))$ for the vector bundle $W\coloneqq \sigma_1\boxplus V$, where $\sigma_1$ is the tautological bundle over $B\Z/2$ and $V$ is the bundle associated to the standard representation of $\Z/2^k$ on $\C$. The notation $\boxplus$ denotes external direct sum, i.e.\ pull these bundles back to the product $B\Z/2\times B\Z/2^k$, then direct sum them. Thus $H_{x_1, y}^*(B\Z/2\times B\Z/2^k; \Z/2)\cong H^*((B\Z/2\times B\Z/2^k)^{W - 3};\Z/2)$, like in the previous example.

The Thom spectrum associated to an external direct sum splits as a smash product, so the Künneth formula calculates its cohomology:
\begin{equation}
\label{align_mess}
    \begin{aligned}
        H_{x_1, y}^*(B\Z/2\times B\Z/2^k; \Z/2)
            &\cong H^*((B\Z/2\times B\Z/2^k)^{W - 3};\Z/2)\\
            &\cong H^*((B\Z/2)^{\sigma-1}\wedge (B\Z/2^k)^{V-2};\Z/2)\\
            &\cong H^*((B\Z/2)^{\sigma-1};\Z/2)\otimes H^*((B\Z/2^k)^{V-2};\Z/2)\\
            &\underset{\eqref{spZ8_A1}}{\cong} P\otimes \left(\textcolor{BrickRed}{C\eta} \oplus
        \textcolor{RedOrange}{\Sigma C\eta} \oplus
        \textcolor{Goldenrod!67!black}{\Sigma^4 C\eta} \oplus
        \textcolor{Green}{\Sigma^5 C\eta} \oplus F\right).
    \end{aligned}
\end{equation}
Here $P\coloneqq H^*((B\Z/2)^{\sigma-1};\Z/2)$. Letting $R_6\coloneqq P\otimes C\eta$,
\begin{equation}
\label{nolonger_aligned}
        H_{x_1, y}^*(B\Z/2\times B\Z/2^k; \Z/2)
        \cong \textcolor{BrickRed}{R_6} \oplus
        \textcolor{RedOrange}{\Sigma R_6} \oplus
        \textcolor{Goldenrod!67!black}{\Sigma^4 R_6} \oplus
        \textcolor{Green}{\Sigma^5 R_6} \oplus F'
\end{equation}
for some $\cA(1)$-module $F'$ concentrated in degrees $8$ and above. We compute $\mathcal Q(R_6)$ in~\cite{DYY2} (compare $\Ext_{\cA(1)}(R_6)$ in \cite[Figure 41]{BC18}) and draw the result in \cref{R6_example}, left. Using this, we draw the $E_2$-page of the HASS for $(B\Z/2\times B\Z/2^k, x_1, y)$-twisted $\tau_{\le 2}\ko$-homology in \cref{R6_example}, center. The differentials $d_2\colon E_2^{0,5}\to E_2^{2,6}$ and $d_2\colon E_2^{0,6}\to E_2^{2,7}$ could be nonzero; all other differentials in range vanish because their source or target is the zero group. To describe the differentials more carefully, we name the following classes.
\begin{enumerate}
    \item $\mathcal Q^{s,t}(\textcolor{BrickRed}{R_6})\cong\Z/2$ for each of $(s,t) = (0, 4)$, $(2, 7)$, and $(0, 6)$; let $a$, $e$, and $f$ be the nonzero elements of each of these groups, respectively. Thus, through the split inclusion $\textcolor{BrickRed}{R_6}\hookrightarrow H_{x_1, y}^*(B\Z/2\times B\Z/2^k; \Z/2)$ in~\eqref{nolonger_aligned}, we obtain classes $a$, $e$, and $f$ in $E_2^{0,4}$, $E_2^{2,7}$, and $E_2^{0,6}$, respectively.
    \item Repeat this procedure to define $c\in\mathcal Q^{0,5}(\textcolor{RedOrange}{\Sigma R_2})\hookrightarrow E_2^{0,5}$, $g\in\mathcal Q^{0,6}(\textcolor{Goldenrod!67!black}{\Sigma^4 R_2})\hookrightarrow E_2^{0,6}$, and $b\in\mathcal Q^{0,5}(\textcolor{Green}{\Sigma^5 R_2})\hookrightarrow E_2^{0,5}$ as the unique nonzero elements in their respective $\mathcal Q^{s,t}$ groups, then included into the $E_2$-page.
\end{enumerate}
These classes are labeled in \cref{R6_example}, center.

\begin{figure}[H]
\begin{subfigure}[c]{0.32\textwidth}
\begin{sseqdata}[name=QR6, classes = fill, scale=0.58, xrange={0}{6}, yrange={0}{4},
x tick step = 2,
y tick step = 2,
x tick gap = 0.25cm,
y tick gap = 0.25cm,
x axis tail = 0.2cm,
y axis tail = 0.2cm,
y axis clip padding = 0cm,
x axis clip padding = 0cm,
class labels = {below = 0.07em, font=\tiny },
]
\class["q_0"](0, 0)
\class["q_1"](2, 0)
\class(2, 1)\structline
\class["q_2"](4, 0)
\class(4, 1)\structline
\class(4, 2)\structline
\class["q_0'"](5, 2)
\class["q_3"](6, 0)
\class(6, 1)\structline
\class(6, 2)\structline
\class["q_1'"](7, 0)
\class["q_4"](8, 0)
\class(8, 1)\structline
\class(8, 2)\structline
\end{sseqdata}
\printpage[name=QR6, page=2]
\end{subfigure}
\begin{subfigure}[c]{0.33\textwidth}
\begin{sseqdata}[name=QR6use, classes = fill, scale=0.6, xrange={0}{6}, yrange={0}{4}, >=stealth,
x tick step = 2,
y tick step = 2,
x tick gap = 0.25cm,
y tick gap = 0.25cm,
x axis tail = 0.2cm,
y axis tail = 0.2cm,
y axis clip padding = 0cm,
x axis clip padding = 0cm,
class labels = {below = 0.07em, font=\small },
]
\begin{scope}[BrickRed]
    \class(0, 0)
    \class(2, 0)
    \class(2, 1)\structline
    \class["a"](4, 0)
    \class(4, 1)\structline
    \class(4, 2)\structline
    \class["e"](5, 2)
    \class["f"](6, 0)
    \class(6, 1)\structline
    \class(6, 2)\structline
\end{scope}
\begin{scope}[draw=none, fill=none]
    \class(4, 1)\class(4, 2)
    \class(5, 1)
\end{scope}
\class[class labels = {above = 0.07em}, "b", Green](5, 0)
\begin{scope}[RedOrange]
    \class(1, 0)
    \class(3, 0)
    \class(3, 1)\structline
    \class["c"](5, 0)
    \class(5, 1)\structline
    \class(5, 2)\structline
    \class(6, 2)
\end{scope}
\begin{scope}[Goldenrod!67!black]
    \class(4, 0)
    \class[class labels = {right = 0.07em}, "g"](6, 0)
    \class(6, 1)\structline
\end{scope}
\begin{scope}[gray]
    \d2(6, 0, -1)(5, 2, -1)
\end{scope}
\end{sseqdata}
\printpage[name=QR6use, page=2]
\end{subfigure}
\begin{subfigure}[c]{0.33\textwidth}
    \printpage[name=QR6use, page=3]
\end{subfigure}
\caption{Left: $\mathcal Q(R_6)$, computed in~\cite{DYY2}. Center: the $E_2$-page of the HASS computing $\tau_{\le 2}\ko_*(B\Z/2\times B\Z/2^k, x_1, y)$. We calculate the $d_2$s in range in \cref{d265}. Right: the $E_\infty$-page.}
\label{R6_example}
\end{figure}

Thus $d_2\colon E_2^{0,5}\to E_2^{2,6}$ sends $b$, $c$, or both to $0$ or $h_0^2a$, and $d_2\colon E_2^{0,6}\to E_2^{5,7}$ sends $f$ and $g$ to elements of $\{0,e,h_0^2c, e+h_0^2c\}$.

The map $\tau_{\le 2}\colon \ko\to\tau_{\le 2}\ko$ induces a map of Adams spectral sequences; $a$, $b$, $c$, $f$, and $g$ are in the image of this map, so their differentials are as well, but $e$ is \emph{not} in the image of this map, as can be seen by comparing $\Ext_{\cA(1)}(R_6)$ (see~\cite[Figure 41]{BC18}) and $\mathcal Q(R_6)$. This proves part~\eqref{cokernel} of the lemma statement. Thus $d_2(f)$ and $d_2(g)$ are either $0$ or $h_0^2c$.

To finish the proofs of parts~\eqref{Adams_4line}, \eqref{Adams_5line}, and~\eqref{no_ext} of the lemma statement, we prove the following lemma.
\begin{lem}
\label{d265}
$d_2(b) = d_2(c) = 0$, $d_2(g) = h_0^2c$, and $d_2(f) = \lambda h_0^2c$ for some $\lambda\in\Z/2$. Equivalently, $(\tau_{\le 2}\ko)_n(B\Z/2\times B\Z/2^k, x_1, y)$ is isomorphic to $\Z/2\oplus\Z/8$ for $n = 4$ and $\Z/4\oplus (\Z/2)^{\oplus 2}$ for $n = 5$.
\end{lem}
\begin{proof}
Rather than directly compute these differentials, we will use a different technique, the \term{Smith long exact sequence}, to compute these twisted $\tau_{\le 2}\ko$-homology groups. See~\cite{HS13, Guo:2018vij, DL23, DDHM24, debray2024smith, Deb24, DNT24, debray2025smith, Debray:2023iwf, Debray:2025iqs, JTVP25} for more examples of this technique.

\begin{thm}[James]
\label{Smith_LES}
Let $V, W\to X$ be vector bundles of ranks $r_V$, $r_W$, respectively, and $p\colon S(W)\to X$ be the sphere bundle of $W$. For any generalized homology theory $E_*$, there is a long exact sequence
\begin{equation}\label{eq:LESsmith}
	\dotsb \rightarrow E_k(S(W)^{p^*V-r_V}) \overset{p_*}{\rightarrow} E_k(X^{V-r_V})
	\overset{\mathrm{sm}_W}{\rightarrow} E_{k-r_W}(X^{V\oplus W -(r_V+r_W)})\rightarrow
	E_{k-1}(S(W)^{p^*V-r_V})\rightarrow \dotsb
\end{equation}
\end{thm}
\begin{thm}[\cite{debray2024smith}]
\label{LES_is_Smith}
With notation as in \cref{Smith_LES}, suppose $E = \Omega^\xi$ is a bordism homology theory for a tangential structure $\xi$. Then, under the identification of $\Omega_k^\xi(X^{V-r_V})$ as the abelian group of bordism classes of $(X, V)$-twisted $n$-dimensional $\xi$-manifolds,\footnote{Given a vector bundle $V\to X$, an \term{$(X, V)$-twisted $\xi$-structure}~\cite[\S 4]{Hason:2020yqf} on a vector bundle $E\to M$ is the data of a map $f\colon M\to X$ and a $\xi$-structure on $E\oplus f^*(V)$. The bordism groups of manifolds whose tangent bundles have $(X, V)$-twisted $\xi$-structures are naturally isomorphic to the $\xi$-bordism groups of the Thom spectrum $X^{V - \mathrm{rank}(V)}$~\cite[Corollary 10.19]{DDHM24}.\label{shfoot}} $\mathrm{sm}_W$ is the \term{Smith homomorphism}, which sends the bordism class of an $(X, V)$-twisted $\xi$-manifold $(M, f\colon M\to X)$ to the bordism class of the Poincaré dual of the Euler class of $f^*(W)$.\footnote{It is true, yet nontrivial, that the Poincaré dual carries a canonical $(X, V\oplus W)$-twisted $\xi$-structure and that its bordism class does not depend on the choice of $M$.}\textsuperscript{,}\footnote{Depending on $\xi$, one may have to use a generalized cohomology Euler class in \cref{LES_is_Smith}; see~\cite[Appendix B]{debray2024smith}. This detail will not play a role in this paper.}
\end{thm}
See also \cite{COSY20, Hason:2020yqf, DNT24, Debray:2023ior,Debray:2025iqs, JTVP25} for applications and interpretations of the Smith long exact sequence in quantum physics. 

To apply \cref{Smith_LES}, let $E_* = \tau_{\le 2}\ko_*((B\Z/2)^{\sigma-1}\wedge\text{---})$, $X = B\Z/2^k$, and both $V$  and $W$ be the complex line bundle associated to the rotation representation of $\Z/2^k$. By~\cite[Example 7.28]{debray2024smith}, the map $S(V)\to B\Z/2^k$ can be identified up to homotopy with the modulo $2^k$ reduction map $S^1\simeq B\Z\to B\Z/2^k$. For any generalized homology theory $E$, $E_n(S^1)\cong E_n\oplus E_{n-1}$, as can be shown by using the Atiyah--Hirzebruch spectral sequence for the reduced $E$-homology of $S^1$. Letting $M\coloneqq (B\Z/2)^{\sigma-1}$ for brevity, we have the following long exact sequence:
\begin{equation}\label{eq:firstsmith}
\ldots\rightarrow 
(\tau_{\le 2}\ko)_n(M) \oplus
(\tau_{\le 2}\ko)_{n-1}(M) \to
(\tau_{\le 2}\ko)_n(M\wedge (B\Z/2^k)^{V-2}) \to
(\tau_{\le 2}\ko)_{n-2}(M\wedge (B\Z/2^k)_+) \overset\partial \to \dotsb
\end{equation}
\begin{lem}
\label{pinm_SH}
$(\tau_{\le 2}\ko)_n((B\Z/2)^{\sigma-1})$ is isomorphic to $\Z/2$ for $n = 0,1,5$, $\Z/8$ for $n = 2,6$, and $0$ for $n = 3,4,7$.
\end{lem}
Wang--Gu~\cite[Table III]{PhysRevX.10.031055} study the corresponding supercohomology groups in degrees $4$ and below.
\begin{proof}[Proof sketch]
This can be computed using the HASS in the same way as we computed $(\tau_{\le 2}\ko)_*(B\Z/2, 0, x^2)$ in \S\ref{subsub:twistedappendix}. See \cref{pinm_HASS}, left, for a picture of the $\cA(1)$-module structure on $P\coloneqq H^*((B\Z/2)^{\sigma-1};\Z/2)$ and \cref{pinm_HASS}, right, for $E_2 = \mathcal Q(P)$, which is calculated in~\cite{DYY2}. The spectral sequence collapses, mostly for degree reasons. The only remaining differential is the $d_2$ from degree $6$ to degree $5$. This differential is in the image of the map of Adams spectral sequences induced by $\ko\to\tau_{\le 2}\ko$: in the Adams spectral sequence for $\ko_*((B\Z/2)^{\sigma-1})$, whose $E_2$-page is calculated in~\cite[\S 2]{GMM68}, this differential does vanish, so we are done. We draw the $E_2 = E_\infty$-page of the Adams spectral sequence for $\ko_*((B\Z/2)^{\sigma-1})$ in \cref{pinm_HASS}, center.
\end{proof}

\begin{figure}[H]
\begin{subfigure}[c]{0.15\textwidth}
\begin{tikzpicture}[scale=0.4]
        \foreach \y in {0, ..., 8} {
                \tikzpt{0}{\y}{}{};
        }
        \foreach \y in {0, 2, 4, 6} {
                \sqone(0, \y);
        }
        \sqtwoL(0, 1);
        \sqtwoR(0, 2);
        \sqtwoL(0, 5);
        \sqtwoR(0, 6);
        \begin{scope}
                \clip (-1, 6.8) rectangle (1, 8.5);
                \sqone(0, 8);
        \end{scope}
\end{tikzpicture}\end{subfigure}
\begin{subfigure}[c]{0.35\textwidth}
\begin{sseqdata}[name=pinm, classes = fill, scale=0.4, xrange={0}{8}, yrange={0}{4},
x tick step = 2,
y tick step = 2,
x tick gap = 0.25cm,
y tick gap = 0.25cm,
x axis tail = 0.2cm,
y axis tail = 0.2cm,
y axis clip padding = 0cm,
x axis clip padding = 0cm,
class labels = {below = 0.07em, font=\tiny },
]
\class["\kappa"](0, 0)
\class(1, 1)\structline
\class(2, 2)\structline
\class(2, 1)\structline
\class["k_1"](2, 0)\structline

\class["k_2"](6, 0)
\class(6, 1)\structline
\class(6, 2)\structline
\class(6, 3)\structline

\class["w\kappa"](8, 4)
\class(9, 5)\structline
\end{sseqdata}
\printpage[name=pinm, page=2]
\end{subfigure}
\begin{subfigure}[c]{0.35\textwidth}
\begin{sseqdata}[name=QP, classes = fill, scale=0.4, xrange={0}{8}, yrange={0}{4},
x tick step = 2,
y tick step = 2,
x tick gap = 0.25cm,
y tick gap = 0.25cm,
x axis tail = 0.2cm,
y axis tail = 0.2cm,
y axis clip padding = 0cm,
x axis clip padding = 0cm,
class labels = {below = 0.07em, font=\tiny },
]
\class["\kappa"](0, 0)
\class(1, 1)\structline
\class(2, 2)\structline
\class(2, 1)\structline
\class["k_1"](2, 0)\structline

\class["k_1'"](5, 2)
\class["k_2"](6, 0)
\class(6, 1)\structline
\class(6, 2)\structline
\end{sseqdata}
\printpage[name=QP, page=2]
\end{subfigure}
\caption{Left: the $\cA(1)$-module structure on $P\coloneqq H^*((B\Z/2)^{\sigma-1};\Z/2)$.
Center: $\Ext_{\cA(1)}(P)$, the $E_2$-page of the Adams spectral sequence computing $\ko_*((B\Z/2)^{\sigma-1})$.
Right: $\mathcal Q(P)$, the $E_2$-page of the HASS computing $(\tau_{\le 2}\ko)_*((B\Z/2)^{\sigma-1})$. The classes $\kappa$, $k_1$, and $k_2$ are in the image of the map of $E_2$-pages induced by the truncation $\ko\to\tau_{\le 2}\ko$. We use this in the proof of \cref{pinm_SH}.}
\label{pinm_HASS}
\end{figure}

\begin{lem}
\label{cMSH}
$\mathcal M_n\coloneqq (\tau_{\le 2}\ko)_n((B\Z/2)^{\sigma-1}\wedge (B\Z/2^k)^{V-2})$ is isomorphic to $\Z/2$ for $n = 0,1$ and $\Z/4$ for $n = 2,3$. In higher degrees:
\begin{itemize}
    \item $\mathcal M_4$ is isomorphic to either $\Z/8\oplus\Z/2$, if $d_2(b) = d_2(c) = 0$, or to $\Z/4\oplus\Z/2$, if at least one of $d_2(b)$ or $d_2(c)$ is nonzero.
    \item If $d_2(b) = d_2(c) = 0$, then $\mathcal M_5$ is isomorphic to either $\Z/8\oplus (\Z/2)^{\oplus 2}$, if $d_2(f) = d_2(g) = 0$, or to $\Z/4\oplus (\Z/2)^{\oplus 2}$.
\end{itemize}
\end{lem}
\begin{proof}
These follow from the computation of the $E_2$-page of the HASS for these groups in \cref{R6_example}, as well as the observation we made that $e\not\in\mathrm{Im}(d_2)$. In principle, there could be a hidden extension from $b$ or $h_0c$ to $e$ in degree $5$, but because $b$ and $c$ are in the image of the map of spectral sequences induced by $\tau_{\le 2}$, and $e$ is not, this cannot occur.
\end{proof}
\begin{lem}
\label{cNSH}
Let $\mathcal N_n\coloneqq (\tau_{\le 2}\ko)_n((B\Z/2)^{\sigma-1}\wedge (B\Z/2^k)_+)$. Then $\mathcal N_0\cong\Z/2$ and $\mathcal N_1\cong (\Z/2)^{\oplus 2}$. In higher degrees:
\begin{itemize}
    \item $\mathcal N_2$ is isomorphic to either $\Z/8\oplus (\Z/2)^{\oplus 2}$ or $\Z/8\oplus\Z/4$.
    \item $\mathcal N_3\cong \Z/4\oplus\Z/2$.
\end{itemize}
\end{lem}
Our proof is an adaptation of the ideas in~\cite[\S 7.2.2]{Guo:2018vij}, which are used there to compute $\Omega_*^{\Pin^-}(B\Z/4)$ in low degrees. We replace $B\Z/4$ with $B\Z/2^k$ and truncate spin bordism to $\tau_{\le 2}\ko$, but the outline of the proof is not very different.
\begin{proof}
For any spaces $X$ and $Y$ and generalized cohomology theory $E$, there is a natural isomorphism
\begin{equation}
    E_*(X\wedge Y_+)\overset\cong\longrightarrow
        \widetilde E_*(X)\oplus \widetilde E_*(X\wedge Y),
\end{equation}
which is exactly the splitting of the $E_*(X\wedge\text{--})$-homology of $Y$ into the $E_*(X\wedge\text{--})$-homology of a point and the reduced $E_*(X\wedge\text{--})$-homology of $Y$. Therefore $\mathcal N_n$ is the direct sum of $(\tau_{\le 2}\ko)_n((B\Z/2)^{\sigma-1})$, which we computed in \cref{cMSH}, and $\widetilde{\mathcal N}_n\coloneqq (\tau_{\le 2}\ko)_n((B\Z/2)^{\sigma-1}\wedge (B\Z/2^k))$. We will focus on the latter, then implicitly direct-sum on $(\tau_{\le 2}\ko)_n((B\Z/2)^{\sigma-1})$ to obtain the groups in the lemma statement.

We attack $\widetilde{\mathcal N}_n$ with the hastened Adams spectral sequence. The $E_2$-page is $\mathcal Q$ applied to the $\cA(1)$-module
\begin{equation}
    \widetilde H^*((B\Z/2)^{\sigma-1}\wedge B\Z/2^k; \Z/2) \cong
    \widetilde H^*((B\Z/2)^{\sigma-1};\Z/2) \otimes_{\Z/2}
    \widetilde H^*(B\Z/2^k;\Z/2).
\end{equation}
There is an isomorphism
\begin{equation}
    \widetilde H^*(B\Z/2^k;\Z/2)\cong \textcolor{Green}{\Sigma \Z/2} \oplus \textcolor{MidnightBlue}{\Sigma^2 C\eta} \oplus \textcolor{Fuchsia}{\Sigma^3 C\eta}\oplus \overline F,
\end{equation}
where $\overline F$ is concentrated in degrees $6$ and above (see, e.g., \cite[Figure 7.5]{Cam17} or \cite[Proposition 13.20]{BDDHM25}). Recalling from around~\eqref{align_mess} that $R_6\coloneqq P\otimes C\eta$, we get
\begin{equation}
\widetilde H^*((B\Z/2)^{\sigma-1}\wedge B\Z/2^k; \Z/2) \cong
    \textcolor{Green}{\Sigma P} \oplus
    \textcolor{MidnightBlue}{\Sigma^2 R_6} \oplus
    \textcolor{Fuchsia}{\Sigma^3 R_6} \oplus F,
\end{equation}
where $F$ is concentrated in degrees $6$ and above (and so we can ignore it). We obtained $\mathcal Q(\textcolor{Green}{P})$ in \cref{pinm_HASS} and $\mathcal Q(R_6)$ in \cref{R6_example}, left, so we can draw the HASS $E_2$-page in \cref{R6_example_2}, left. For degree reasons, there is only one possible nonzero differential in this range, $d_2\colon E_2^{0,4}\to E_2^{2,5}$. Moreover, by inspecting the $E_2$-page, the value of $\mathcal N_3$ claimed in the lemma statement is equivalent to the claim that the differential in question is nonzero.

\begin{figure}[H]
\begin{subfigure}[c]{0.3\textwidth}
\begin{sseqdata}[name=QR6use2, classes = fill, scale=0.6, xrange={0}{5}, yrange={0}{4}, Adams grading, >=stealth,
x tick step = 2,
y tick step = 2,
x tick gap = 0.25cm,
y tick gap = 0.25cm,
x axis tail = 0.2cm,
y axis tail = 0.2cm,
y axis clip padding = 0cm,
x axis clip padding = 0cm,
class labels = {below = 0.07em, font=\tiny },
]
\begin{scope}[Green]
    \class(1, 0)
    \class(2, 1)\structline
    \class(3, 2)\structline
    \class(3, 1)\structline
    \class(3, 0)\structline
\end{scope}
\begin{scope}[draw=none, fill=none]
    \class(3, 1)\class(3, 2)
\end{scope}
\begin{scope}[MidnightBlue]
    \class(2, 0)
    \class["\phi"](4, 0)
    \class(4, 1)\structline
\end{scope}
\begin{scope}[Fuchsia]
    \class(3, 0)
    \class(5, 0)
    \class(5, 1)\structline
\end{scope}
\d[gray]2(4, 0)
\end{sseqdata}
\printpage[name=QR6use2, page=2]
\end{subfigure}
\begin{subfigure}[c]{0.3\textwidth}
    \printpage[name=QR6use2, page=3]
\end{subfigure}
\caption{Left: the $E_2$-page of the HASS computing $(\tau_{\le 2}\ko)_*((B\Z/2)^{\sigma-1}\wedge (B\Z/2^k)^{V-2})$. We use this spectral sequence in the proof of \cref{cNSH}, where we show that the pictured $d_2$ is nonzero. Right: the $E_3 = E_\infty$-page.}
\label{R6_example_2}
\end{figure}

Looking at \cref{R6_example_2}, left, the source of this differential, $E_2^{0,4}$, is isomorphic to $\Z/2$. Let $\phi\in E_2^{0,4}$ be the nonzero element. If $M$ is an $\ell$-connected $\cA(1)$-module (i.e.\ it vanishes in degrees $\ell$ and below), exactness of~\eqref{cQLES} implies the map $t\colon \Ext_{\cA(1)}(M)\to\mathcal Q(M)$ is an isomorphism for $t-s\le 4+\ell$; therefore, since $\widetilde H^*((B\Z/2)^{\sigma-1}\wedge B\Z/2^k; \Z/2)$ is $0$-connected, all classes in topological degree $\le 4$ are in the image of $t$. This includes $\phi$ and all possible values of $d_2(\phi)$, so if $\widetilde\phi$ is the unique preimage of $\phi$ in $\Ext_{\cA(1)}^{0,4}$, then $d_r(\phi)\ne 0$ if and only if $d_r(\widetilde\phi)\ne 0$ for all $r\ge 2$. Thus, it suffices to show $\widetilde\phi$ does not survive to the $E_\infty$-page in the Adams spectral sequence for $\ko$-homology: since $\widetilde\phi$ is in filtration $0$, it cannot be in the image of a differential, and the only differential it could possibly support is a $d_2$, for degree reasons.

Since $\widetilde\phi$ is in filtration $0$, it corresponds uniquely to an $\cA(1)$-module homomorphism
\begin{equation}
\Phi\colon \widetilde H^*((B\Z/2)^{\sigma-1}\wedge B\Z/2^k; \Z/2) \longrightarrow \Z/2.
\end{equation}
Since $E_2^{0,4}\cong\Z/2$, there must be a unique nonzero such $\cA(1)$-module homomorphism, and a straightforward calculation shows that such a homomorphism is nonzero on $Ux_1^2 y$.

The behavior of filtration-$0$ classes in an Adams spectral sequence for bordism is standard: see~\cite[\S 8.4]{FH21}. In particular, the following are equivalent.
\begin{enumerate}
    \item $\widetilde\phi$ survives to the $E_\infty$-page.
    \item\label{phi2} There is a closed, $4$-dimensional $(B\Z/2\times B\Z/2^k, \sigma \boxplus V)$-twisted spin manifold (see Footnote~\ref{shfoot}) $N$ with $\int_N x_1^2y\ne 0$.
\end{enumerate}
Moreover, using the Whitney sum formula and the definition of an $(X, V)$-twisted spin structure, one can show that the notion of twisted spin structure appearing in item~\ref{phi2} above is the data of a \pinm structure and a principal $\Z/2^k$-bundle $P\to N$, and that $x = w_1(N)$. Thus $\int_N x_1^2y = \int_N w_1(N)^2 y(P)$. Since we want to show that $d_2(\phi)\ne 0$ to finish the proof of the lemma, it will therefore suffice to show that there is no closed \pinm $4$-manifold $N$ with principal $\Z/2^k$-bundle $P\to N$ with $\int_N w_1(N)^2y(P)\ne 0$.

Now consider the Smith homomorphism from \cref{LES_is_Smith} associated to the data $\xi = \Spin$, $X = B\Z/2\times B\Z/2^k$, $V = 0$, and $W = \sigma$. Because the sphere bundle of $\sigma\to B\Z/2$ is contractible, the long exact sequence in \cref{Smith_LES} simplifies to an isomorphism
\begin{equation}
\label{smith_iso_GOPWW}
    \mathrm{sm}_\sigma\colon \widetilde \Omega_k^\Spin(B\Z/2\wedge B\Z/2^k) \longrightarrow \widetilde\Omega_{k-1}^{\Pin^-}(B\Z/2^k).
\end{equation}
Thus this map is called a \term{Smith isomorphism}. It is a special case of a general family of Smith isomorphisms discussed in~\cite[\S 7.1]{debray2024smith}; other examples in this family include the Smith isomorphisms discussed in~\cite{CF64, ABP69, Sto69, Uch70, Kom72, BG87a, BG87b, Guo:2018vij, COSY20, Hason:2020yqf, debray2025smith}. The example in~\eqref{smith_iso_GOPWW} was first studied in~\cite[\S 7.2.2]{Guo:2018vij}.

Recall that we have reduced the proof of the lemma to the assertion that there is no closed \pinm $4$-manifold $N$ and principal $\Z/2^k$-bundle $P\to N$ such that $\int_N w_1(N)^2y(P)\ne 0$. We can pull this back across~\eqref{smith_iso_GOPWW}: it suffices to show that there is no closed, spin $5$-manifold $W$ with principal $\Z/2$-bundle $Q_1\to W$ and principal $\Z/2^k$-bundle $Q_2\to W$ such that $\int_{\mathrm{sm}_\sigma(W)} w_1^2 y\ne 0$. By \cref{LES_is_Smith}, any smooth submanifold representative of the Poincaré dual of $x(Q_1)$ (i.e.\ the Euler class of the line bundle associated to $Q_1$) represents the bordism class $\mathrm{sm}_\sigma(W)$. That is, we want to show that for all $(W, Q_1, Q_2)$ as above,
\begin{equation}
\label{PD_int}
    \int_{\mathrm{PD}(x(Q_1))} w_1(\mathrm{PD}(x_1(Q_1))^2 \cdot y(Q_2|_{\mathrm{PD}(x_1(Q_1))}) = 0.
\end{equation}
where $\mathrm{PD}$ means any choice of submanifold representative of the Poincaré dual; the integral does not depend on this choice.

It is standard that if $i\colon N\hookrightarrow M$ is a smooth representative of the Poincaré dual of the Euler class $e(E)$ of a vector bundle $E\to M$, then the normal bundle $\nu$ of $N\subset M$ is isomorphic to $E|_N$, and that if $z\in H^*(M;\Z/2)$, then
\begin{equation}
\label{integral_transform}
    \int_N i^*(z) = \int_M e(E) i^*(z).
\end{equation}
In the situation at hand, $M$ is oriented, so $w_1(TN) = w_1(\nu)$ by the Whitney sum formula. Thus, applying~\eqref{integral_transform} to~\eqref{PD_int}, we obtain
\begin{equation}
    \int_{\mathrm{PD}(x_1(Q_1))} w_1(\mathrm{PD}(x_1(Q_1))^2 \cdot y(Q_2|_{\mathrm{PD}(x_1(Q_1))}) = \int_W x_1(Q_1)^3 y(Q_2).
\end{equation}
To finish the proof of the lemma, we will show this vanishes. Since $W$ is a closed, oriented $5$-manifold,
the Wu formula implies
\begin{equation}
    \int_W x_1(Q_1)^3 y(Q_2) = \int_W \Sq^1( x_1(Q_1)^2 y(Q_2)) = \int_W w_1(W) x_1(Q_1)^2 y(Q_2) = 0.
\end{equation}
We draw the $E_\infty$-page of this HASS in \cref{R6_example_2}, right.
\end{proof}

\begin{rem}
There is a potential hidden extension by $2$ in degree $2$ which our proof does not address; this is why $\mathcal N_2$ is left ambiguous in the statement of \cref{cNSH}. It is possible to show that this extension splits, so that $\mathcal N_2\cong\Z/8\oplus (\Z/2)^{\oplus 2}$. For $k = 2$ this follows from~\cite[Theorem 17]{Guo:2018vij}. One way to prove that the extension splits for all $k$ is to pull back across $\tau_{\le 2}$ and answer the equivalent question in $\ko$-homology, using that the multiplication-by-$2$ map factors as
\begin{equation}
    \ko_n(X) \overset{c}{\longrightarrow}
    \ku_n(X) \overset{b}{\longrightarrow}
    \ku_{n+2}(X) \overset{R}{\longrightarrow}
    \ko_n(X),
\end{equation}
where $c$ is complexification, $b$ is the complex Bott periodicity map, and $R$ is obtained from the realification map (see~\cite[Theorem 1]{Bru12}). By studying the effects of $c$, $b$, and $R$ on the corresponding Adams spectral sequences, one can show that their composition must vanish, so that $\widetilde{\mathcal N}_2$ contains no elements of order $4$.
\end{rem}
Using \cref{pinm_SH,,cMSH,,cNSH}, we write down the Smith long exact sequence~\eqref{eq:firstsmith} in low degrees in \cref{fig:LESinPin-}. Some of the maps are determined up to isomorphism by exactness; we also depict those in \cref{fig:LESinPin-}. These maps are calculated starting in degrees $0$ and $1$, and then propagating that information upwards in order to degrees $2$, $3$, and $4$ using exactness of the sequence.

\begin{figure}[ht]
 \hspace*{-2cm} 
{\scriptsize \begin{tikzcd}
	{k} & {(\tau_{\le 2}\ko)_n(M) \oplus
(\tau_{\le 2}\ko)_{n-1}(M)} & {(\tau_{\le 2}\ko)_n(M\wedge (B\Z/2^k)^{V-2})} & {(\tau_{\le 2}\ko)_{n-2}(M\wedge (B\Z/2^k)_+)} \\
    0 & \Z/2\arrow["\cong", r] & {\Z/2} & 0 \\
	1 & {\Z/2\oplus  \Z/2}
        \arrow["{[0\ 1]}", r, ->>] & \Z/2 & 0 \\
	2 & {\Z/8 \oplus \Z/2} \arrow["{[1\ 0]}", r, ->>]
    & \Z/4\arrow["0", r] & {\Z/2} \\
	3 & {\Z/8}\ar["1", r, ->>] & \Z/4\ar[r, "0"] & {(\Z/2)^{\oplus2}} \\
	4 & {0}
    & {\mathcal M_4}\ar[r, "\phi_2", hookrightarrow] & {\mathcal N_2} \\
	5 & {\Z/2}\ar[r, "\phi_3"] & {\mathcal M_5}\ar[r, "\phi_4", ->>] & {\Z/4\oplus\Z/2}\\
	\arrow["{\begin{bsmallmatrix}1\\0\end{bsmallmatrix}}"{pos=0.92}, from=4-4, to=3-2, in=-170, out=10, hookrightarrow]
    	\arrow["{\begin{bsmallmatrix}4 \, 0\\0\,1\end{bsmallmatrix}}"{pos=0.92}, from=5-4, to=4-2, in=-170, out=10, hookrightarrow]
        \arrow["\phi_1"{pos=0.92}, from=6-4, to=5-2, in=-170, out=10]
\end{tikzcd}}
    \caption{The long exact sequence~\eqref{eq:firstsmith}. We calculated the $\tau_{\le 2}\ko$-homology groups appearing in this sequence in \cref{pinm_SH,,cMSH,,cNSH}; $\mathcal N_2$, $\mathcal M_4$, and $\mathcal M_5$ were not completely determined by those lemmas. We use this long exact sequence in the proof of \cref{d265}.}
    \label{fig:LESinPin-}
\end{figure}

Since $\mathrm{Im}(\phi_1) = \ker(1\colon\Z/8\to\Z/4) = 4\Z/8\cong\Z/2$, we obtain a short exact sequence
\begin{equation}
\label{1st_smith_exact}
    \shortexact[\phi_2][\phi_1]{\mathcal M_4}{\mathcal N_2}{\mathrm{Im}(\phi_1)\cong \Z/2}.
\end{equation}
Recall from \cref{cMSH} that $\mathcal M_4$ is isomorphic to either $\Z/8\oplus\Z/2$ or $\Z/4\oplus\Z/2$, and from \cref{cNSH} that $\mathcal N_2$ is isomorphic to either $\Z/8\oplus (\Z/2)^{\oplus 2}$ or $\Z/8\oplus \Z/4$. Of the four possible options, only the two with $\mathcal M_4\cong\Z/8\oplus\Z/2$ are compatible with exactness of~\eqref{1st_smith_exact}. \Cref{cMSH} then tells us that $d_2(b) = d_2(c) = 0$ and that $\mathcal M_5$ is isomorphic to one of $\Z/8\oplus (\Z/2)^{\oplus 2}$ or $\Z/4\oplus (\Z/2)^{\oplus 2}$. In particular, $N\coloneqq |\mathcal M_5|$ is either $16$ or $32$. Since $\phi_4$ is surjective, $\mathrm{Im}(\phi_4)$ has order $8$, so $\ker(\phi_4) = \mathrm{Im}(\phi_3)$ has order $N/8$. Since the domain of $\phi_3$ is $\Z/2$, $\mathrm{Im}(\phi_3)$ has order at most $2$, so $N/8\le 2$, or $N\le 16$, implying $\mathcal M_4\cong \Z/4\oplus (\Z/2)^{\oplus 2}$.
\end{proof}

Dualizing the results of this lemma, we get twisted supercohomology groups:
\begin{itemize}
    \item $\SH^4(B\Z/2\times B\Z/2^k, x_1, y)\cong\Z/8\oplus\Z/2$.
    \item $\SH^5(B\Z/2\times B\Z/2^k, x_1, y) = \Z/2\oplus \Z/2 \oplus \Z/4$.
\end{itemize}

Thus the $d_3\colon E^{2,2}_3 \rightarrow E^{5,0}_3$ in the AHSS of \Cref{eq:E3timerev} vanishes. Consulting the $E_\infty$-page of the same AHSS that computes $\SH^5(B\Z/2\times B\Z/2^k, x_1, y)$, we see that there must be a generator of $\Z/2$ that resides in the DW layer. The class $e$ in \Cref{R6_example} does not appear in the analogous twisted spin cobordism computation, and so this $\Z/2$ generator must be in the Majorana layer as that is the only layer that can differ between supercohomology and spin cobordism. Therefore, the generator for $\Z/4$ must be in the Gu--Wen layer. This establishes \cref{prop:22kxy}.
\end{proof}
\begin{rem}
Essentially the same argument, just using $\Ext$ instead of $\mathcal Q$, can be used to show $\Omega_5^\Spin(B\Z/2\times B\Z/2^k, x_1, y)\cong\textcolor{RedOrange}{\Z/4} \oplus \textcolor{Green}{\Z/2}$.
\end{rem}

\subsubsection{Consequences for symmetry extension}
\label{subsubsec:cons}
In this example the symmetry algebra not only includes fermion parity and a $\Z/2^{k+1}$ unitary symmetry in which the generator $g$ satisfies $g^k = (-1)^F$, but also a $\Z/2$ time-reversal symmetry, which reverses the orientation of the background manifold. This corresponds to a $G$-structure for the group $G = \Pin^+\times_{\set{\pm 1}}\Z/2^{k+1}$.\footnote{This structure is equivalent to $\Pin^-\times_{\set{\pm 1}}\Z/2^{k+1}$ via an automorphism of $\Z/2\times\Z/2^{k+1}$, analogously to how $\Pin^c$ is isomorphic to both $\Pin^+\times_{\set{\pm 1}}\U_1$ and $\Pin^-\times_{\set{\pm 1}}\U_1$. Thus, depending on one's choice of generator $T$ for the time-reversal symmetry, one could have $T^2 = 1$ or $T^2 = (-1)^F$.} When $k = 1$, this is the ``pin-$\Z/4$ structure'' studied by Montero--Vafa~\cite{MV21} and Krulewski--Stehouwer~\cite{KS}; in general, this structure is analogous to a pin\textsuperscript{$c$} structure, with $\U_1$ replaced by $\Z/2^{k+1}$. Thus, analogously to how a pin\textsuperscript{$c$} structure is equivalent to a $(B\Z/2\times B\U_1, x_1, c_1)$-twisted spin structure~\cite[\S 10]{FH}, where $x_1\in H^1(B\Z/2;\Z/2)$ and $c_1\in H^2(B\U_1;\Z/2)$ are the generators, $\Pin^+\times_{\set{\pm 1}}\Z/2^{k+1}$ structures are equivalent to $(B\Z/2\times B\Z/2^k, x_1, y)$-twisted spin structures.

For the rest of this subsection, assume $k>1$. By \Cref{prop:22kxy}, there is an isomorphism $\SH^5(B\Z/2 \times B\Z/2^k,x_1,y) \cong \Z/2 \oplus \Z/2 \oplus \Z/4$, and we may choose the isomorphism such that
\begin{itemize}
    \item the class $\alpha_{\mathrm{Maj}} \coloneqq (1, 0, 0)$ is in the Majorana layer,
    \item the class $\alpha_{\mathrm{DW}} \coloneqq(0, 1, 0)$ is in the Dijkgraaf--Witten layer, and
    \item the class $\alpha_{\mathrm{GW}} \coloneqq(0, 0, 1)$ is in the Gu--Wen layer.
\end{itemize}
Moreover, it follows from \cref{long_HASS_analysis}, part~\eqref{cokernel} that $\alpha_{\mathrm{Maj}}$ generates the kernel of the map to $\mho_\Spin^5(B\Z/2\times B\Z/2^{k}, x_1, y)$, so we will focus on $\alpha_{\mathrm{DW}}$ and $\alpha_{\mathrm{GW}}$.

\begin{proof}[Proof of \Cref{thm:timerev}]
Since we do not know which of $(-1)^{x_1^4x}$, $(-1)^{xy^2}$ corresponds to $\alpha_{\mathrm{DW}}$, we will trivialize all of the classes on the $E_\infty$-page that could correspond to $\alpha_{\mathrm{DW}}$ and $\alpha_{\mathrm{GW}}$ by pulling back to $\textcolor{BrickRed}{B\Z/8}\times \textcolor{MidnightBlue}{B\Z/2^{k+1}}$. These classes are $x_1^3x\in E_\infty^{4,1}$, $(-1)^{x_1^4x}\in E_\infty^{5,0}$, $(-1)^{xy^2}\in E_\infty^{5,0}$, and linear combinations of them. Thus it suffices to trivialize these three classes.

To trivialize ${x_1^3 x}$ and $(-1)^{x_1^4x}$, first pull back to $\SH^5(B\Z/4\times B\Z/2^k,x_1,y)$, so $x_1^2\mapsto 0$. This implies that for the Dijkgraaf--Witten layer, $(-1)^{x_1^4x}\mapsto 0$ as well, but it does not suffice to trivialize $\alpha_{\mathrm{GW}}$ (corresponding to $x_1^3x$) -- all we know is that it pulls back to some class in the Dijkgraaf--Witten layer.

Thus, to trivialize $\alpha_{\mathrm{GW}}$, we may pull back to $B\Z/4 \times B\Z/2^k$, then work in $x_1$-twisted $\C^\times$-cohomology. 
\begin{lem}
\label{2is0}
For $p,q\ge 2$, $2 = 0$ in $H^*(B\Z/2^p\times B\Z/2^q;\C^\times_{x_1})$.
\end{lem}
\begin{proof}
Use the long exact sequence associated to $0\to\Z\to\R\to\C^\times\to 0$ as usual to reduce to the analogous claim with $\Z_{x_1}$ coefficients.
The result then follows from the Künneth formula for twisted cohomology and the calculations of $H^*(B\Z/2^p;\Z)$ and $H^*(B\Z/2^p;\Z_x)$, which can be found in \cref{p_on_coh} and \cite[Lemma A.12]{Debray:2023iwf}, respectively.
\end{proof}
\begin{lem}
\label{mod_2_twisted}
Let $\alpha \in H^6(B\Z/4\times B\Z/2^k;\Z/2)$ be a class in the image of the twisted mod $2$ reduction map $\widetilde r_2\colon H^6(\bl;\Z_{x_1})\to H^6(\bl;\Z/2)$. Then the pullback of $\alpha$ to $H^6(\textcolor{BrickRed}{B\Z/8}\times \textcolor{MidnightBlue}{B\Z/2^{k+1}};\Z/2)$ vanishes.
\end{lem}
\begin{proof}
The set 
$\{y_1^3, y_1^2y_2, y_1y_2^2, y_2^3, x_1x_2y_1^2, x_1x_2 y_1y_2, x_1x_2y_2^2\}$ is a basis for $H^6(B\Z/4\times B\Z/2^k;\Z/2)$, where $x_1$ and $y_1$ come from $B\Z/4$ and $x_2$ and $y_2$ come from $B\Z/2^k$. Thus every class is either $y_1$ or $y_2$ times some degree-$4$ class. But $y_1$ and $y_2$ pull back to $0$ for $\textcolor{BrickRed}{\Z/8}\times \textcolor{MidnightBlue}{\Z/2^{k+1}}$, as follows from \cref{p_on_coh} after mod $2$ reduction, so $\alpha\mapsto 0$.
\end{proof}

By \cref{mod_2_twisted}, when we pull $\alpha_{\mathrm{GW}}$ back to $H^5(\textcolor{BrickRed}{B\Z/8}\times \textcolor{MidnightBlue}{B\Z/2^{k+1}};\C^\times_{x_1})$, its mod $2$ reduction vanishes, but by \cref{2is0}, this implies the pullback of $\alpha_{\mathrm{GW}}$ is $0$.

This leaves $(-1)^{xy^2}$. Pull back to $\SH^5(B\Z/2 \times \textcolor{MidnightBlue}{B\Z/2^{k+1}},x_1,0)$, in which $y$ trivializes. Thus if we pull back to $\textcolor{BrickRed}{B\Z/8}\times \textcolor{MidnightBlue}{B\Z/2^{k+1}}$, all three of these classes map to $0$.

\end{proof}

\bibliography{reference}
\bibliographystyle{alpha}

\end{document}